\documentclass[reprint,amsfonts, amssymb, amsmath,  showkeys, nofootinbib,pra, superscriptaddress, twocolumn,longbibliography]{revtex4-2}

\usepackage{float}
\makeatletter
\let\newfloat\newfloat@ltx
\makeatother
\usepackage[english]{babel}
\usepackage[utf8]{inputenc}
\usepackage{graphics}
\usepackage{selinput}

\usepackage{braket}
\usepackage{amsthm}
\usepackage{mathtools}
\usepackage{physics}
\usepackage{xcolor}
\usepackage{graphicx}
\usepackage[left=16mm,right=16mm,top=35mm,columnsep=15pt]{geometry} 
\usepackage{adjustbox}
\usepackage{placeins}
\usepackage[T1]{fontenc}
\usepackage{lipsum}
\usepackage{csquotes}
\usepackage{tikz}

\usepackage{tcolorbox}

\newtcolorbox[auto counter]{pabox}[2][]{fonttitle=\bfseries,
title=Example~\thetcbcounter: #2,#1,colframe=gray}

\usepackage{braket}
\usepackage{algorithm}
\usepackage[noend]{algpseudocode}

\def\HC{\mathcal{H}}
\def\KC{\mathcal{K}}
\def\LC{\mathcal{L}}

\newcommand{\ie}{\textit{i.e.}}

\def\ad{^{\dagger}}

\usepackage[makeroom]{cancel}
\usepackage[toc,page]{appendix}
\usepackage[colorlinks=true,citecolor=blue,linkcolor=magenta,linktocpage=true]{hyperref}

\newcommand{\dya}[1]{\ket{#1}\!\bra{#1}}

\newcommand{\BC}{\mathcal{B}}
\newcommand{\CC}{\mathcal{C}}

\newcommand{\FC}{\mathcal{F}}

\newcommand{\NC}{\mathcal{N}}
\newcommand{\OC}{\mathcal{O}}

\newcommand{\RC}{\mathcal{R}}

\newcommand{\TC}{\mathcal{T}}

\newcommand{\YC}{\mathcal{Y}}

 \newcommand{\Adj}{{\rm Ad}}
  \newcommand{\adj}{{\rm ad}}

\renewcommand{\geq}{\geqslant}
\renewcommand{\leq}{\leqslant}

\renewcommand{\Re}{\text{Re}}
\renewcommand{\Im}{\text{Im}}

\newcommand{\CP}{\mathcal{C}\mathcal{P}}

\newcommand{\SWAP}{\mathrm{SWAP}}

\renewcommand{\vec}[1]{\boldsymbol{#1}}  

\newcommand*{\id}{\openone}

\newcommand{\bs}{\textsf{BS}}

\newcommand{\lm}{\lambda }

\newcommand{\sg}{\sigma }

\newcommand{\thv}{\vec{\theta}}

\def\be{\begin{equation}}
\def\ee{\end{equation}}
\def\bs{\begin{split}}
\def\e{\end{split}}
\def\ba{\begin{eqnarray}}
\def\bea{\begin{eqnarray}}

\def\tea{\end{eqnarray}}
\def\ea{\end{eqnarray}}
\def\eea{\end{eqnarray}}

\def\tn{^{\otimes n}}
\def\tn{^{\otimes n}}

\newcommand\mf[1]{\mathfrak{#1}}
\newcommand\mbb[1]{\mathbb{#1}}

\newcommand\spn{\text{span}}

\newtheorem{theorem}{Theorem}

\newtheorem{lemma}{Lemma}
\newtheorem{proposition}{Proposition}

\newtheorem{definition}{Definition}

\usepackage{amssymb}
\usepackage{dsfont}

\def\be{\begin{equation}}
\def\te{\end{equation}}
\def\ee{\end{equation}}
\def\ba{\begin{eqnarray}}
\def\bea{\begin{eqnarray}}

\def\tea{\end{eqnarray}}
\def\ea{\end{eqnarray}}
\def\eea{\end{eqnarray}}

\definecolor{myblue}{RGB}{0,163,243}
\definecolor{myred}{RGB}{255,100,100}
\definecolor{mygreen}{RGB}{0,153,0}
\definecolor{mypurple1}{RGB}{46, 156, 202}
\definecolor{mypurple2}{RGB}{190,90,175}
\definecolor{fundamental}{RGB}{55, 110, 111}
\definecolor{tensor}{RGB}{161, 195, 209}

\newtcolorbox[use counter from=pabox]{red_boxed_example}[2][]{colback=myred!5!white,colframe=myred!75!black,fonttitle=\bfseries,floatplacement=h!t,float,title=Example~\thetcbcounter: #2,#1}
\newtcolorbox[use counter from=pabox]{green_boxed_example}[2][]{colback=mygreen!5!white,colframe=mygreen!75!black,fonttitle=\bfseries,floatplacement=h!t,float,title=Example~\thetcbcounter: #2,#1}

\newtcolorbox[use counter from=pabox]{purple_boxed_example1}[2][]{colback=mypurple1!5!white,colframe=mypurple1!75!black,fonttitle=\bfseries,floatplacement=h!t,float,
title=Example~\thetcbcounter: #2,#1}
\newtcolorbox[use counter from=pabox]{purple_boxed_example2}[2][]{colback=mypurple2!5!white,colframe=mypurple2!75!black,fonttitle=\bfseries,floatplacement=h!t,float,title=Example~\thetcbcounter: #2,#1}

\newtcolorbox[use counter from=pabox]{purple_boxed_example}[2][]{%
colback=mypurple2!5!white,colframe=mypurple2!75!black,fonttitle=\bfseries,floatplacement=h!t,float,
title=Example~\thetcbcounter: #2,#1}

\newtcolorbox[use counter from=pabox]{blue_boxed_example}[2][]{%
colback=myblue!5!white,colframe=myblue!75!black,fonttitle=\bfseries,floatplacement=h!t,float,
title=Example~\thetcbcounter: #2,#1}

\newtcolorbox[use counter from=pabox]{fundamental_boxed_example}[2][]{%
colback=fundamental!5!white,colframe=fundamental!75!black,fonttitle=\bfseries,floatplacement=h!t,float,
title=Example~\thetcbcounter: #2,#1}

\newtcolorbox[use counter from=pabox]{tensor_boxed_example}[2][]{%
colback=tensor!5!white,colframe=tensor!75!black,fonttitle=\bfseries,floatplacement=h!t,float,
title=Example~\thetcbcounter: #2,#1}

\usepackage[normalem]{ulem}

\usepackage{tikz}
\usetikzlibrary{quantikz}
\usepackage{braket, enumitem}

\begin{document}
\title{Theory for Equivariant Quantum Neural Networks}

\author{Quynh T. Nguyen}
\affiliation{Theoretical Division, Los Alamos National Laboratory, Los Alamos, New Mexico 87545, USA}
\affiliation{School of Engineering and Applied Sciences, Harvard University,
Cambridge, Massachusetts 02138, USA}

\author{Louis Schatzki}
\affiliation{Information Sciences, Los Alamos National Laboratory, Los Alamos, New Mexico 87545, USA}
\affiliation{Department of Electrical and Computer Engineering, University of Illinois at Urbana-Champaign, Urbana, Illinois 61801, USA}

\author{Paolo Braccia}
\affiliation{Theoretical Division, Los Alamos National Laboratory, Los Alamos, New Mexico 87545, USA}
\affiliation{Dipartimento di Fisica e Astronomia, Universit\`{a} di Firenze, Sesto Fiorentino (FI), 50019 , Italy}

\author{Michael Ragone}
\affiliation{Theoretical Division, Los Alamos National Laboratory, Los Alamos, New Mexico 87545, USA}
\affiliation{Department of Mathematics, University of California Davis, Davis, California 95616, USA}

\author{Patrick J. Coles}
\affiliation{Theoretical Division, Los Alamos National Laboratory, Los Alamos, New Mexico 87545, USA}

\author{Fr\'{e}d\'{e}ric Sauvage}
\affiliation{Theoretical Division, Los Alamos National Laboratory, Los Alamos, New Mexico 87545, USA}

\author{Mart\'{i}n Larocca}
\affiliation{Theoretical Division, Los Alamos National Laboratory, Los Alamos, New Mexico 87545, USA}
\affiliation{Center for Nonlinear Studies, Los Alamos National Laboratory, Los Alamos, New Mexico 87545, USA}

\author{M. Cerezo}
\affiliation{Information Sciences, Los Alamos National Laboratory, Los Alamos, New Mexico 87545, USA}

\begin{abstract}
Quantum neural network architectures that have little-to-no inductive biases are known to face trainability and generalization issues. Inspired by a similar problem, recent breakthroughs in machine learning address this challenge by creating models encoding the symmetries of the learning task. This is materialized through the usage of equivariant neural networks whose action commutes with that of the symmetry. In this work, we import these ideas to the quantum realm by presenting a comprehensive theoretical framework to design equivariant quantum neural networks (EQNN) for essentially any relevant symmetry group. We develop multiple methods to construct equivariant layers for EQNNs and analyze their advantages and drawbacks. Our methods can find unitary or general equivariant quantum channels efficiently even when the symmetry group is exponentially large or continuous. As a special implementation, we show how standard quantum convolutional neural networks (QCNN) can be generalized to group-equivariant QCNNs where both the convolution and pooling layers are equivariant to the symmetry group. We then numerically demonstrate the effectiveness of a $\mbb{SU}(2)$-equivariant QCNN over symmetry-agnostic QCNN on a classification task of phases of matter in the bond-alternating Heisenberg model. Our framework can be readily applied to virtually all areas of quantum machine learning. Lastly, we discuss about how symmetry-informed models such as EQNNs provide hopes to alleviate central challenges such as barren plateaus, poor local minima, and sample complexity.
\end{abstract}

\maketitle

\section{Introduction}

Recognizing the underlying symmetries in a given dataset has played a fundamental role in classical machine learning. For instance, noting that the picture of a cat still depicts a cat when we translate the pixels of the image, gives a hint as to why convolutional neural networks~\cite{lecun1998gradient} have been so successful in image classification: they process images in a translationally-symmetric way~\cite{bronstein2021geometric}.

In recent years, the importance of symmetries in machine learning has been studied in problems with more general symmetry groups than translations, leading to the burgeoning field of geometric deep learning~\cite{bronstein2021geometric}. The central thesis of this field is that prior symmetry knowledge should be incorporated into the model, thus effectively constraining the search space and easing the learning task. Indeed, symmetry-respecting models have been observed to perform and generalize better than problem-agnostic ones in a wide variety of tasks~\cite{bronstein2021geometric, bekkers2018roto, jegelka2022theory, kondor2018clebsch, cohen2018spherical, bogatskiy2020lorentz,finzi2020generalizing, elesedy2021provably}. As such, a great deal of work has also gone into developing a mathematically rigorous framework for designing symmetry-informed models through the machinery of representation theory. This has provided the basis for so-called \emph{equivariant neural networks} (ENNs)~\cite{cohen2016group, kondor2018generalization, cohen2019general, cohen2019gauge, cohen2021equivariant}, whose key property is that their action commutes with that of the symmetry group. In other words, applying a symmetry transformation to the input and then sending it through the ENN produces the same result as sending the raw input through the ENN and then applying the transformation. 

Recently, some of the ideas of geometric deep learning have been imported to the field of \textit{quantum machine learning} (QML)~\cite{biamonte2017quantum,schuld2014quest,schuld2015introduction, cerezo2020variationalreview,cerezo2022challenges}.  QML has become a rapidly growing framework to make practical use of noisy intermediate-scale quantum devices~\cite{preskill2018quantum}. Here, the hope is that by accessing the exponentially large Hilbert space,  quantum models can obtain a computational advantage over their classical counterparts~\cite{huang2021quantum, liu2021rigorous}, especially for quantum data~\cite{cong2019quantum,schatzki2021entangled,caro2021generalization}.  Despite its promise, there are still several challenges one needs to address before unlocking the full potential of QML. In particular, a growing amount of evidence suggests that models with little-to-no inductive biases have poor trainability and generalization, greatly limiting their scalability~\cite{kubler2021inductive,mcclean2018barren,cerezo2020cost,sharma2020trainability,holmes2020barren,holmes2021connecting,cerezo2020impact,arrasmith2020effect,arrasmith2021equivalence,marrero2020entanglement,patti2020entanglement,uvarov2020barren,thanasilp2021subtleties,larocca2021diagnosing,wang2020noise,wiersema2020exploring}.

\begin{figure*}[t]
    \centering
    \includegraphics[width=\linewidth]{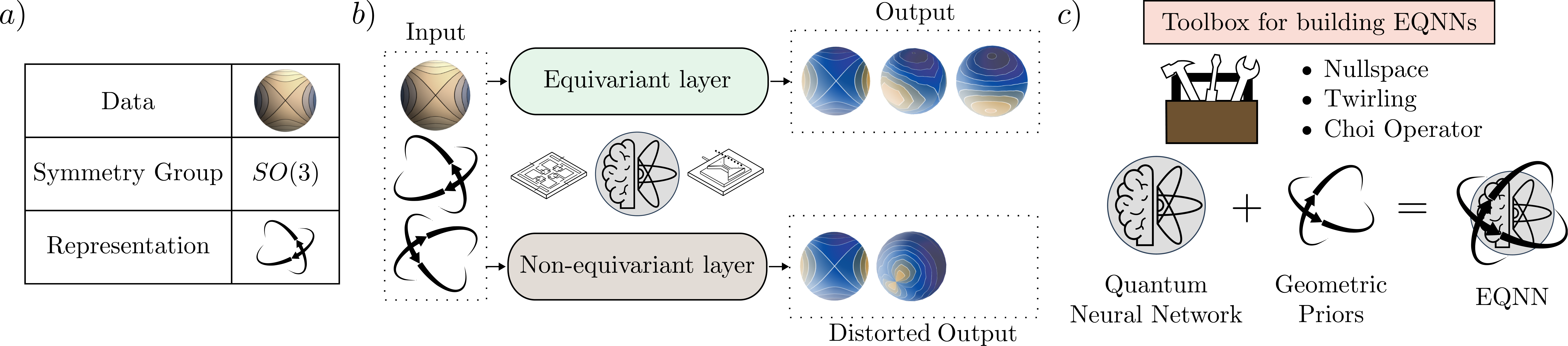}
    \caption{\textbf{Schematic representation of our main results.} a) In  GQML we start by identifying the symmetry group -or groups-- that leave the data labels invariant. For the example shown, the data can be visualized on a three-dimensional sphere, and the labels are invariant under the action of $\mathbb{SO}(3)$. b) Both in classical and quantum machine learning it has been shown that models with equivariant layers often have an improved performance over non-equivariant architectures. The key feature of equivariance is that  applying a rotation to the input data and sending it through the layer is the same as first sending the data through the layer and then rotating the output. On the other hand, feeding either a raw or a rotated data instance into a non-equivariant layer usually leads to distorted outputs which are not related by a rotation. c) In this work we provide a toolbox of methods for creating equivariant quantum neural networks (EQNNs) that can be readily used to construct quantum architectures with strong geometric priors. }
    \label{fig:schematic}
\end{figure*}

\textit{Geometric quantum machine learning} (GQML) attempts to solve the aforementioned issues by leveraging ideas from geometric deep learning to construct quantum models with sharp inductive biases based on the symmetries of the problem at hand. For instance, when classifying between states presenting a large, or a low, amount of multipartite entanglement~\cite{horodecki2009quantum,szalay2015multipartite,schatzki2021entangled}, it is natural to employ models whose outputs remain invariant under the action of any local unitary~\cite{larocca2022group}. While recent proposals have started to reveal the power of GQML~\cite{larocca2022group,meyer2022exploiting, mernyei2021equivariant, skolik2022equivariant,zheng2021speeding,zheng2022super,zheng2022benchmarking,li2022group,verdon2019quantumgraph}, the field is still in its infancy, and a more systematic approach to symmetry-encoded model design is needed.

The goal of this work  is to offer a theoretical framework for building GQML models based on extending the notion of classical ENNs to \emph{equivariant quantum neural networks} (EQNNs) (see Fig.~\ref{fig:schematic}). Our main contributions can be summarized as follows:
\begin{itemize}
    \item We provide an interpretation for EQNN layers as a form of generalized Fourier space action, meaning that they: perform a group Fourier transform, act on the Fourier components, transform back. This allows us to quantify the number of free parameters in an EQNN layer, and unravels the exciting possibility of using different group representations as hyperparameters to act on different generalized Fourier spaces. (Sec.~\ref{sec:connecting})
    \item We introduce a general framework for EQNNs, extending previous results from unitary quantum neural networks to channels. We characterize the different types of EQNN layers such as standard, embedding, pooling, lifting and projection layers. This permits a  classification of EQNN layers depending on their input and output representations. In addition, we also explore how equivariant non-linearities can be introduced via multiple copies of the data. (Sec.~\ref{sec:connecting}) 
    \item We describe three alternative methods for constructing and parametrizing EQNNs. These are based on finding the \textit{nullspace} of certain system of matrix equations, applying the \textit{twirling} formula over the symmetry group, and using the \textit{Choi operator} representation of channels.  Our methods have a better complexity than existing methods and can efficiently find unitary or non-unitary equivariant layers even when the symmetry group is exponentially large. We discuss strengths and weaknesses of each approach and general methods to optimize/train equivariant channels. (Sec.~\ref{sec:framework})
    \item We exemplify our techniques by showing how to generalize standard quantum convolutional neural networks (QCNN) to group-equivariant QCNNs where the convolutional and pooling layers are equivariant to the task's symmetry group.  For examples, we show how to construct  $\mathbb{Z}_2\times \mathbb{Z}_2$, and $\mathbb{Z}_n$-equivariant EQNNs. (Sec. ~\ref{sec:examples}) Moreover, we present a new architecture called $\mbb{SU}(2)$-equivariant QCNN and numerically demonstrate its advantage over symmetry-agnostic QCNN in a quantum phase classification task of the bond-alternating Heisenberg model on up to 13 qubits. (Sec. ~\ref{sec:numerics}) 
\end{itemize}

We conclude with a discussion on how EQNNs provide hope to alleviate critical challenges in QML such as ill-behaved training landscapes (barren plateaus and local minima), and to reduce the sample complexity (data requirements) of the model. Taken together, we hope that our results will serve as blueprints and guidelines to a more representation-theoretic approach to QML.

\section{Related work}\label{sec:background}

\subsection{Equivariance in geometric deep learning}
In this section, we provide an overview of literature on classical equivariant neural networks (ENNs), leaving the formal treatment of equivariance to Section~\ref{sec:prelim}.

At a high level, equivariance is a mathematical property that preserves symmetries in features throughout a multi-layer ENN. One imposes equivariance onto ENN layers via tools from group representation theory, the workhorse behind geometric deep learning~\cite{bronstein2021geometric}. 
The most well-known equivariant architecture is the convolutional neural network (CNN)~\cite{lecun1998gradient}, ubiquitous in image and signal processing. The relevant symmetry group in CNNs is the translation group in the plane $\mbb{R}^2$, and one can show that their convolution and pooling layers are equivariant to this group~\cite{kondor2018generalization}. Ideas to generalize CNNs to other groups and data were first laid out in~\cite{cohen2016group} and further made mathematically rigorous in~\cite{kondor2018generalization, cohen2019general}. These works are concerned with the so-called homogeneous ENNs, which include as special cases spherical CNNs (where the relevant group is $\mathbb{SO}(3)$) for spherical images~\cite{cohen2018spherical}, and Euclidean neural networks (the Euclidean group $\mathbb{E}(n)=\mbb{R}^n\rtimes \mbb{O}(n)$ and its subgroups) for molecular data~\cite{anderson2019cormorant, geiger2022e3nn, thomas2018tensor, chen2021direct}. Notable non-homogeneous architectures include graph neural networks (the permutation group $S_n$)~\cite{jegelka2022theory, pan2022permutation, thiede2020general}. In addition, more advanced representation-theoretic treatments on non-homogeneous data \cite{cohen2021equivariant} have led to steerable CNNs~\cite{cohen2016steerable}, gauge-equivariant CNNs \cite{cohen2019gauge} on general manifolds. Moreover, it has been shown that equivariant layers can be constructed from either the real space or Fourier space perspectives~\cite{kondor2018clebsch}.
More recently, methods for designing equivariant layers have been studied in~\cite{finzi2020generalizing, finzi2021practical, brandstetter2022clifford,li2022group}. For a theoretical analysis of the improvements in training and generalization error arising from using ENNs we refer the reader to~ \cite{li2020why, elesedy2021provably, sannai2021improved, sokolic2017generalization}, while the expressibility and universality of ENNs have been studied in~\cite{maron2019universality, yarotsky2022universal, keriven2019universal, ravanbakhsh2020universal,villar2021scalars}.

\subsection{Equivariance in quantum information}

Equivariance\footnote{The alternative term \textit{covariance} is used in some quantum information theory literature. Here we will adopt the terminology \textit{equivariance} for cohesiveness with geometric deep learning literature.} has a long history in quantum information theory~\cite{holevo2005additivity, ritter2005quantum, hastings2009superadditivity, wilde2017converse, koenig2009strong, datta2016second, leditzky2018approaches, gschwendtner2021programmability, grinko2022linear, kong2022near, faist2020continuous, zhou2021new, dadriano2004extremal, hayashi2017group, krovi2019efficient, harrow2005applications, bacon2006efficient, marvian2022restrictions, marvian2022rotationally, hulse2021qudit, childs2010quantum}. As such, we will not attempt here to review its full impact on the field, but we will rather focus on several relevant works where equivariance has been studied in the context of quantum channels.

To begin, the set of all irreducible $\mbb{SU}(2)$-equivariant channels has been characterized in~\cite{al2014extreme}, with extensions to a wide class of finite groups presented in \cite{mozrzymas2017structure}. The work in~\cite{memarzadeh2022group} presents conditions to construct group-equivariant generalized-extreme channels. On the other hand, the history of equivariance in QML is much more recent. In~\cite{larocca2022group} and~\cite{meyer2022exploiting} the authors lay a theoretical groundwork for the integration of symmetries into QML.  However, prior work are either non-constructive~\cite{larocca2022group} or only work efficiently on restricted sets of problems and symmetries~\cite{meyer2022exploiting}. 
In particular, two main types of symmetries have been most extensively explored: the action of the local unitary group $\mbb{SU}(d)$ on each qudit in a correlated manner $U\tn$, and the action of the permutation group $S_n$ by permuting the qudits. It is well known fact in representation theory, called the Schur-Weyl duality, that these two group representations commute.

On the side of $\mbb{SU}(d)$-symmetry, Ref.~\cite{zheng2022super} has proposed a specific task -- approximating matrix elements of $S_n$ irreps evaluated on arbitrary group algebra elements-- together with a novel polynomial-time quantum algorithm, based on the combination of Quantum Schur Transform and Hamiltonian simulation, that potentially achieves a super-exponential quantum speedup given that best known classical algorithms require $\OC(n!n^2)$ time.
In turn,~\cite{zheng2021speeding} exploits the ideas in~\cite{zheng2022super} to derive an ansatz -- the $S_n$-Equivariant Convolutional Quantum Alternating Ansatze ($S_n$-CQA)-- that is universal for the subgroup of $\mbb{SU}(d)$-equivariant unitaries. As the name suggests, it is based on the qudit permutation action of $S_n$ on the quantum system which, via Schur-Weyl duality, linearly spans the subspace of $\mbb{SU}(d)$-equivariant operators. Interestingly, $S_n$-CQA can be shown to achieve universality in the subgroup of symmetric unitaries with only four-body interactions, something that is remarkable given the typical limitations of universality imposed by locality constraints~\cite{marvian2022restrictions,kazi2023universality}. $S_n$-CQA's performance and resource requirements are benchmarked in ~\cite{zheng2022benchmarking}.

On the $S_n$ symmetry side, ~\cite{schatzki2022theoretical} has shown that $S_n$-equivariant QNNs exhibit a wide range of favorable properties such as being immune to barren plateaus, efficiently reaching overparametrization, and being able to generalize well from few training points. Moreover, methods for constructing equivariant quantum circuits for graph problems were given in~\cite{verdon2019quantumgraph, mernyei2021equivariant, skolik2022equivariant,sauvage2022building}. 
We also note that, while not explicitly mentioned, some recent quantum algorithms can be analyzed from an equivariance point of view \cite{cincio2018learning, van2012measuring, cong2019quantum, huang2021quantum, huang2020predicting}. We discuss these  in Appendix~\ref{sec:perspective}.

\section{Preliminaries}\label{sec:prelim}

Here we give some of the necessary background in quantum machine learning and representation theory to tackle GQML. For a more comprehensive treatment of these topics, we refer the reader to standard textbooks in representation theory~\cite{simon1996representations, fulton1991representation, goodman2000representations} and geometric ML theory literature \cite{bronstein2021geometric, cohen2019general, cohen2021equivariant}. For a QML-oriented approach to group theory and representation theory, see~\cite{ragone2022representation}.

\subsection{From QML to GQML}

For simplicity and concreteness, in this paper we focus on quantum supervised learning with scalar labels. However, we remark that GQML is relevant in other contexts such as unsupervised learning~\cite{otterbach2017unsupervised,kerenidis2019q}, generative modeling~\cite{dallaire2018quantum,benedetti2019generative,kieferova2021quantum,romero2021variational} or reinforcement learning~\cite{saggio2021experimental,skolik2021quantum}. For instance the constructions presented here can be readily adapted to learning problems with non-scalar output (e.g., quantum generative models, where the output is a quantum state, or a probability distribution).

Suppose we are given some dataset composed of $M$ quantum states and scalar labels $\{\rho_i, y_i\}_{i=1}^M$, where $\rho_i \in\RC$ are quantum states from a data domain $\RC\subset \mathcal{B}(\mathcal{H})$, with  $\mathcal{B}(\mathcal{H})$ the set of bounded linear operators on $\mathcal{H}$. The labels come from a label domain $\mathcal{Y}$ and are obtained from a (potentially probabilistic) function:
\begin{equation}
    f: \RC \rightarrow \mathcal{Y}\qquad\text{(Underlying function)}\,.
\end{equation}
For example, in binary classification one has  $\mathcal{Y} = \{0,1\}$. This data may come from some physical quantum mechanical process (quantum data~\cite{schatzki2021entangled}) or may be classical information embedded into quantum states (classical data~\cite{havlivcek2019supervised}). Given the dataset, one then optimizes a learning model:
\begin{equation}
    h_{\thv}: \RC \rightarrow \mathcal{Y}\qquad\text{(Model)}\,,
\end{equation}
where $\thv$ are trainable parameters, with the intent of closely approximating the underlying function $f$.

In variational QML~\cite{cerezo2020variationalreview}, the states in the dataset are fed into a trainable quantum circuit, which is usually modelled by a sequence of parametrized unitary matrices. However, in this work we will consider more general operations -- parametrized quantum channels -- which we refer to as \textit{quantum neural networks} (QNNs). Respectively denoting the spaces of bounded linear operators in $\HC^\text{in}$ and $\HC^\text{out}$ as $\BC^{\text{in}}:=\mathcal{B}(\mathcal{H}^\text{in})$ and $\BC^{\text{out}}:=\mathcal{B}(\mathcal{H}^\text{out})$, a QNN is a parametrized completely positive and trace-preserving (CPTP) linear map $\mathcal{N}_{\thv}: \mathcal{B}^\text{in} \rightarrow \mathcal{B}^\text{out}$ . 

We can further decompose the QNN as a concatenation of channels, or \textit{layers}. We say that $\NC_{\thv}$ is an $L$-layered QNN if it can be expressed as $\mathcal{N}_{\thv} = \mathcal{N}^L_{\thv_L} \circ \cdots \circ \mathcal{N}^1_{\thv_1}$, where the $\NC^l_{\thv_l}$ (with $l=1,\ldots, L$) are parametrized CPTP channels such that $\thv=(\thv_1,\ldots,\thv_L)$.  From the previous, the $l$-th layer maps between operators acting on some Hilbert space $\HC^{l-1}$ to operators acting on some (potential different) Hilbert space $\HC^{l}$. That is, $\NC^l_{\thv_l}:\BC^{l-1}\rightarrow \BC^{l}$, where we have defined for simplicity of notation $\BC^l:=\BC(\HC^{l})$.

After applying the QNN to an input state $\rho$, one measures the resulting state with respect to a set of observables $\{O_j\}_j$ to obtain the expectation values $\{\Tr[\mathcal{N}_{\thv}(\rho) O_j]\}_j$. Finally, a classical post-processing step, $\CC$, maps these outcomes to a loss function
\begin{equation}\label{eq:predicted-label}
    \ell_{\thv}(\rho) = \CC( \{\Tr[\NC_{\thv}(\rho) O_j]\}_j).
\end{equation}
We quantify the performance of the model over the dataset via the so-called empirical loss
\begin{align}
    \widehat{\LC}_{\thv}(\{\rho_i,y_i\}_{i=1}^M) = \frac{1}{M}\sum_{i=1}^M \FC(\ell_{\thv}(\rho_i),y_i)\,,
\end{align}
defined in terms of some problem-dependent function $\FC$. Finally, employing a classical computer, one optimizes over the parameters $\boldsymbol{\theta}$ to minimize the empirical loss until certain convergence conditions are met. The optimal parameters, along with the loss function, are used to predict labels.

One of the most important aspects that make or break the QML scheme are its inductive biases, i.e., the assumptions about the problem that one embeds in the structure of the model. In our case, this amounts to an adequate choice of the parameterized layers $\NC_{\thv_l}^l$ forming the QNN and of the measurement operators $O_j$. In a nutshell, the inductive biases are responsible for the model exploring only a subset of all possible functions from $\RC$ to $\YC$. 
If these inductive biases are too general or not accurate, the model is expected to train poorly, while models with appropriate inductive biases can often benefit from an improved performance~\cite{kubler2021inductive,holmes2021connecting,schatzki2022theoretical}. GQML aims at providing a framework for incorporating prior geometrical information in the model with the hope of improving its trainability, data requirements, generalization, and overall performance. In particular, the main goal of GQML is to create models respecting the underlying symmetries of the domain over which they act. In the next sections we will briefly review how to use tools from representation theory to deal with symmetries, as well as recall basic concepts such as equivariance and invariance.

\subsection{Symmetry groups and representation theory}\label{sec:rep-theory}

The first step towards building a GQML model is identifying the set of relevant operations that the model needs to preserve. We say that a QML problem has symmetry with respect to a group $G$ if the labels are unchanged under the action of a representation of $G$ on the input states. 
\begin{definition}[Label symmetries and $G$-invariance]
Given a compact group $G$ and some unitary representation $R$ acting on quantum states $\rho$, we say the underlying function $f$ has a label symmetry if it is $G$-invariant, i.e., if
\begin{equation}
f(R(g) \rho R(g)\ad) = f(\rho),\ \forall g\in G\,.
\end{equation}
\end{definition}
As previously mentioned, the goal of GQML is  to build  models that respect the label symmetries of the data. That is, we want to build $G$-invariant models such that  $h_{\thv}(\rho) = h_{\thv}(R(g) \rho R(g)\ad)$, for any $g\in G$, and for all values of $\thv$.

To further understand how symmetry groups act, and how one can manipulate them, we recall here some basic concepts from representation theory. (See Ref.~\cite{ragone2022representation} for further background.) Namely, given a group $G$, its \textit{representation} describes its action on some vector space $\HC$, which we assume for simplicity to be a Hilbert space.
\begin{definition}[Representation]\label{def:representation}
A representation $(R,\HC)$ of a group $G$ on a vector space $\HC$ is a homomorphism  $R: G \rightarrow GL(\HC)$ from the group $G$ to the space of invertible linear operators on $\HC$, that preserves the group structure of $G$.
\end{definition}
Specifically, a group homomorphism $R$ satisfies
\begin{align}
    R(g_1)R(g_2) = R(g_1g_2) \quad \forall g_1,g_2 \in G\,.
\end{align}
This implies that, for all $g\in G$, the representation of its inverse is the inverse of its representation, $R(g^{-1})=R(g)^{-1}$, and the representation of the identity element $e$ is the identity operator on $\HC$, $R(e) = \id_{\dim(\HC)}$. Given a representation, it is relevant to define its commutant.
\begin{definition}[Commutant]
    Given a representation $R$ of $G$, we define the commutant of $R$ as the set of bounded linear operators on $\HC$ that commute with every element in $R$, i.e.,
    \begin{align}
        \mf{comm}(R) = \{H\in \mathcal{B}(\HC)\,|\, [H,R(g)]=0 \ \forall g \in G\}\,.
    \end{align}
    \label{def:comm}
\end{definition}

Consider the following remarks about representations:
\begin{itemize}
    \item A representation is \emph{faithful} if it maps distinct group elements to distinct elements in $\HC$. As an example of unfaithfulness, the \emph{trivial representation} maps all group elements to the identity in $\HC$.
    \item Two representations $R_1$ and $R_2$ are \textit{equivalent} if there exists a change of basis $W$ such that $V R_1(g) V\ad = R_2(g)$ for all $g \in G$, in which case we denote $R_1 \cong R_2$. 
    \item A \textit{subrepresentation} is a subspace $\KC \subset \HC$ that is invariant under the action of the representation, i.e., $R(g)\ket{w}\in \KC$ for all $g\in G$ and $\ket{w}\in \KC$. The group can then be represented through $R|_\KC$, the restriction of $R$ to the vector subspace $\KC$. A subrepresentation $\KC$ is non-trivial if $\KC\neq \{0\}$ (the zero vector) and $\KC\neq \HC$.
\end{itemize}

\begin{definition}(Irreps)
    A representation is said to be an irreducible representation (irrep) if it contains no non-trivial subrepresentations. 
\end{definition}
Irreps are the fundamental building blocks of representation theory. For any finite or compact group, the representations can be chosen to be unitary~\cite{hall2013lie}. Hence in the rest of this paper we will consider unitary representations on complex Hilbert spaces. In the case that the representation is finite-dimensional, we can go a step further and say that the vector space can be decomposed into a direct sum over irreducible subrepresentations. This leads to the so-called \textit{isotypic} decomposition
\begin{align}
    \HC \cong \bigoplus_{\lambda} \HC_{\lambda} \otimes \mathbb{C}^{m_\lambda}, \qquad  R(g) \cong \bigoplus_\lambda R_\lambda(g) \otimes \id_{m_\lambda},
    \label{eq:isotypic}
\end{align}
where $\cong$ indicates that there exists a global change of basis matrix $W$ that simultaneously block-diagonalizes the unitaries $R(g)$ for all $g\in G$. Here, $\lambda$ labels the irreps, $m_\lambda$ is the multiplicity of the irrep $R_\lambda$ and $R_\lambda(g) \in \mbb{C}^{d_\lm \times d_\lm}$. Note that $\sum_\lm d_\lm m_\lm = \dim(\HC)$.

When $G$ is not a finite group, we assume it to be a compact Lie group  with an associated Lie algebra $\mf{g}$ such that $e^{\mf{g}}=G$. That is, $\mf{g}=\{a | e^a\in G\}$.  In particular, if $G$ has a representation $R$ then $\mf{g}$ has a representation $r$ given by the differential of $R$, i.e., given $a\in\mf{g}$, $R(e^a)=e^{r(a)}$.

In this work we will mainly focus on the  \emph{adjoint representation} of $G$, as it describes how the group acts on density matrices (and other bounded operators). A unitary representation $R$ on $\mathcal{H}$ induces an action on $\mathcal{B}(\mathcal{H})$, given by
\begin{align}\label{eq:adjoint}
   \Adj_{R(g)} (\rho) = R(g) \rho R(g)^\dagger, \ \forall g \in G, \rho \in \mathcal{B}(\mathcal{H}).
\end{align}
where $\Adj_{R(g)}$ denotes the adjoint representation.
Note that for the case of Lie groups, the adjoint representation also exists at the Lie algebra level and is given by $\text{ad}_{r(a)}(\cdot) = [r(a),\cdot]$.

To finish this section, we find convenient to define a distinction between symmetry groups.

\begin{definition}[Inner and outer symmetries]\label{def:inner-outer}
Given a composite Hilbert space, we call a representation of a group an \textbf{inner} symmetry if it acts locally on each subsystem, and an \textbf{outer} symmetry if it permutes the subsystems.
\end{definition}

For instance, when working with $n$-qubit systems, the tensor representation of $\mbb{SU}(2)$,  $R(g\in \mbb{SU}(2))=g^{\otimes n}$ is an inner symmetry, as it acts locally on each subsystem. On the other hand, the qubit-permuting representation of $S_n$, given by $R(g)\bigotimes_{j=1}^n \ket{\psi_j} =   \bigotimes_{j=1}^n \ket{\psi_{g^{-1}(j)}}$ is an outer symmetry.

\subsection{Equivariance and invariance in quantum neural networks}\label{sec:equiv_inv}

Here we present a recipe to obtain $G$-invariant QML models based on the key concept of \textit{equivariance}, which we will first define for linear maps and then for operators. 

Given a group $G$ and a representations $R$, we typically say that a linear map $\phi: \BC(\HC)\rightarrow \BC(\HC)$ is equivariant if and only if $\phi \circ \Adj_{R(g)} = \Adj_{R(g)} \circ \phi$ for all $g\in G$. We can extend this definition by noting that neither the input and output representations, nor the Hilbert spaces, need to be the same. Thus, we consider the following general definition.
\begin{definition}[Equivariant map] Given a group $G$ and its representations $(R^{\text{in}}, \HC^{\text{in}})$ and $(R^{\text{out}}, \HC^{\text{out}})$. A linear map $\phi: \BC^{\text{in}} \rightarrow \BC^{\text{out}}$ is $(G,R^{\text{in}},R^{\text{out}})$-equivariant if and only if
    \begin{align}
    \phi \circ \Adj_{R^{\text{in}}(g)} = \Adj_{R^{\text{out}}(g)} \circ \phi, \ \forall g\in G.
    \end{align}
    \label{def:equiv}
\end{definition}
The property of equivariance can be visualized via the following commutative diagram:
\begin{equation}
\begin{tikzcd}[row sep=large,column sep=huge]
\BC^{\text{in}}_{}\,\,\arrow[r, "\text{Ad}_{R^{\text{in}}(g)}"]\arrow[d, "\phi"]&
\,\,\BC^{\text{in}}_{}\arrow[d, "\phi" ] \\
\BC^{\text{out}}\,\,\arrow[r, "\text{Ad}_{R^{\text{out}}(g)}"]&\,\, \BC^{\text{out}}
\end{tikzcd}\,.\nonumber
\end{equation}
The action of an equivariant map $\phi$ \textit{commutes} with the action of the group. That is, for an equivariant $\phi$ it is equivalent to (i) first acting with $\phi$ and then acting with $R^{\text{out}}$, or to (ii) first acting with $R^{\text{in}}$ and then acting with $\phi$. Note that, for the special case of $R^{\text{out}}$ being the trivial representation, the map is \emph{invariant}: such that $\phi \circ \Adj_{R^{\text{in}}(g)} = \phi$ for all $g\in G$. 

Next, let us define what an equivariant operator is. 
\begin{definition}[Equivariant operator]\label{def:equiv-op}
Given a group $G$ and its representation $(R, \HC)$, an operator  $O\in\BC(\HC)$ is $(G,R)$-equivariant if and only if 
    \begin{align}
    [O,R(g)]=0, \ \forall g\in G.
    \end{align}
\end{definition}
Evidently, Definition~\ref{def:equiv-op} implies that $O\in\mf{comm}(R)$, from which we can easily see that $\mf{comm}(R)$ is the space of all equivariant operators. Moreover, we can also see that the adjoint action of a $(G,R)$-equivariant operator, is a $(G,R,R)$-equivariant map. That is, $\text{ad}_{O}\circ \text{Ad}_{R(g)} =  \text{Ad}_{R(g)}\circ \text{ad}_{O}$.

The previous definitions present us with a recipe to build GQML models of the form in Eq.~\eqref{eq:predicted-label} whose outputs are invariant under the action of the group.
\begin{proposition}[Invariance from equivariance]\label{prop:invar}
A model consisting of an $(G,R^{\text{in}},R^{\text{out}})$-equivariant QNN and a $(G,R^{\text{out}})$-equivariant set of measurements is $G$-invariant.
\end{proposition}
\begin{proof}
For every $g\in G$, $\rho\in \BC^{\text{in}}$ and $\thv$ we have
\begin{align}
h_{\thv}(\text{Ad}_{R^\text{in}(g)}(\rho))&=\CC( \{\Tr[\NC_{\thv}(\text{Ad}_{R^\text{in}(g)}(\rho)) O_j]\}_j)\nonumber\\
&=\CC( \{\Tr[\text{Ad}_{R^\text{out}(g)}(\NC_{\thv}(\rho)) O_j]\}_j)\nonumber\\
&=\CC( \{\Tr[\NC_{\thv}(\rho)R^\text{out}(g)\ad O_j R^\text{out}(g)]\}_j)\nonumber\\
&=\CC( \{\Tr[\NC_{\thv}(\rho) O_j] \}_j)=h_{\thv}(\rho)\,.\qedhere
\end{align}
\end{proof}

Armed with the previous definitions we are now ready to present the basic framework for EQNNs.

\begin{figure}
    \centering
    \includegraphics[width=\columnwidth]{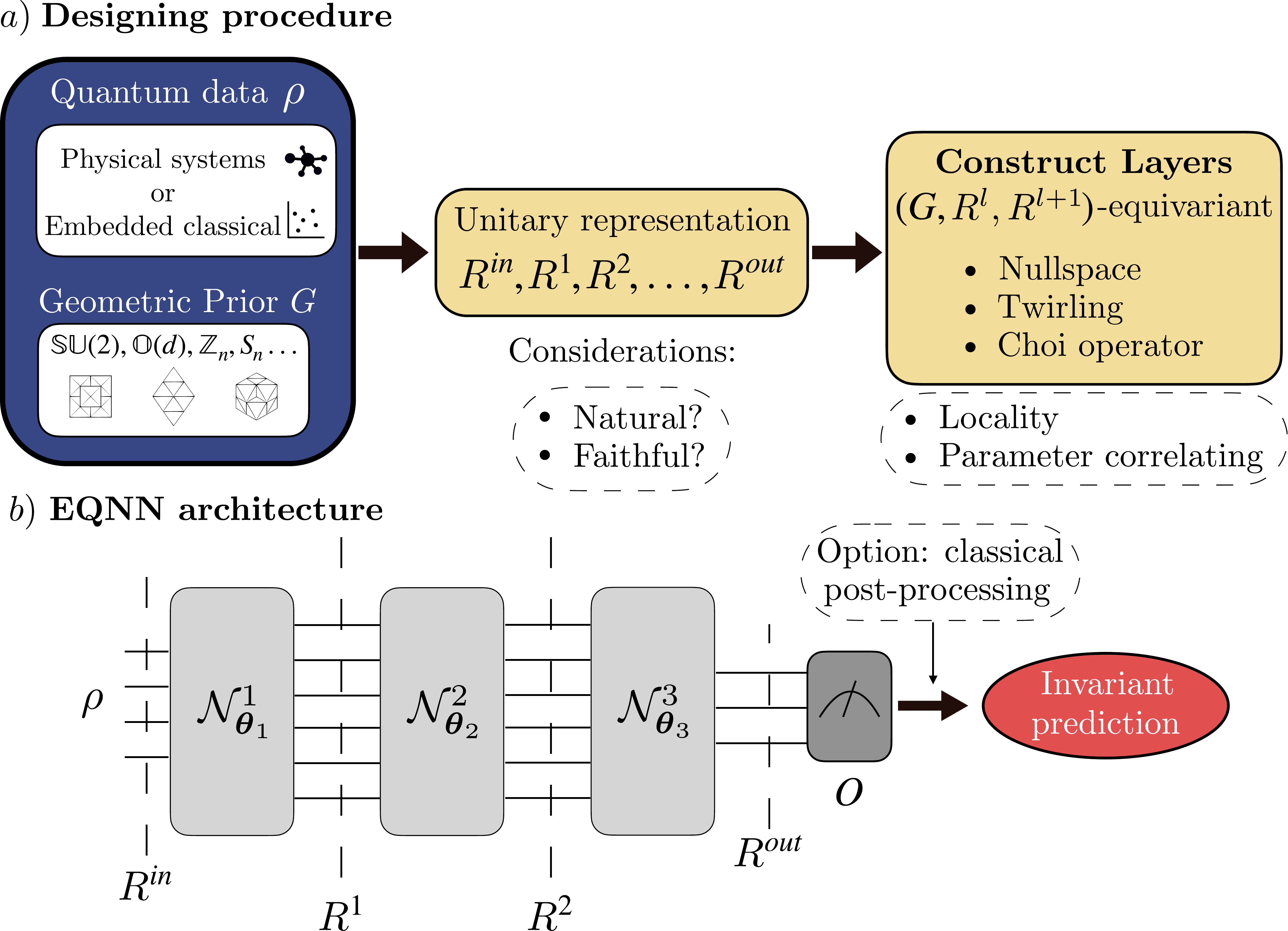}
    \caption{\textbf{Equivariant quantum neural network.} 
    a) We consider  a QML problem composed of a dataset (that can either be quantum mechanical in nature, or corresponding to classical data that have been encoded in quantum states) as well as a label symmetry group $G$. The first step  is to define the input and output representation of $G$ at each layer, where these can be natural, faithful, non-faithful, etc. From here, we will provide different techniques which allow us to construct the EQNN layers and control, for instance, the locality of their gates. b) Dashed lines indicate the representations of the symmetry group $G$ at specific stages in the EQNN, which may change between layers. At first, the input state $\rho_\text{in}$ is acted upon by the representation $R^\text{in}$. The $l$\textsuperscript{th} layer of the EQNN, $\mathcal{N}^l_{\thv_l}$, must be $(G,R^{l},R^{l+1})$-equivariant. In sum, the full architecture, $\phi = \mathcal{N}^L_{\thv_L} \circ \cdots \circ \mathcal{N}^1_{\thv_1}$, is $(G,R^\text{in},R^\text{out})$-equivariant. The $(G,,R^\text{out})$-equivariant measurement operator $O$ is in the commutant of the output representation $R^\text{out}$. Note that if we only want the EQNN to produce an output state equivariantly or invariantly (e.g. in generative models), we can omit the measurements.}
    \label{fig:equiv_circuit}
\end{figure}

 First, however, we find it instructive to provide an example of a classification problem naturally amenable to these methods. Suppose that we are given the ground states of the bond-alternating Heisenberg model, which has the Hamiltonian
\begin{align}
    H = J_1 \sum_{i \, \text{even}} S_{i} \cdot S_{i+1} + J_2 \sum_{i \, \text{odd}} S_i \cdot S_{i+1} \, ,
\end{align}
where $S_i = (S_x^i, S_y^i, S_z^i)$ is the spin operator for the $ith$ qubit. There are two phases of matter for this Hamiltonian: trivial and topologically protected. As a learning problem, we consider the task of determining if the states are in the trivial or topologically protected phases. Consider the representation $R(g) = g^{\otimes n}$ of $\mbb{SU}(2)$. For a ground state $\ket{\psi}$, one can show that $R(g)\ket{\psi}$ is also a ground state. Thus, the labels of states are invariant under an action of $\mbb{SU}(2)$. In Sec.~\ref{sec:numerics} we return to this problem and show that a EQNN significantly outperforms a quantum convolutional neural network for this task.

\section{Theory of equivariant layers for EQNNs}\label{sec:connecting}

In this section we will shed some light on the importance of the choice of representation by studying how EQNNs act on data, and how many degrees of freedom they have. Most notably, we will show that layers that are equivariant to different representations can  process data in different ways, so that a given layer could potentially ``see'' information that is inaccessible to another one. The latter will point to the crucial importance that intermediate representations have. Finally, we will  present a classification of EQNN layers based on their input and output representations, allowing them to  be non-linear, and even change the symmetry group under which they are equivariant. Our results can be summarized in Fig.~\ref{fig:equiv_circuit}.

\subsection{Equivariant layers as Fourier space actions}\label{sec:interpret}

Let us start by analyzing how EQNNs act on data. For simplicity, we first consider the case when $R^{\text{in}}=R^{\text{out}}=R$, the states in the dataset are pure $\rho=\ketbra{\psi}{\psi}$, and the EQNN is unitary, i.e., $\NC_{\thv}(\rho)=U(\thv)\rho U(\thv)\ad$.  Note that if $\NC_{\thv}$ is a $(G,R,R)$-equivariant map, then $U(\thv)$ is a $(G,R)$-equivariant operator, and hence, it belongs to $\mf{comm}(R)$. Hence, we can understand the action of $U(\thv)$ on $\ket{\psi}$ by studying the structure of the commutant.

\begin{theorem}[Structure of commutant, Theorem IX.11.2 in \cite{simon1996representations}] Let $R$ be a unitary representation of a finite-dimensional compact group $G$. Then under the same change of basis $W$, which block diagonalizes $R$ as in Eq.~\eqref{eq:isotypic}, any operator $H\in\mf{comm}(R)$ takes the following block-diagonal form
\begin{equation}
    H \cong \bigoplus_{\lambda}
\id_{d_\lambda} \otimes H_{\lambda} ,
\label{eq:comm-structure}
\end{equation}
where each $H_{\lambda}$ is an $m_\lambda$-dimensional operator that is repeated $d_\lambda$ times.
\label{thm:comm}
\end{theorem}

The previous theorem shows that any equivariant unitary can be expressed as $U(\thv)=W\ad\left( \bigoplus_\lm \id_{d_\lm}\otimes U_\lm(\thv) \right)W$, indicating that in the irrep basis, it can only act non-trivially on the multiplicity space. Drawing a parallelism with the classical machine learning literature, where it has been shown that linear equivariant maps can only act on the group Fourier components of the the data~\cite{kondor2018clebsch, cohen2018spherical, thiede2020general, pan2022permutation, anderson2019cormorant}, we can also here interpret EQNNs as a form of generalized Fourier space action. Specifically, the action of $U(\thv)$ can be understood as: (i) first transforming the data to the generalized Fourier space $W \ket{\psi}=\bigoplus_\lm \ket{\overline{\psi}_\lm} \otimes  \ket{ \psi_\lm}$, (ii)  acting on each Fourier component $\ket{ \psi_\lm}$ with $U_\lm(\thv)$, and (iii) transforming back with $W\ad$. That is,
\begin{equation}
U(\thv)\ket{\psi} = W^\dagger\left( \bigoplus_\lambda \ket{\overline{\psi}_\lambda} \otimes U_\lambda(\thv) \ket{\psi _\lambda}\right)\,.\label{eq:gen-fourier}
\end{equation}
Note that this interpretation can be readily generalized to channels.

Here, we can see that once the representation of $G$ is fixed, so is the information in the input state one has access to (equivariantly). Explicitly, the EQNN cannot manipulate information stored in the components $\ket{\overline{\psi}_\lm}$ of the input state. As we will see in the next subsection, one can still try to access this information via changes of representation.

Notably, Eq.~\eqref{eq:gen-fourier} generalizes \emph{group convolution} in the Fourier basis: When $R$ is the regular representation, the change of basis is the well-known group Fourier transform \cite{folland2016course, childs2010quantum} (see Appendix \ref{app:structure}). 
This generalized Fourier space picture has proved crucial in designing various classical architectures \cite{kondor2018clebsch, cohen2018spherical, thiede2020general, pan2022permutation, anderson2019cormorant}. This also provides a representation-theoretic justification for the recent quantum ``convolutional layers'' in~\cite{cong2019quantum}.
Recently, this interpretation of equivariant unitaries has also been noted for the special case of $\mbb{SU}(d)$-equivariant quantum circuits in~\cite{zheng2021speeding,zheng2022super}.

\subsection{Free parameters in EQNNs}

\subsubsection{Equivariant unitaries}

The Fourier space picture previously discussed enables the counting of free parameters in equivariant unitaries.
\begin{theorem}[Free parameters in equivariant unitaries] Under the same setup as Theorem \ref{thm:comm}, the unitary operators in $\mf{comm}(R)$ can be fully parametrized by $\sum_{\lm} m_\lm^2$ real scalars.
\label{thm:num-parameters}
\end{theorem}
\begin{proof} Any unitary $U$ in $\mf{comm}(R)$ takes the block-diagonal form $U = \bigoplus_{\lambda}
\id_{d_\lambda} \otimes U_{\lambda}$ in the Fourier basis. Observe that the operators $U_\lambda$ must also be unitaries since $U^{\dagger}U = \bigoplus_{\lambda}
\id_{d_\lambda} \otimes U_{\lambda}^{\dagger} U_{\lambda}$. A unitary in $\mathbb{U}(m_\lambda)$ is parametrized by $m_\lambda^2$ real scalars, hence a total number of $\sum_{\lambda} m_\lambda^2$ parameters suffice to parametrize $U$.
\end{proof}

Theorem \ref{thm:num-parameters} describes how ``significant'' the symmetry is to the problem, in the sense that the larger the representations of $G$, the smaller the commutant and thus the fewer parameters needed to fully characterize equivariant unitaries. In Table~\ref{tab:free_params} we present examples of different symmetries constraining the number of free parameters in a unitary EQNN to both exponentially many, polynomially many, and constant.

\noindent
\begin{table}[t]
    \centering
\begin{tabular}{|c||*{3}{c|}}\hline
Group&{Representation}&{Free parameters} & $\mathfrak{comm}(R)$\\\hline\hline
None   &  $R_{\text{trivial}}(g)=\id_{2^n}$ & $4^{n}$ &  $\mathbb{C}[\mathbb{U}(2^n)]$ \\\hline
$\mathbb{U}(2^n)$   &  $R_{\text{def}}(g)=g$ & $1$ &  $\mathbb{C}[\id_{2^n}]$ \\\hline
$\mathbb{U}(2)$   &  
$R_{\text{tens}}(g)=g^{\otimes n}
$ & $\frac{1}{n+2}\binom{2n+2}{n+1}\in\Omega(2^n)$ &  $\mathbb{C}[R_{\text{qub}}(S_n)]$ \\\hline
$S_n$   & {\footnotesize
$\begin{array}{l} R_{\text{qub}}(g)\bigotimes_{i=1}^n \ket{\psi_i}  \\  =   \bigotimes_{i=1}^n \ket{\psi_{g^{-1}(i)}} \end{array}$}
  & $\binom{n+3}{3}\in\Theta(n^3)$  &  $\mathbb{C}[R_{\text{tens}}(\mathbb{U}(2))]$ \\\hline
\end{tabular}
    \caption{ \textbf{Free parameters in  unitary EQNNs.} We show how different symmetries impact the number of free parameters in a  $(G,R,R)$-equivariant unitary. Here, for set a $S$, we have defined  $\mathbb{C}[S]\equiv\text{span}_{\mathbb{C}}(S)$.}
    \label{tab:free_params}
\end{table}

\subsubsection{Equivariant channels}
We have already seen how the inductive biases in unitary EQNNs affect their structure, and concomitantly their number of free parameters. We now turn our attention to $(G,R^{\text{in}},R^{\text{out}})$-equivariant channels. First, recall that any linear channel $\phi: \mathcal{B}^\text{in}\rightarrow \BC^\text{out}$, can be fully characterized through its Choi operator~\cite{wilde2013quantum}
\begin{align}\label{eq:choi_operator}
    J^{\phi} = \sum_{i,j}\ket{i}\bra{j}\otimes \phi(\ket{i}\bra{j}),
\end{align}
which acts in $\mathcal{B}(\mathcal{H}^\text{in}\otimes \mathcal{H}^\text{out})$. The action of $\phi$ on an input state $\rho \in \mathcal{B}^{\text{in}}$
can be recovered from $J^\phi$ as follows~\cite{watrous2018thetheory}
\begin{equation}
    \phi(\rho) = \operatorname{Tr}_{\text{in}}[J^\phi (\rho^\top \otimes \id_{\dim(\HC^{\text{out}})})].
\end{equation}
The Choi operator is related to equivariance via the following theorem.
\begin{lemma}[Lemma 11 in \cite{gschwendtner2021programmability} paraphrased] A channel $\phi$ is $(G,R^{\text{in}},R^{\text{out}})-$equivariant if and only if $J^{\phi} \in \mathfrak{comm}({R^{\text{in}}}^* \otimes R^{\text{out}})$,
where the $*$ symbol denotes complex conjugate.
\label{lem:choi-equivariance}
\end{lemma}

Noting that $({R^{\text{in}}}^* \otimes R^{\text{out}})$ is a valid representation as per Definition~\ref{def:representation}, we can combine Theorem~\ref{thm:comm} and Lemma~\ref{lem:choi-equivariance} to determine a parameter count for general equivariant channels.
\begin{theorem}[Free parameters in equivariant channels] Let the irrep decomposition of $R:=R^{\text{in}*} \otimes R^{\text{out}}$ be $R(g) \cong \bigoplus_{q} R_q(g) \otimes  \id_{m_q}$. Then any $(G,R^{\text{in}},R^{\text{out}})$-equivariant CPTP channels can be fully parametrized via $\sum_{q} m_q^2 - C(R^{\text{in}},R^{\text{out}})$ real scalars, where $C(R^{\text{in}},R^{\text{out}})$ is a positive constant that depends on the considered representations.
\label{thm:paramcountchannel}
\end{theorem}

We defer the proof to Appendix~\ref{app:proofs}. Intuitively, the parameter count of equivariant CP maps follows similarly to the proof of Theorem~\ref{thm:num-parameters} and the extra term $C(R^{\text{in}},R^{\text{out}})$ arises from imposing that the channel $\phi$ must be trace preserving (TP).

Similar to classical ML literature~\cite{cohen2016steerable}, the parameter count benefit of using equivariant layers can be assessed via the \emph{parameter utilization} metric
\begin{equation}
    \mu = \frac{\operatorname{dim} \operatorname{Hom}^{\text{CPTP}}(R^{\text{in}},R^{\text{out}})}{\operatorname{dim} \operatorname{Hom}_G^{\text{CPTP}}(R^{\text{in}},R^{\text{out}})},
\end{equation}
where $\operatorname{Hom}^{\text{CPTP}}(R^{\text{in}},R^{\text{out}})$ denotes the set of CPTP maps between $\BC^{\text{in}}$ and $\BC^{\text{out}}$ and $\operatorname{Hom}_G^{\text{CPTP}}(R^{\text{in}},R^{\text{out}})$ its $(G,R^{\text{in}},R^{\text{out}})$-equivariant subspace. Note that, $\operatorname{dim} \operatorname{Hom}^{\text{CPTP}}(R^{\text{in}},R^{\text{out}})=\lvert \mathcal{H}^\text{in}\rvert^2\lvert \mathcal{H}^\text{out}\rvert^2-\lvert \mathcal{H}^\text{in}\rvert^2$ \cite{wilde2013quantum}. That is, the larger $\mu$, the larger the benefit of using an EQNN is, in the sense that available parameters are used more effectively. For instance, by imposing $\mbb{SU}(2)$ equivariance on 2-to-2 qubit channels, one reduces the number of free parameters to less than or equal to 14 (see Section~\ref{sec:examples}), yielding a reduction of $\mu \geq 240/14 \approx 17$.

\subsection{Intermediate representations as hyperparameters}

Let us here discuss an aspect of EQNNs which has been purposely overlooked up to this point. Namely, while the  input representation $R^{\text{in}}$  is fixed by the action of the symmetry group on the input data, the intermediate and output representations acting on the spaces $\BC^l$ are not. This means there exists freedom in choosing a sequence of representations $(R^{\text{in}},R^1,\hdots, R^{\text{out}})$ under which the layers are equivariant. That is,
\begin{definition}[Layered EQNN] An $L$-layered $G$-equivariant QNN is defined by a sequence of $L+1$ representations of $G$, $(R^{\text{in}},R^1,\hdots, R^{\text{out}})$, and a sequence of $(G,R^{l},R^{l+1})$-equivariant layers.
\label{def:eqnn}
\end{definition}

The equivariance encoded in the $L$-layered EQNN of definition~\eqref{def:eqnn} can be visualized via the following commutative diagram
\begin{equation}
\begin{tikzcd}[row sep=large,column sep=huge]
\BC^{\text{in}_{}}\,\,\arrow[r, "\Adj_{R^{\text{in}}(g)}"]\arrow[d, "\NC^1_{\thv_1}"]& 
\,\,\BC^{\text{in}}_{}\arrow[d, "\NC^1_{\thv_1}" ] \\
\BC^{1}_{}\,\,\arrow[r, "\Adj_{R^{1}(g)}"]\arrow[d, "\NC^2_{\thv_2}"]&\,\, \BC^{1}_{}\arrow[d, "\NC^2_{\thv_2}" ]\\
\vdots\,\,\arrow[r, " "]\arrow[d, "\NC^L_{\thv_L}"]&\vdots\arrow[d, "\NC^L_{\thv_L}"]\\
\BC^{\text{out}}_{}\,\,\arrow[r, "\Adj_{R^{\text{out}}(g)}"]&\,\, \BC^{\text{out}}_{}
\end{tikzcd}\,.\nonumber
\end{equation}
The previous allows us to see that $\mathcal{N}=\NC^L_{\theta_L} \circ \cdots \circ \NC^1_{\theta_1}$ is an $(G,R^\text{in},R^\text{out})$-equivariant QNN. Evidently, if we follow such QNN with $(G,R^{\text{out}})$-equivariant measurements, we achieve a $G$-invariant model.

Note that, as previously discussed, a representation defines a Fourier space, meaning that 
it determines the space over which a layer of EQNN can act, or alternatively the information in the states that can be accessed. 
As such, one can use intermediate representations to change how the model accesses and processes information, which can fundamentally determine the success of the learning model. 

The most general way of fully specifying a representation is via the multiplicities of the irreps. Thus, the irrep multiplicities $m_\lambda^l$ of the intermediate representations can be understood as hyperparameters of the EQNN, similar to the number of channels in a conventional CNN. While in general there is no strict rules on what representations to use, here  we discuss strategies to choose the intermediate representations that are physically meaningful and ease calculations of equivariant layers.

First, one should choose representations that are \textit{natural} on quantum systems. For example, the unitary group $\mathbb{U}(2)$ admits a natural representation on $n$-qubits, where $\HC=(\mbb{C}^2)\tn$, as $R(U)=U^{\otimes n}$ and the cyclic group $\mathbb{Z}_n$ admits the representation the cyclic group $\mathbb{Z}_n$ admits the representation $R(g^t) \bigotimes_{j=1}^n \ket{\psi_j}=\bigotimes_{j=1}^n \ket{\psi_{j+t \mod n}}$, i.e., cyclic shifting the qubits. Second, the following proposition asserts that equivalent intermediate representations yield the same model expressibility~\cite{sim2019expressibility,holmes2021connecting}, and hence it suffices to consider inequivalent ones when designing EQNNs. We defer the proof to Appendix~\ref{app:proofs}.

\begin{proposition}[Insensitivity to  equivalent representations] Consider an EQNN as defined in Definition \ref{def:eqnn}. Then changing
an intermediate representation, $R^l$, to another representation equivalent to it, $VR^l V^{\dagger}$, where $V$ is a unitary, does not change the expressibility of the EQNN.
\label{prop:equivalentreps}
\end{proposition}

Finally, we note that the case of finite groups and regular representations (i.e., when the intermediate representations are chosen to be $R_{\text{reg}}:G\rightarrow \mbb{C}[G]$ corresponding to the group action on its own group algebra) has been studied in the classical literature under the name of \textit{homogeneous ENNs}~\cite{cohen2021equivariant}. In this case, any equivariant map is a group convolution \cite{kondor2018generalization}, which can be realized as a unitary operator embedding the classical convolution kernel by the quantum algorithms in \cite{castelazo2021quantum}.
Combining this with quantum algorithms for polynomial transformations of quantum states \cite{gilyen2019quantum, holmes2021nonlinear} allows one to quantize classical homogeneous ENNs. In other words, classical homogeneous ENNs can be implemented on a quantum computer as a special case of EQNNs.

\subsection{Field-guide to equivariant layers}
\label{sec:quantum-conv}

As previously discussed, intermediate representations can be considered as hyperparameters for the EQNN. In what follows we  define and characterize different types of equivariant layers arising from different intermediate representations.

\subsubsection{Standard, embedding and pooling}

We start by presenting a definition that categorizes equivariant layers based on the sizes of input and output representations.

\begin{definition}[Equivariant layers: standard, embedding and pooling] Let  $\NC_{\thv_l}^l:\BC^{l-1}\rightarrow \BC^{l}$ be an $(G,R^{l-1},R^{l})$-equivariant layer. We say that $\NC_{\thv_l}^l$ is a pooling layer if $\dim(\BC^{l})<\dim(\BC^{l-1})$,
an embedding layer if $\dim(\BC^{l})>\dim(\BC^{l-1})$, and
a standard layer if $\dim(\BC^{l})=\dim(\BC^{l-1})$.
\label{def:pooling-embedding}
\end{definition}

Definition~\ref{def:pooling-embedding} does not require the layer to be a quantum channel and is thus applicable beyond the context of quantum-to-quantum layers, e.g., in quantum algorithms with classical post-processing as we discuss in Appendix~\ref{sec:perspective}. For the special case of EQNNs mapping  from a Hilbert space of $n$ qubits to a Hilbert space of $m$ qubits, we say that $\NC_{\thv_l}^l$ is a pooling layer if $m<n$,
an embedding layer if $m>n$, and
a standard layer if $m=n$. 

Equivariant quantum circuits have been proposed and used in previous works \cite{verdon2019quantumgraph, mernyei2021equivariant, skolik2022equivariant, zheng2021speeding, meyer2022exploiting} mostly in the context of graph problems. However, we note that these fall into \emph{standard layers} and our framework provides more flexibility as the operations need not be unitary. An idea of \emph{pooling layers} was proposed in \cite{cong2019quantum}, although the pooling layer they used was not equivariant to the symmetry of the considered classification task (see Appendix~\ref{sec:perspective}). We will provide examples of pooling equivariant layers in later sections. While \emph{embedding layers} have been used to map classical data to quantum data, they are usually not equivariant~\cite{havlivcek2019supervised,thanasilp2021subtleties} (with a few notable recent exceptions~\cite{glick2021covariant,meyer2022exploiting}), meaning that all the symmetry properties of the classical data is lost during the encoding to quantum states. In addition, to our knowledge, embedding layers mapping quantum data to quantum data have not been formalized in QML prior to this work. Intuitively, embedding layers equivariantly embed the quantum data into a larger Hilbert space, allowing access to higher-dimensional irreps and perform non-linearities (discussed below). A prototypical general EQNN architecture using these equivariant layers inspired by classical literature \cite{bekkers2018roto, cohen2021equivariant} is illustrated in Fig.~\ref{fig:architecture}.

\begin{figure}[ht]
    \centering
    \includegraphics[width=0.8\columnwidth]{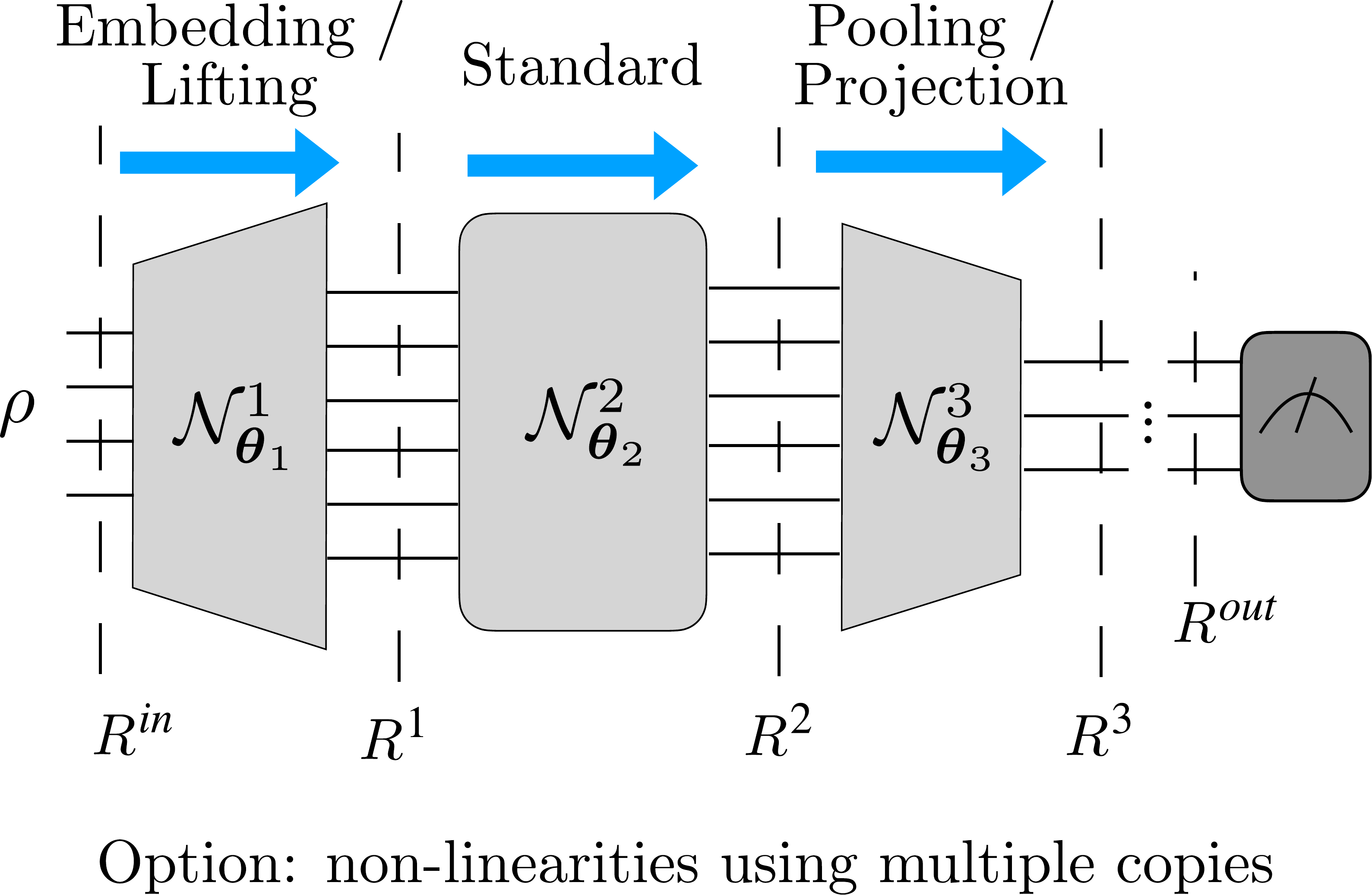}
    \caption{\textbf{Different types of equivariant layers in a general architecture of EQNNs.} A standard layer maps data between spaces of the same dimension. An embedding (pooling) layer maps the data to a higher-dimensional  (smaller-dimensional)  space. In a lifting layer, $\operatorname{ker}(R^{l-1}) > \operatorname{ker}(R^{l})$, while in a projection layer $\operatorname{ker}(R^{l-1}) < \operatorname{ker}(R^{l})$.   }
    \label{fig:architecture}
\end{figure}

\subsubsection{Projection and lifting}

Another common technique in classical geometric deep learning literature is to relax the symmetry constraints in the later layers, typically corresponding to greater-scale features, of the ENNs. This is achieved by \emph{projection layers} (also called reduction layers in some work \cite{cohen2021equivariant}), which go from representations $R^{\text{in}}$ to $R^{\text{out}}$ with $\operatorname{ker}(R^{\text{in}}) < \operatorname{ker}(R^{\text{out}})$. Recall that the kernel of a representation $R$ is defined as $\operatorname{ker}(R):=  \{ g \in  G | R(g) = \id \}$, so that $R$ is faithful if and only if $\operatorname{ker}(R)=\{e\}$. Similarly, one can also define \emph{lifting layers} where  $\operatorname{ker}(R^{\text{in}}) > \operatorname{ker}(R^{\text{out}})$. These lifting layers are used as the first layer in many homogeneous ENN architectures \cite{kondor2018generalization, cohen2019general, bekkers2018roto, finzi2020generalizing}, but their usefulness is not known in general non-homogeneous ENNs \cite{cohen2021equivariant}. Here we similarly define projection and lifting equivariant layers for EQNNs based on the kernels of the representations as follows.

\begin{definition}[Equivariant projection and lifting layers] A $(G,R^{l-1},R^{l})$-equivariant layer is defined as a projection layer if $\operatorname{ker}(R^{l-1}) < \operatorname{ker}(R^{l})$ and a lifting layer if $\operatorname{ker}(R^{l-1}) > \operatorname{ker}(R^{l})$.
\label{def:projandlift}
\end{definition}

Projections layers usually become necessary in pooling layers in the presence of outer symmetries  which exchange subsystems (see Definition~\ref{def:inner-outer}), such as $\mathbb{Z}_n$, $D_n$, or $S_n$ under qubit permuting representations. In contrast to \emph{inner symmetries} which act locally or globally as general unitaries, such as $\mathbb{U}(2)$ with $R(g)=g^{\otimes n}$, outer symmetries typically have no faithful representations when the number of qubits is reduced by a pooling layer. Hence, in this case it is convenient to use a non-faithful representation on the output Hilbert space of fewer qubits, i.e., a projection layer. We provide examples of projection layers in Section~\ref{sec:examples}.

Lifting layers instead can potentially be beneficial when the symmetry of the problem is unsubstantial and does not greatly reduce the number of free parameters in the model, leading to too expressive EQNNs with potential trainability issues. By lifting to a larger group, one can  further reduce the expressibility and potentially improve  trainability~\cite{holmes2021connecting,schatzki2022theoretical}. However, the actual benefit of lifting layers is not known.

Lastly, we note another interpretation of lifting and projection layers as follows. A non-faithful representation $R$ of $G$ with $\text{ker}(R)=H$ can be thought of as a faithful representation of the quotient group $G/H$. Then, lifting layers map from a faithful representation of a quotient group to that of a larger quotient group while projection layers have the opposite effect.

\subsubsection{Non-linearities}
Finally, it is a common practice in QML to assume repeated access to the dataset, which means that one can potentially access multiple copies of the input state $\rho$. The mapping of the form $\rho \rightarrow \rho^{\otimes k}$, which could be applied in the first or an intermediary layer of an EQNN, can thus serve as a \emph{non-linear equivariant embedding layer}, where $R^{\text{out}}=(R^{\text{in}})^{\otimes k}$.
\begin{definition}[Non-linear equivariant embedding layers] An order-$k$ equivariant non-linearity in EQNNs is defined as the composition of a map adding $k-1$ copies of the input state $\phi_{\text{nonlinear}}$: $\rho \rightarrow \rho^{\otimes k}$.
\label{def:nonlinear}
\end{definition}

From the Fourier space perspective, this operation is analogous to the widely used irrep tensor product non-linearity in classical ENNs. For instance, the Clebsch-Gordan decomposition, which computes the tensor product of $\mathbb{SO}(3)$ irreps, has been used in classical literature to achieve universal nonlinearity~\cite{cohen2018spherical, kondor2018clebsch, anderson2019cormorant, bogatskiy2020lorentz}. In the quantum setting, on the other hand, tensor product is performed \emph{naturally} by composing systems, giving opportunities for equivariant data processing on high-dimensional irreps. Indeed, the first step in the quantum-enhanced experiment model \cite{huang2021quantum} performs this non-linear equivariant layer. Doing so can drastically simplify non-linear learning tasks~\cite{huang2021quantum, larocca2022group} (see Appendix~\ref{sec:perspective}).

\section{Methods for constructing equivariant layers for EQNNs}\label{sec:framework}
In this section, we describe methods to construct and train layers in EQNNs. Our first step will be to identify, given a group and its in-and-out representations, the space of equivariant maps. For this purpose we present three distinct approaches based on finding the nullspace of a system of matrix equations, the twirling technique, and on the Choi operator. Once the space of equivariant maps is determined, we discuss how to parametrize and optimize over them. An overview of the results in this section can be found in Table~\ref{table:framework}.

\begin{table*}[]
\centering
\begin{tabular}{|c|c|c|c|c|c|c|c|c|}
        \hline
           & Generating  & Main  &  Time&Locality & CPTP& Kraus rank & Notes & Examples \\
          \textbf{Methods} & set size & technique &  complexity& controlled? & condition & controlled? & &(Sec.~\ref{sec:examples}) \\
         \hline
         \hline
        Nullspace &   Small & Lin. algebra &  Gaussian elim. &Yes & TP easy & Yes & Find all & $\mbb{SU}(2)$
        \\
         &  & based &$\OC(2^{6(m+n)})$  &  & &  &
        equiv. maps & \\
         \hline
         Twirling &   Any & Analytic/&  Depends, mainly & Often & Yes & Nontrivial & One channel &  $\mathbb{Z}_2 \times \mathbb{Z}_2$
         \\
  &  & In-circuit  & Haar-integral  & nontrivial & &  &  at a time & $Z_n$
        \\
        \hline
         Choi  & Any & Irrep   & Depends, mainly&  Nontrivial &  CP easy & Yes & Find all  & $\mbb{SU}(2)$
         \\
         operator &  & decomp.  & irrep decomp.&  & &   &  equiv. maps & \\
         \hline
\end{tabular}
\caption{\textbf{Overview of the different methods for finding equivariant channels}. \textit{Generating set size}: denotes that size of the generating set for which the method is better suited. \textit{Main technique:} indicates the tools used to create the equivariant channels.  Nullspace uses a linear-algebraic approach to impose equivariance on the generating set of the group or its algebra. Twirling uses the twirl formula defined in Eq.~\eqref{eq:twirl-finite} or Eq.~\eqref{eq:twirl-lie}. Choi operator block-parameterizes the Choi operator via an irrep decomposition. \textbf{Time complexity:} denotes the computational complexity of the method. The time complexity of Gaussian elimination is $O(d^3)$ where $d$ is the size of the linear system. Assuming the generating set has size $O(1)$, then the linear system obtained in the nullspace method is of size $2^{2(m+n)}$, where $n$ and $m$ are the number of qubits at the input and output of the map. For twirling and Choi operator, the time complexity is dominated by the Haar-integral and irrep decomposition, respectively. In the case of twirling, it can be computed analytically, approximately, or implemented in-circuit depending on the problem at hand. \textbf{Locality:} determines whether we can control the locality of the operations we need to implement.  \textbf{CPTP:} indicates whether the output channel is CPTP, and how hard it is to impose this condition on the output maps. In nullspace, it is trivial to impose TP on the solution but imposing CP is more involved. Twirling guarantees CPTP as long as the  channel we twirl is CPTP. In Choi operator, imposing CP is straightforward, but TP might be more involved due to the dimension mismatch introduced by the irrep decomposition. \textbf{Kraus rank:} indicates whether we can control the Kraus rank of the channel. \textbf{Notes}: denotes whether we can find a basis for the equivariant map vector space, or if we find one map at a time. In the table, decomp. is short for decomposition, elim. for elimination, lin. for linear, and equiv. for equivariant.}
\label{table:framework}
\end{table*}

\subsection{Simplifying the task of finding equivariant maps}
As per Definition~\ref{def:equiv}, a linear map $\phi$ is equivariant if it satisfies the superoperator equation
\begin{align}
    \phi \circ \Adj_{R^{\text{in}}(g)} - \Adj_{R^{\text{out}}(g)} \circ \phi=0, \ \forall g\in G.
    \label{eq:equivariant-maps}
\end{align}
The set of all such maps forms a vector space, and therefore to characterize them all it suffices to find a basis of this space. While naively it would seem that one needs to solve Eq.~\eqref{eq:equivariant-maps} for every $g\in G$, we will now see that it is usually enough to solve this equation only over a \textit{well-chosen} subset of elements of the group (or of its Lie algebra for Lie groups of symmetries).

\subsubsection{Finite groups}

We first consider the case when $G$ is a finite group. Here, we recall the concept of a \textit{generating set}. A subset $S=\{g_1, \ldots, g_{|S|}\}\subset G$ is a generating set of $G$  if any element of the group can be written as a product of elements in the generating set. Denoting $\langle S \rangle $ the closure of $S$, i.e., the repeated composition of its elements, we say that $S$ generates $G$ if $\langle S \rangle = G$. For example, the symmetric group $S_n$ can be generated by the set of transpositions. It is a well-known fact in group theory that a finite group can be generated with a subset $S$ of, at most, size $\log_2(\lvert G \rvert)$~\cite{fulton1991representation}. Thus, even exponentially large groups can be handled efficiently through their generating set. In particular, we can simplify the task of finding the equivariant maps via the following theorem.
\begin{theorem}[Finite group equivariance]
    Given a finite group $G$ with generating set $S$, a linear map $\phi$ is $(G, R^{\text{in}}, R^{\text{out}})$-equivariant if and only if
    \begin{align}\label{eq_finite_grp_eq}
       \phi \circ \Adj_{R^{\text{in}}(g)} - \Adj_{R^{\text{out}}(g)} \circ \phi=0, \quad \forall g \in S.
    \end{align}
    \label{thm:gen-set}
\end{theorem}

\subsubsection{Lie groups}

While Theorem~\ref{thm:gen-set} is useful when the group $G$ is finitely generated, many relevant groups, such as the Lie group $\mbb{U}(d)$, are not. However, we can consider generating sets, but now at the Lie algebra level. Ref.~\cite{finzi2021practical} provides a method for imposing equivariance under Lie groups, where the equivariance constraint is imposed over a basis of the Lie algebra. Evidently, this becomes impractical in the case of large Lie groups, since the method scales linearly on its dimension. Instead, as we prove below, it suffices to impose the constraint only over a generating set. That is, we can consider $s=\{a_1,\ldots, a_{|s|}\}\subset \mf{g}$ a generating set for $\mf{g}$ if its Lie closure $\langle s \rangle_{\rm{Lie}}$, the repeated nested commutators of the elements of the set, spans the whole Lie algebra.  With this concepts at hand, we are ready to impose equivariance at the algebra level.

\begin{theorem}[Lie group equivariance]\label{thm:equiv_algebra}
    Given a compact Lie group $G$ with a Lie algebra $\mathfrak{g}$ generated by $s$ such that exponentiation is surjective, a linear map $\phi$ is $(G, R^{\text{in}}, R^{\text{out}})$-equivariant if and only if 
    \begin{align}
        \adj_{r^{\text{out}}(a)} \circ \phi - \phi \circ \adj_{r^{\text{in}}(a)}=0, \quad \forall a \in s,
    \end{align}
    where $r^{\text{in}},r^{\text{out}}$ are the representations of $G$ induced by $R^{\text{in}}, R^{\text{out}}$.
    \label{thm:lie-alg}
\end{theorem}
We note that the assumption of surjectivity of the exponential map can be relaxed with incorporation of additional constraints. This relaxation together with pertinent examples and proofs of the theorems are given in Appendix~\ref{app:proofs}.

\subsection{Nullspace, twirling and Choi operator}

With the previous in mind, in this section we present three different techniques that can be used to determine equivariant channels.

\subsubsection{Nullspace method}
In the \textit{nullspace} method, we formulate the equivariance constraints as a linear system of matrix equations, one per element in the generating set, and then solve for their joint nullspace. This yields a basis for the vector space of equivariant linear maps (not necessarily quantum channels). For  the rest of this section we assume we have a finite group and a set of generators at the group level. The case of Lie groups follows analogously by working at the level of the Lie algebra.

Our method generalizes those in~\cite{zeier2011symmetry,finzi2021practical} and proceeds as follows. The first step is to represent the superoperators in Eq.~\eqref{eq_finite_grp_eq} as matrices, sometimes referred to as transfer matrices. This can be achieved through the following map $\phi \mapsto \overline{\phi} = \sum_{i,j} \phi_{i,j} \ket{P_i}\rangle \langle \bra{P_j}$, where $P_j$ and $P_i$ are Pauli operators in the input and output Hilbert spaces, respectively~\cite{wood2011tensor}. Here, $\overline{\phi}$ is a $\dim(\BC^{\text{out}})\times\dim(\BC^{\text{in}})$ matrix. The latter transforms Eq.~\eqref{eq_finite_grp_eq} into a matrix multiplication equation of the form
\begin{align}\label{eq_finite_grp_matrix_rep}
\overline{\phi} \cdot \overline{\Adj}_{R^{\text{in}}(g)} -\overline{ \Adj}_{R^{\text{out}}(g)} \cdot \overline{\phi}=0, \quad \forall g \in S.
\end{align}
Next, we will perform a vectorization~\cite{horn1991topics}, which maps a matrix into a column vector, and allows us to write Eq.~\eqref{eq_finite_grp_matrix_rep} as 
\begin{align}
    M_g \cdot \text{vec}(\overline{\phi}) = 0\,.
\end{align}
Here, $\text{vec}(\overline{\phi})$ is a $\dim(\BC^{\text{in}})\dim(\BC^{\text{out}})$-dimensional column vector and 
\small
\begin{equation}\label{eq_Mg}
    M_g= (\overline{\Adj}_{R^\text{in}(g)})^\top \otimes \id_{\dim(\BC^{\text{out}})}-\id_{\dim(\BC^{\text{in}})}\otimes\overline{\Adj}_{R^\text{out}(g)}\,,
\end{equation}
\normalsize
is a $\dim(\BC^{\text{in}})\dim(\BC^{\text{out}})\times \dim(\BC^{\text{in}})\dim(\BC^{\text{out}})$ matrix. With the previous, we can obtain equivariant maps by computing the intersection of the nullspaces of each $M_g$, i.e.,
\begin{equation}\label{eq:equiv-maps-null}
    {\rm vec}(\overline{\phi}) \in \bigcap_{g\in S} {\rm Null}(M_g)\,.
\end{equation}
In Fig.~\ref{fig:nullspace} we present an example of the nullspace method.

\begin{figure}
    \centering
    \includegraphics[width=\columnwidth]{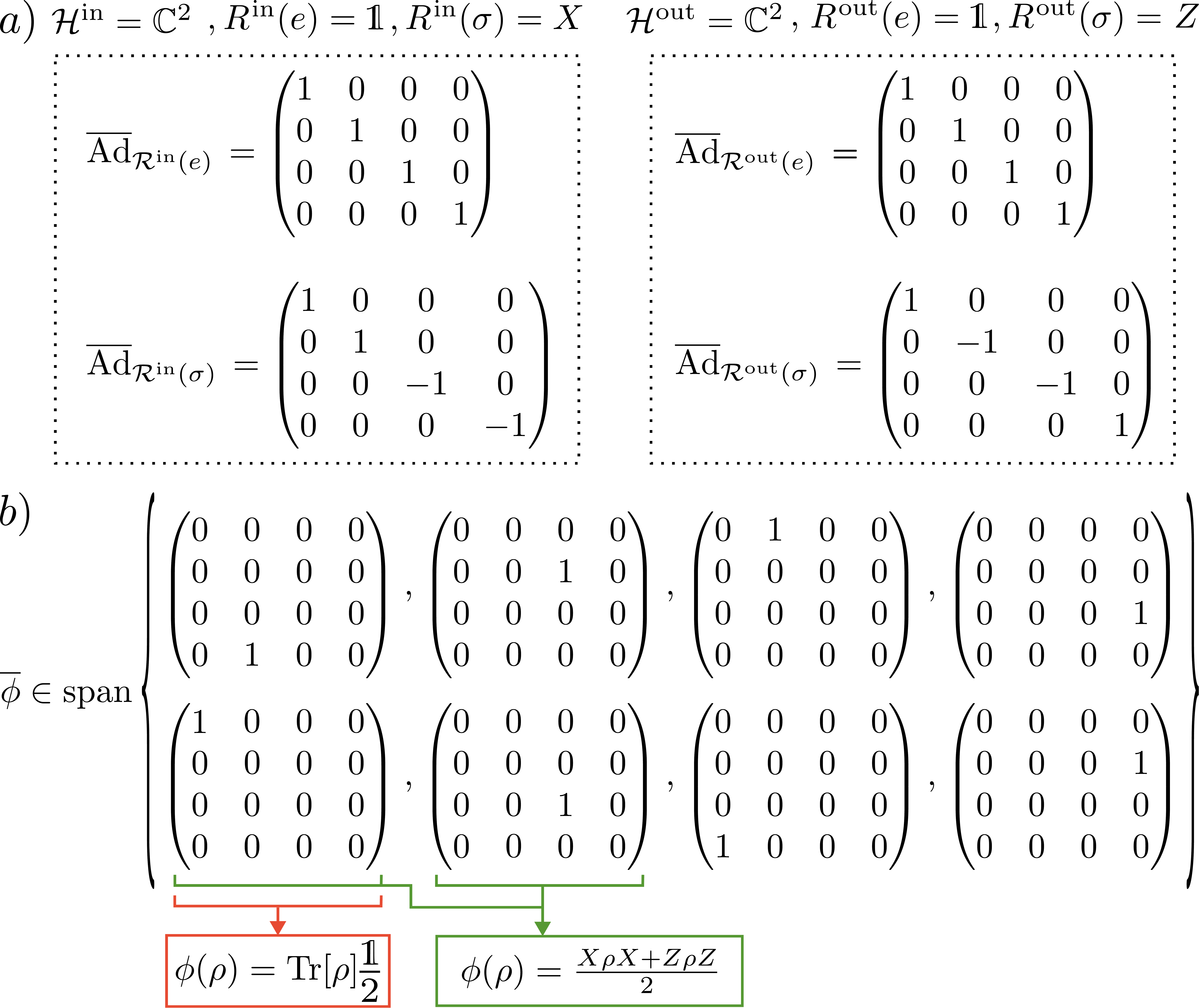}    \caption{\textbf{Example of the nullspace method.} We demonstrate how to use the nullspace method to determine the space of 1-to-1-qubit ($G$,$R^{\text{in}}$,$R^{\text{out}}$)-equivariant quantum channels, with $G=\mbb{Z}_2=\{e,\sg\}$, $R^{\text{in}}=\{\id,X\}$ and $R^{\text{out}}=\{\id,Z\}$. a) The matrix representation of both in and out \textit{adjoint} representations of the symmetry group. b) A basis for the 8-dimensional solution space, as well as two possible equivariant channels: $\phi(\rho)=\Tr[\rho]/2$ obtained from the solution in red, and  $\phi(\rho)=(X\rho X+Z\rho Z)/2$ obtained by combining the two solutions in green. }
    \label{fig:nullspace}
\end{figure}

Let us make here several important remarks about the nullspace method. First, it is clear that this technique can rapidly become computationally expensive. For example, finding equivariant channels mapping from $n$ qubits to $m$  qubits by  solving the nullspaces through Gaussian elimination~\cite{trefethen1997numerical} has a complexity of $\OC(2^{6(m+n)})$. Second, let us note that the solutions of Eq.~\eqref{eq:equiv-maps-null} will lead to a  basis for \textit{all} equivariant linear maps, and therefore additional steps would be required to find the subset of physically realizable operations (see Sec.~\ref{sec:parameterize}). For instance, we can obtain  trace-preserving maps (TP) by noting that $\phi$ is TP if and only if $\overline{\phi}$ contains the term $\frac{\dim(\mathcal{H}^\text{in}) }{\dim( \mathcal{H}^\text{out})} \ket{\id_{\dim(\mathcal{H}^\text{out})}}\rangle \langle\bra{\id_{\dim(\mathcal{H}^\text{in})}}$ and no other terms mapping to $\ket{\id_{\dim(\mathcal{H}^\text{out})}} \rangle$.

In practice, we can significantly reduce the computational complexity of this method by restricting the set of Pauli operators we need to consider in the input and output spaces. This is particularly useful for inner symmetries, where the action of the group can be locally studied. For example, consider the following lemma:
\begin{lemma}[Global equivariance via local equivariance]\label{lemma:local}
Let $\BC^{\text{in (out)}}$ be composite input (output)  spaces of  the form $\BC^{\text{in (out)}}=\bigotimes_j \BC_j^{\text{in (out)}}$. Then,  assume that the representations acting on each of these space  takes a tensor product structure over subsystems as $R^{\text{in (out)}}(g) = \bigotimes_{j} R_{j}^{\text{in (out)}}(g)$. For local equivariant channels mapping between each pair of in-and-out subsystems $\phi_j: \BC^{\text{in}}_j \rightarrow  \BC^{\text{out}}_j$ that are $(G,R_j^{\text{in}}, R_j^{\text{out}})$-equivariant, we have that $\bigotimes_{j} \phi_j$ is $(G,R^\text{in},R^\text{out})$-equivariant.
\end{lemma}
\noindent Thus, we can find equivariant maps locally and take tensor products of them to obtain a global equivariant layer. While such approach can greatly reduce the computational cost (for instance, solving for  2-to-1 qubit maps requires dealing with $64\times 64$ matrices), this will come at the cost of expressibility, as the composition of local equivariant channels may have a restricted action when compared to a general equivariant global channel~\cite{marvian2022restrictions}.  

In the case of outer symmetries such as $G=S_n$, and $R^{\text{in}}(g)=R^{\text{out}}(g)=R_{\text{qub}}(g)$ (as defined in Table.~\ref{tab:free_params}), we can use a generating set $S$ including only local transpositions (i.e., involving only two-body operators). Thus, if we want to obtain maps $\phi$ containing only one- and two-body terms we only need to consider the sub-block of $M_g$ corresponding to one- and two-body Pauli operators, reducing its size from exponential to only polynomial.

\subsubsection{Twirling method}
We now explain a second method for finding equivariant maps based on  \textit{twirling}. This approach was first proposed in~\cite{meyer2022exploiting} to determine equivariant unitary channels. We here extend this framework to general non-unitary quantum channels with (possibly) different representations in the input and output spaces of $\phi$. 

Starting with a given channel $\phi:\BC^{\text{in}}\rightarrow \BC^{\text{out}}$, we define its twirl over a finite symmetry group $G$ as
\begin{align}
    \mathcal{T}_G[\phi] = \frac{1}{|G|} \sum_{g \in G} \Adj_{R^{\text{out}}(g)} \circ \phi \circ \Adj_{R^{\text{in}}(g)}^{\ad}\,.
    \label{eq:twirl-finite}
\end{align}
For the case of Lie groups we replace the summation with an integral over the Haar measure
\begin{align}\label{eq:twirl-lie}
    \mathcal{T}_G[\phi] =  \int_G d\mu(g) \Adj_{R^{\text{out}}(g)} \circ \phi \circ \Adj_{R^{\text{in}}(g)}^{\ad}.
\end{align}
From the invariance of the Haar measure $d\mu$, it is clear that for all $\phi$, $\TC_G[\phi]$ is $(G,R^{\text{in}},R^{\text{out}})$-equivariant. Furthermore, $\mathcal{T}_G[\phi] = \phi$ for all equivariant $\phi$. Combining these observations one can see that $\mathcal{T}_G$ is the projection onto the space of equivariant maps. This realization allows us to write any channel $\phi$ as 
\begin{equation}
    \phi=\mathcal{T}_G[\phi]+\phi_A\,,
\end{equation}
where $\phi_A$ is the ``anti-symmetric'' part of $\phi$, i.e., the part satisfying $\mathcal{T}_G[\phi_A]=0$. As such, any measure of the form $\norm{\phi_A}$ quantifies how symmetric $\phi$ is.

\begin{figure}
    \centering
    \includegraphics[width=.8\columnwidth]{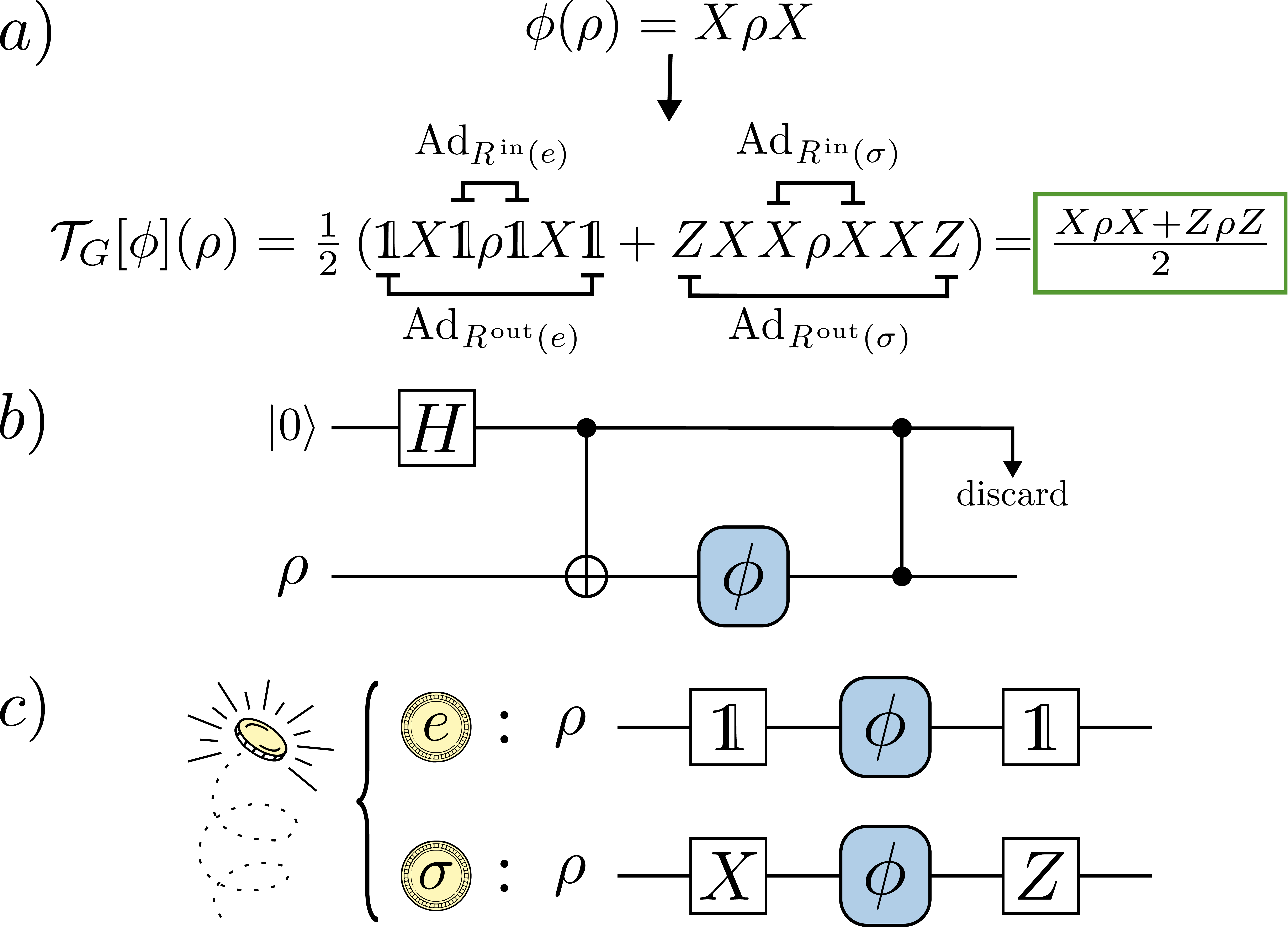}
    \caption{\textbf{Example of the twirling method.} We demonstrate how to use the twirling method to determine the space of 1-to-1 qubit ($G$,$R^{\text{in}}$,$R^{\text{out}}$)-equivariant quantum channels, with $G=\mbb{Z}_2=\{e,\sg\}$, $R^{\text{in}}=\{\id,X\}$ and $R^{\text{out}}=\{\id,Z\}$. a) Explicit calculation using the twirling formula of Eq.~\eqref{eq:twirl-finite}.  b) Ancilla-based scheme for in-circuit twirling. c) Classical-randomness scheme for in-circuit twirling. Both schemes in b) and c), detailed in Appendix~\ref{app:advanced_twirling}, recover the twirling in a). }\label{fig:twirling}
\end{figure}

On the practical side, twirling is easy for small groups, as one can efficiently evaluate  Eq.~\eqref{eq:twirl-finite} (see Fig.~\ref{fig:twirling} for an example). However, for large finite groups or for Lie groups a direct computation of the twirling becomes cumbersome, requiring the use of more advanced techniques. In Appendix~\ref{app:advanced_twirling} we discuss different approaches to implement Eqs.~\eqref{eq:twirl-finite} and~\eqref{eq:twirl-lie}. These range from analytical methods based on the \textit{Weignarten calculus}~\cite{collins2006integration,puchala2017symbolic} (which requires knowledge of the commutant of the representations) to experimental schemes such as \textit{in-circuit} twirling, and \textit{approximate} twirling approaches~\cite{toth2007efficient}. In particular, we present two approaches for in-circuit twirling based on either the use of ancilla qubits, or classical randomness.  Both of these are exemplified in Fig.~\ref{fig:twirling}.

For completeness, let us compare the twirling method to the nullspace approach. First, we note that one of the main advantages of twirling is that, unlike in the nullspace method, we are guaranteed that the twirl of a CPTP channel is also CPTP. However, while the nullspace method allows us to find all equivariant maps, twirling is performed one map  at a time, meaning that finding a complete basis for the equivariant map vector space could be more intricate (although still possible, as we will show in Section~\ref{sec:examples}).  As such, if one wants a single equivariant channel, twirling is strongly recommended.

\subsubsection{Choi operator method}

Here we present a third method for finding equivariant channels. We recall from Eq.~\eqref{eq:choi_operator} that the Choi operator of a channel $\phi$ is given by $J^{\phi} = \sum_{i,j}\ket{i}\bra{j}\otimes \phi(\ket{i}\bra{j})$. Then, as indicated by Lemma~\ref{lem:choi-equivariance}, the channel will be equivariant if $J^{\phi} \in \mathfrak{comm}({R^{\text{in}}}^* \otimes R^{\text{out}})$ (where $*$ denotes complex conjugate and $R^*$ the so-called \textit{dual} representation of $R$), or alternatively if for all $g\in G$
\begin{equation}\label{eq:cho_expanded}
    J^{\phi}({R^{\text{in}}(g)}^* \otimes R^{\text{out}}(g))-({R^{\text{in}}(g)}^* \otimes R^{\text{out}}(g)) J^{\phi}=0\,.
\end{equation}
While we can vectorize Eq.~\eqref{eq:cho_expanded} and obtain equivariant maps by solving for the nullspace of the matrices $({R^{\text{in}}(g)}^* \otimes R^{\text{out}}(g))^\top \otimes \id_{\dim(\BC^{\text{out}})}-\id_{\dim(\BC^{\text{in}})}\otimes({R^{\text{in}}(g)}^* \otimes R^{\text{out}}(g))$ , this would not be significantly different from the nullspace method previously presented. 

Instead, here we focus on a different technique based on the fact that, since  $R^{\text{in}*} \otimes R^{\text{out}}=R$ is a valid representation of $G$, then it has some irrep decomposition  $R(g) \cong \bigoplus_{q} R_q(g) \otimes  \id_{m_q}$ (see Theorem~\ref{thm:paramcountchannel}) and the Choi operator of any equivariant map takes the form
\begin{align}
     J^\phi \cong \bigoplus_q\id_{d_q}  \otimes  J^\phi_q \,.
    \label{eq:choi-decompose}
\end{align}
Equation~\eqref{eq:choi-decompose} allows us to build equivariant maps by controlling precisely how the associated Choi operator acts on each irrep component of the quantum states. In Fig.~\ref{fig:choi} we exemplify this method.
\begin{figure}[t]
    \centering
    \includegraphics[width=\columnwidth]{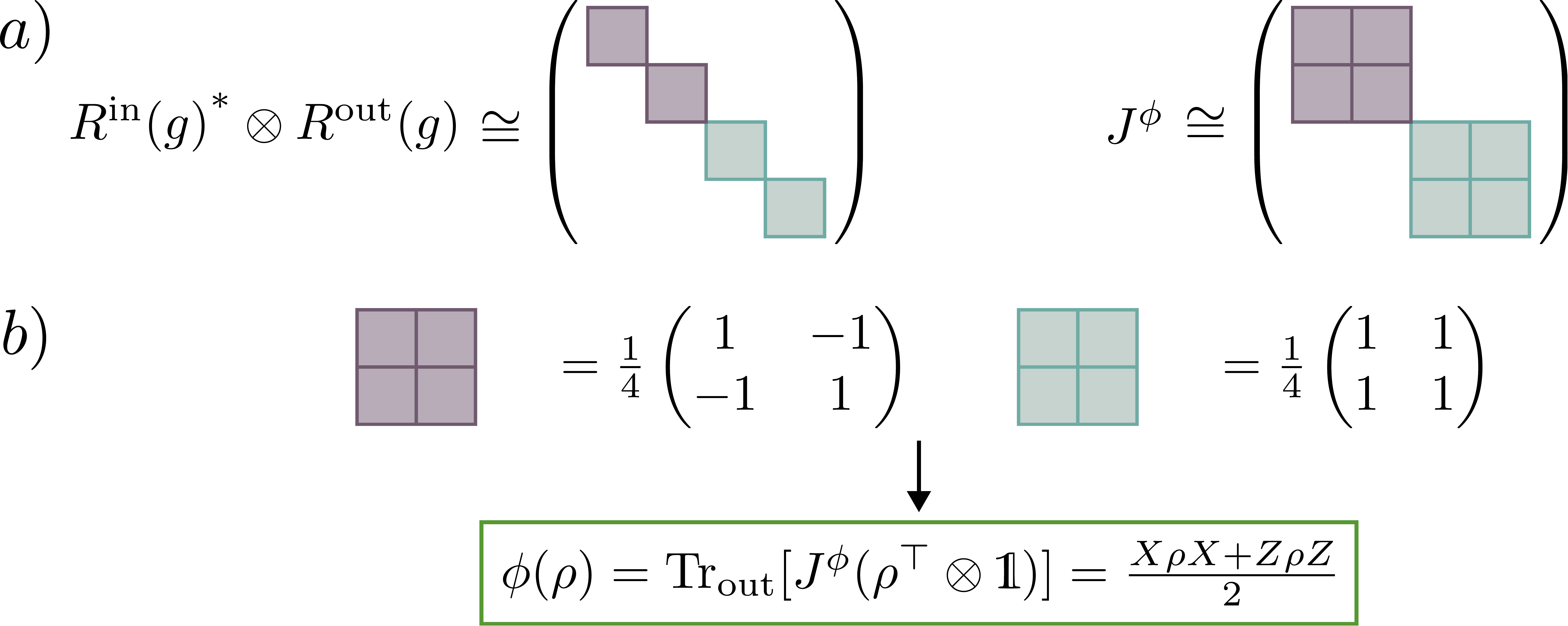}
    \caption{\textbf{Example of the Choi operator method.} We demonstrate how to use the Choi operator method to determine the space of 1-to-1 qubit ($G$,$R^{\text{in}}$,$R^{\text{out}}$)-equivariant quantum channels, with $G=\mbb{Z}_2=\{e,\sg\}$, $R^{\text{in}}=\{\id,X\}$ and $R^{\text{out}}=\{\id,Z\}$. a) Isotypic decomposition of the group representation . b) We show that a specific choice for the block-diagonal components of $J^\phi$ leads to the map $\phi(\rho)=(X\rho X+Z\rho Z)/2$.}
    \label{fig:choi}
\end{figure}

Just like in the nullspace method, this approach produces general equivariant linear maps, and hence, additional constraints need to be imposed to find the subset of physical channels. For instance, we can impose TP by requiring that $ \Tr_{\text{in}}[J^{\phi}] = \id_{\dim(\mathcal{H}^\text{out})}$, where $\Tr_{\text{in}}$ indicates the partial trace over $\mathcal{H}^\text{in}$. Then, we know that $\phi$ will be completely positive (CP) if and only if $ J^{\phi} \geq 0$. The last condition implies that $J^\phi$ can be further expressed as~\cite{dadriano2004extremal}
\begin{align}
    J^\phi \cong \bigoplus_q\id_{d_q}  \otimes  w_q\ad w_q\,,
    \label{eq:choi-decompose2}
\end{align}
where $w_k \in \mathbb{C}^{m_q \times m_q}$. Moreover, the TP conditions leads to $\sum_q \Tr_{\text{in}}[ \id_{d_q}\otimes w_q^\dagger w_q ] = \id_{\dim(\mathcal{H}^\text{out})}$. Thus, given the irrep decomposition of $R$ one can construct a basis of CP maps in the block-diagonal form of Eq.~\eqref{eq:choi-decompose2} and impose the trace-preserving condition afterwards (note however that, taking the partial trace is now more involved due to the subspace mismatch introduced by the isomorphism in Eq.~\eqref{eq:choi-decompose}).
Finally, we remark that one could even go a step further and consider conditions for $\phi$ to be an \emph{extremal} equivariant CPTP channel.\footnote{The set of equivariant CPTP channels forms a convex hull. Channels on the boundary of this  hull are said to be extremal and fully characterize the convex hull. However, this set of extremal channels may not be a finite set.} Conditions for such, however, are much more involved \cite{dadriano2004extremal}.

Let us here remark that the main limitation of the Choi operator approach is that identifying the isomorphism in the irrep decomposition of ${R^{\text{in}}}^* \otimes R^{\text{out}}$ can be in general challenging. For common compact Lie groups, these decompositions are conveniently implemented in a variety of software packages~\cite{bosma1997magma, van1992lie}. That being said, this method is best suited for local channels since the size of the Choi operator scales as $\dim(\HC^{\text{out}})\dim(\HC^{\text{in}})$. Thus, if $\dim(\HC^{\text{out}})\dim(\HC^{\text{in}})$ is not prohibitively large, one can solve for the change of basis of the isotypic decomposition and identify maps with specific irrep actions. 

Lastly, implementing the nullspace method requires representing channels as matrices. (This could be done as well for twirling.) To check that these maps are actually channels, i.e. CPTP, we may want to convert from a matrix to the Choi operator, for which CPTP are readily verified. We discuss how to perform this conversion in Appendix~\ref{app:transfer_to_choi}.

\subsection{Parametrizing the layers of an EQNN}\label{sec:parameterize}
In GQML we are not only interested in finding equivariant maps, but we also want to parametrize and optimize over them. In this section we show how one can parametrize the layers of an EQNN. For simplicity, we first consider the case of unitary channels, and then study the case of general maps.  An overview of the methods proposed in this section can be found in Fig.~\ref{fig:optimizing}.
\begin{figure}
    \centering
    \includegraphics[width=\columnwidth]{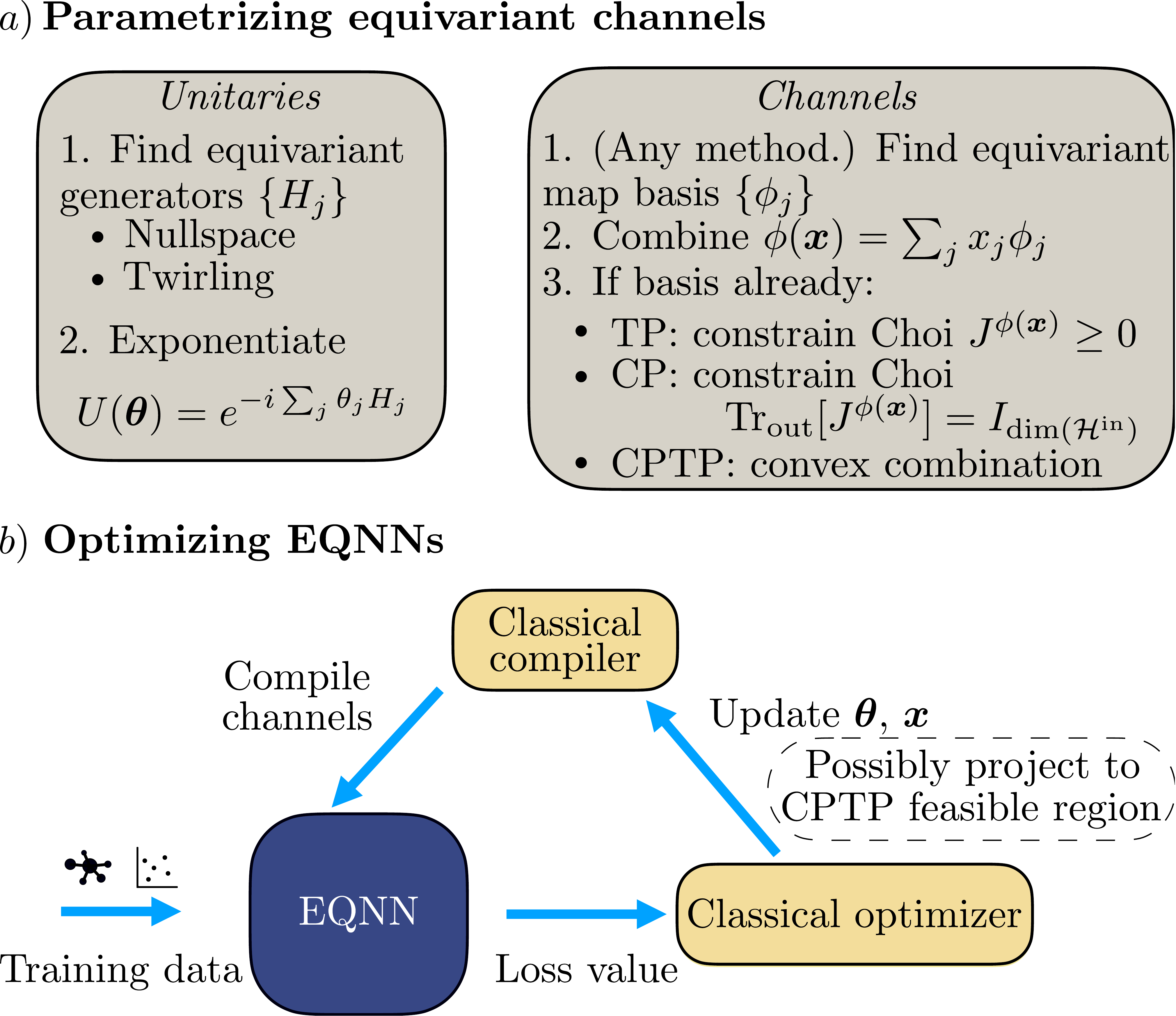}
    \caption{\textbf{Procedure to parametrize and optimize equivariant quantum neural networks.} a) We provide techniques to parametrize both equivariant unitaries  and general equivariant channels. b) Once we have a parametrized EQNN, we can proceed to train it. Here, we feed the training data into the EQNN, whose outputs are used to compute the loss function. Leveraging  a classical optimizer we find updates for the parameters in the EQNN. In the case of channels, if might be necessary to project  the updated map to the feasible CPTP region. Note that we can use classical compilers to transform a linear combination of channels into an sequence of gates that we can implement on a quantum device (Appendix~\ref{app:compile}). The procedure is repeated until convergence is achieved.}
    \label{fig:optimizing}
\end{figure}

\subsubsection{Parametrizing equivariant unitaries}

Let us here consider the case of a unitary EQNN layer with the same input and output representations. That is, $\HC^{\text{in}}=\HC^{\text{out}}$, $R^{\text{in}}=R^{\text{out}}=R$ and $\NC_{\thv_l}^l(\rho)=U_l(\theta_l)\rho U_l(\theta_l)\ad$. While this case has been considered in~\cite{larocca2022group,meyer2022exploiting,sauvage2022building} we will here review it for completeness. 

The simplest way to parametrize a unitary is by expressing it as the exponential of some Hermitian operator usually known as a \textit{generator}, i.e., $U_l(\theta_l)=e^{-i \theta_l H_l}$, where  $\theta_l$ is a trainable parameter. Evidently, we can obtain $(G,R)$-equivariant unitaries by taking equivariant generators, i.e., $H_l\in\mf{comm}(R)$. Note that we can find equivariant generators via the nullspace or twirling approaches previously detailed. While these methods were presented for superoperators, they can be straightforwardly adapted to the case of operators. 
Alternatively, one could use the Choi operator approach and require the solution to be rank-$1$ (recall that the rank of the Choi operator is the Krauss rank of the associated channel, with unitaries being Kraus-rank-1 channels).

\subsubsection{Parametrizing equivariant channels}

Here we describe how to parametrize and optimize over equivariant channels. We assume that a basis of equivariant maps (or at least a subset of this basis) has been found via the nullspace or Choi operator method. 
As mentioned before, while it is relatively easy to find equivariant maps, these need not be physical channels, as they may  not be TP, CP, or neither. However, one can still parametrize a set of non-CPTP equivariant maps and optimize over them by appropriately constraining the parameters such that the final map is CPTP. 

For instance, when using the Choi operator method, we know that the CP condition is satisfied if $J^{{\phi}} \geq 0$. Hence, one could start with some basis of trace-preserving and trace-annihilating equivariant maps $\{J^{\phi_j}\}$, linearly combine them as $J(\vec{x})=\sum_j x_j J^{\phi_j}$, and optimize the set of real parameters $\vec{x}=\{x_j\}$ under the constraint that the eigenvalues of $J(\vec{x})$ are non-negative.  The latter will  yield a  region of feasible equivariant quantum channels (see Section~\ref{sec:examples} for an example). Note, that during the  optimization of $\vec{x}$ the update rule might take us outside of the equivariant region, in which case one needs to project back  to the feasible space. We further discuss how such projection can be performed in Appendix~\ref{app:compile}. Finally, we note that while it might not be directly obvious how to implement the ensuing channel, one can always transform it into an implementable sequence of gates, acting on a potentially larger space,  via compilation techniques~\cite{bharti2020quantum,chong2017programming,haner2018software,sharma2019noise}. Here we also remark that in many cases, particularly when the maps act on large-dimensional spaces, finding  the eigenvalues of $J(\vec{x})$ might may be quite difficult. For these scenarios, one can simply optimize over  a subset of equivariant channels (i.e., maps that are already CPTP)  which can be found via twirling. Here we are guaranteed that any convex combination of equivariant channels will lie in the feasible region since CPTP channels form a convex set~\cite{dadriano2004extremal}. 

An alternative approach to constructing equivariant channels is via the Stinespring dilation picture~\cite{wilde2013quantum}. In this case we use the fact that any channel can be written as a unitary operation on a larger space, i.e., 
\begin{equation}
    \phi(\rho) = \Tr_E[U(\rho\otimes \dya{e})U^\dagger],
\end{equation}
where $\ket{e}\in\HC^E$ is a fixed reference state on an environment Hilbert space $\HC^E$, and where $\Tr_E$ denotes the trace over $\HC^E$.  If $ U (R^\text{in}(g)\otimes \id_{\dim(\HC^E)}) = (R^\text{out}(g)\otimes R^{(E)}(g))U,\ \forall g \in G$, then $\phi$ is a $(G,R^{\text{in}},R^{\text{out}})$-equivariant channel. Here we can use any of the tools previously discussed to find and parametrize $U$. This approach has the advantage that by fixing the dimension of the environment, we can look for channels of small Kraus rank which are easier to implement in practice.

\begin{figure}
    \centering
    \includegraphics[width=.9\columnwidth]{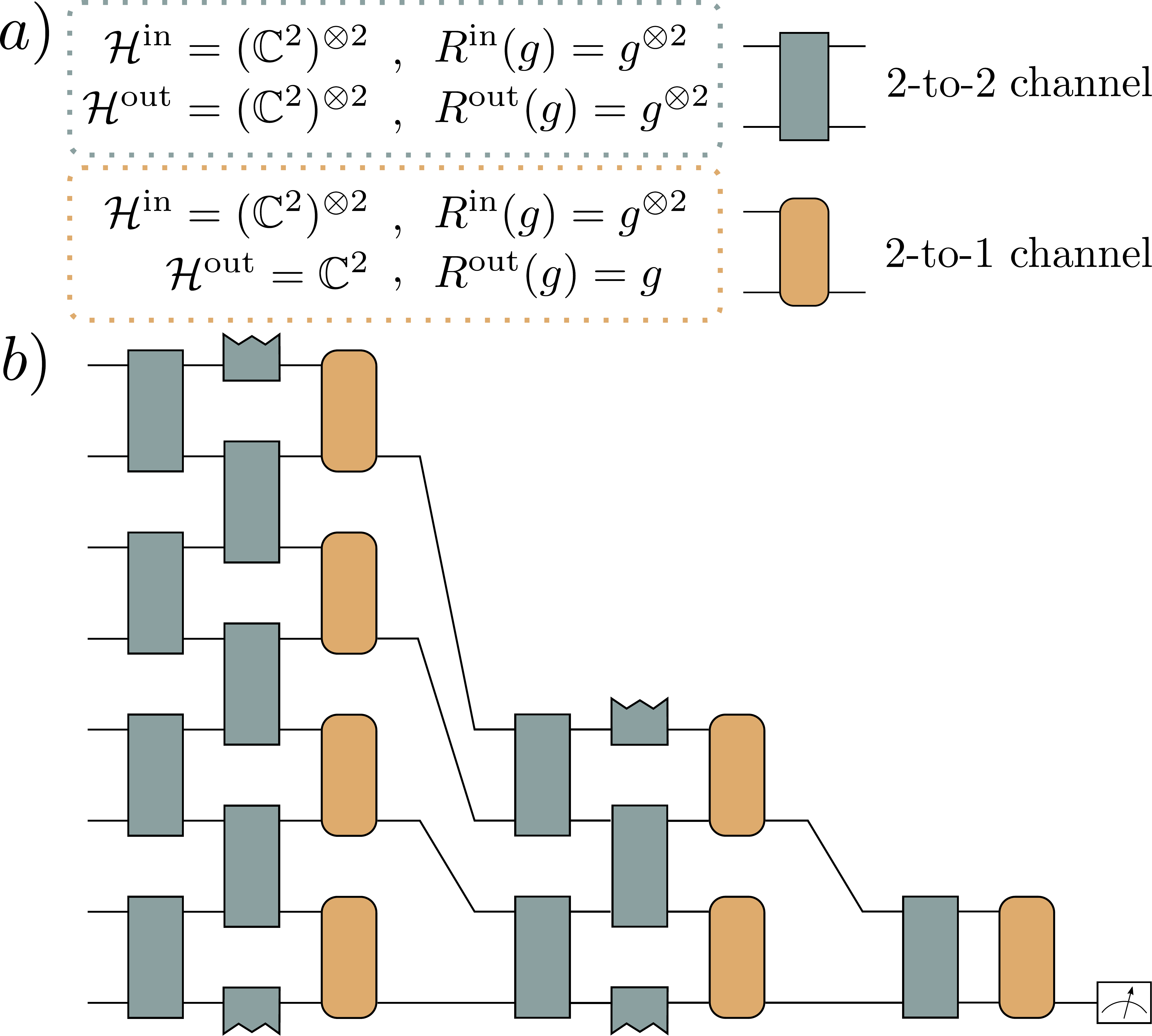}
    \caption{\textbf{$\mbb{SU}(2)$-equivariant QCNN.} a) We consider the problem of building  2-to-2 standard equivariant channels  and 2-to-1 equivariant pooling channels. In the figure we present the respective input and output Hilbert spaces, as well as the input and output representation. b) In an $\mbb{SU}(2)$-equivariant QCNN, we alternate between  2-to-2 channels acting on neighbouring qubits and 2-to-1 equivariant pooling channels which reduce the feature space dimension. }
    \label{fig:qcnn-main}
\end{figure}

\section{Applications}\label{sec:examples}
In this section we exemplify the applicability of the methods presented in the previous sections to design EQNNs.

\subsection{$\mbb{SU}(2)$-equivariant QCNN}\label{sec:su2_qcnn}
As a first practical application of the present framework we propose to generalize standard QCNNs~\cite{cong2019quantum} to group-equivariant QCNNs. We start by recalling that QCNN have been successfully implemented for error correction, quantum phase detection~\cite{cong2019quantum, maccormack2020branching}, image recognition~\cite{franken2020explorations}, and entanglement detection~\cite{schatzki2021entangled}. QCNNs exhibit several key features that make them promising architecture for the near-term, such as having a shallow depth or not exhibit barren plateaus~\cite{pesah2020absence}. Despite these advantages, standard QCNNs need not respect the symmetries of a given task. In what follows, we will show how one can design equivariant layers for QCNNs, thus promoting them to group-equivariant QCNNs.

We consider problems where the symmetry group is $\mbb{SU}(2)$. This symmetry appears in tasks where the data arises from certain spin chain models~\cite{parkinson2010introduction,elben2020many,huang2021provably} and in tasks related to entanglement measures~\cite{horodecki2009quantum,beckey2021computable,schatzki2022hierarchy}.  Moreover, we consider the case where the data correspond to $n$-qubit quantum states, and where the input representation of $G = \mbb{SU}(2)$ is $R^{\text{in}}(g) = R_{\text{tens}}(g)= g^{\otimes n}$. For ease of implementation on quantum hardware we restrict ourselves to channels with locality constraints (see Lemma~\ref{lemma:local}). That is, as illustrated in Fig.~\ref{fig:qcnn-main} we want to build an equivariant QCNN that composed of alternating layers of 2-to-2 standard equivariant channels acting on neighbouring qubits and 2-to-1 equivariant pooling channels (for completeness the 1-to-2 qubit equivariant embedding maps are presented in Appendix~\ref{app:su2}). Of course, this choice of architecture trades locality at the cost of expressibility, as there may be more general equivariant channels on $n$ qubits. However, the success of models with locality constraints~\cite{cong2019quantum,pesah2020absence} suggests this may be an interesting regime regardless.

\subsubsection{2-to-2 layers via Choi operator}

Let us commence by studying 2-to-2-qubit maps via the Choi operator method. Since the input and output representations are $R(g) = g^{\otimes 2}$, the Choi operator must commute with the representation $(g^*)^{\otimes 2}\otimes g^{\otimes 2}$ (see Eq.~\eqref{eq:cho_expanded}). As $\mbb{SU}(2)^*$ shares the same irrep structure as  $\mbb{SU}(2)$, we can find that the Choi operator for any completely positive $\mbb{SU}(2)$-equivariant map takes the form
\begin{equation}
    J^\phi = (\id_{5}\otimes A) \oplus (\id_{3}\otimes B) \oplus C,
    \label{eq:2-to-2}
\end{equation}
 where $A$ is a non-negative scalar and $B$ and $C$ are complex positive semidefinite matrices of dimensions 3 and 2, respectively. Thus, the space of such CP  equivariant maps is $1^2+2^2+3^2=14$ dimensional.

In the special case of 2-to-2 equivariant \textit{unitary} layers, where $\NC_{\thv_l}^l:(\mathbb{C}^2)^{\otimes 2}\rightarrow (\mathbb{C}^2)^{\otimes 2}$ and $\NC_{\thv_l}^l(\rho)=U_l(\theta_l)\rho U_l(\theta_l)\ad$, we know that if $U_l(\theta_l)=e^{-i \theta_l H_l}$, it suffices to use equivariant generators, that is, such that $[H_l,g^{\otimes 2}]=0$ for all $g\in \mbb{SU}(2)$. 
Here we can use the Schur-Weyl duality~\cite{goodman2009symmetry,larocca2022group}, which states that the only possible equivariant operators are $\id$ and ${\rm SWAP}$, which correspond to the two elements of the qubit-permutational representation of $S_2$. Without loss of generality, we can choose $H_l={\rm SWAP}$ so that $U_l(\theta_l)=e^{-i\theta_l {\rm SWAP}}$. Following Lemma~\ref{lemma:local}, we know that if we compose these two-qubit equivariant unitaries as in Fig.~\ref{fig:qcnn-main}, the result will be an $n$-qubit equivariant unitary.

\subsubsection{2-to-1 layers via Nullspace}

Next, let us focus on finding 2-to-1-qubit channels using the nullspace approach. Since $\mbb{SU}(2)$ is a Lie group, we will work at the level of the generators of its Lie algebra, $\mf{su}(2)=\spn \{X,Y,Z\}$. Given the    representations $g^{\otimes 2}$ and $g$, the associated basis representations of the algebra are $\{\id \otimes X+X\otimes\id, \id \otimes Y+Y\otimes \id, \id \otimes Z+Z\otimes \id\}$ and $\{X,Y,Z\}$. Thus, one needs to simultaneous solve for the nullspace of the following matrices
\begin{equation}
\begin{aligned}
     M_X &=  \overline{\text{ad}_{\rm{IX+XI}}}^\top  \otimes \id_2 -\id_4\otimes  \overline{\text{ad}_X}\,,\\
     M_Y &=  \overline{\text{ad}_{\rm{IY+YI}}}^\top\otimes \id_3 - \id_4\otimes\overline{\text{ad}_Y}\,,\\
     M_Z &= \overline{\text{ad}_{\rm{IZ+ZI}}}^\top \otimes \id_3 - \id_4\otimes\overline{\text{ad}_Z}\,.
\end{aligned}
\end{equation}

Solving, we find five superoperators that form a basis for 2-to-1 qubit $(\mbb{SU}(2), g^{\otimes 2}, g)$-equivariant maps. These are
\begin{align}
    \phi_1(\rho) & = \Tr[\rho]\frac{\id}{2},\quad
    \phi_2(\rho) = \Tr[\rho\SWAP]\frac{\id}{2},\label{eq:su2-equiv-maps}\\
    \phi_3(\rho) & = \Tr_A[\rho],\quad 
    \phi_4(\rho) = \Tr_B[\rho],\notag\\
    \phi_5(\rho) & = \sum_{ijk=1}^3\Tr[\rho \sigma_i\sigma_j]\epsilon_{ijk}\sigma_k\notag.
\end{align}

\begin{figure}
    \centering
    \includegraphics[width=.7\columnwidth]{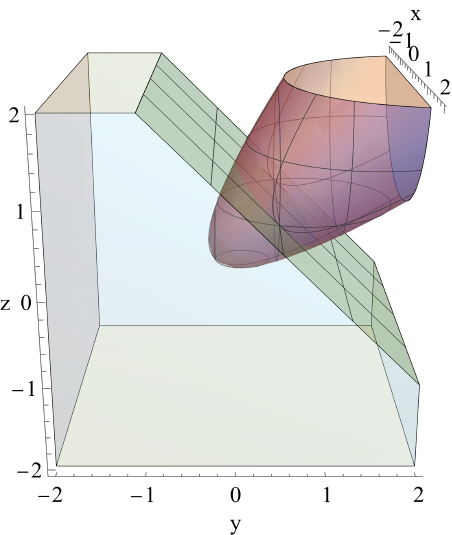}
    \caption{\textbf{Region of parameter space leading to CPTP channels.} Using the nullspace method we can find a basis for all 2-to-1 $(\mbb{SU}(2), g^{\otimes 2},g)$-equivariant pooling maps. These can then be linearly combined to form a general parametrized equivariant map as in Eq.~\eqref{eq:linear-comb}, and we find in~\eqref{eq:feasible} the region  in parameters space leading to CPTP channels. Here we depict said region as the volume of the hyperbole (red) below the plane (green). }
    \label{fig:CPTP-region}
\end{figure}

Notably, in addition to expected $\mbb{SU}(2)$-equivariant maps such as trace, partial traces, and ${\rm SWAP}$-measurement, we identify a potentially interesting new equivariant map $\phi_5(\rho)$ that is dubbed the \textit{cross-product map} (we further study its properties in Appendix~\ref{app:su2}). Note that $\phi_1, \phi_3, \text{ and } \phi_4$ are trace preserving while $\phi_5$ is trace-annihilating. One can also verify that $\phi_2$ may non-trivially alter trace. As the only map that may do so, we can drop it from our basis set for being non-physical. To continue and find the set of equivariant quantum channels, we first make a modification to our basis set. In the Pauli basis we have  $\phi_1 \leftrightarrow 2 \ket{\id}\rangle\langle\bra{\id,\id}$, and it is easy to see that both $\phi_3$ and $\phi_4$ also contain this term. Thus, we can remove it, leaving trace-annihilating versions of partial trace, which we will denote by $\phi_3'$ and $\phi_4'$. Thus, any TP map must take the form
\begin{align}\label{eq:linear-comb}
   \phi(x,y,z)=\phi_1 + x \phi_5 + y \phi_3' + z \phi_4',
\end{align}
where the coefficients are real numbers. It remains to find the region such that this channel is CP. This can be done via the Choi operators of these channels. That is, we would like to find 
\begin{align}
    \{x,y,z \in \mathbb{R}: J^{\phi_1}+xJ^{\phi_5}+yJ^{\phi_3'}+zJ^{\phi_4'} \geq 0\}.
\end{align}
Note that the coefficients here must be real numbers for the Choi operator of the sum to be positive (as the Choi operators in the sum are linearly independent). Requiring the eigenvalues of this linear combination to be non-negative yields the feasible region
\begin{align}\label{eq:feasible}
\begin{split}
     x,y,z: y + z \leq 1\,,\quad \text{ and }\\ y + z \geq \sqrt{3x^2 + 4(y^2 - yz+z^2)}-1\,.
     \end{split}
\end{align}
This region is illustrated in Fig.~\ref{fig:CPTP-region}.  Here we note that as the set of equivariant channels is convex, this feasible parameter region is a convex subset of $\mathbb{R}^3$.

A crucial aspect to note is that when training the $\mbb{SU}(2)$-equivariant QCNN, one can directly train over the coefficients $x$,$y$, and $z$ of each pooling channel $\phi(x,y,z)$ of the form in Eq.~\eqref{eq:linear-comb}.  When training an equivariant QCNN, e.g. using gradient descent, we will obtain parameter updates $(x^{(t+1)},y^{(t+1)},z^{(t+1)}) \leftarrow (x^{(t)},y^{(t)},z^{(t)}) - \alpha D_t((x^{(t)},y^{(t)},z^{(t)}))$. To ensure that the operations remain physical, one would  continually solve the projection
at each iteration. This can be turned into the convex optimization  problem
\begin{equation}
\begin{aligned}
    \min_{x,y,z} & \lVert (x^{(t+1)},y^{(t+1)},z^{(t+1)}) - (x,y,z) \rVert^2,\\
    &\text{subject to }  \text{Eq.~\eqref{eq:feasible}}
\end{aligned}
\label{eq:projection}
\end{equation}
over a convex domain (see Appendix~\ref{app:compile}).

\subsection{Various examples and physical considerations}\label{sec:scattered}

In this section we present additional applications of the methods detailed in Sec.~\ref{sec:framework}. In particular these are now applied to discrete groups, and motivated from practical problems.

\subsubsection{$\mathbb{Z}_2 \times \mathbb{Z}_2$-equivariant layers}

Here we consider problems with $\mathbb{Z}_2 \times \mathbb{Z}_2$ symmetry which are common in spin chain models such as the $S=1$ Haldane chain~\cite{ haldane1983nonlinear, cong2019quantum, parkinson2010introduction}, or in classical data on the two-dimensional plane \cite{meyer2022exploiting}. 

We start with a task on $n$-qubits, where the representation of $\mathbb{Z}_2 \times \mathbb{Z}_2$ is given by
\begin{equation}\nonumber
\begin{split}
R^{\text{in}}(e,e)=\id_{2^{n}}\,,\quad &R^{\text{in}}(\sg,e)=\prod_{i\in O} X_i \otimes \id_{2^{n/2}} \,, \\
R^{\text{in}}(e,\sg)=\id_{2^{n/2}} \otimes \prod_{i\in E} X_i\,,\quad & R^{\text{in}}(\sg,\sg)=\prod_{i\in E\cup O} X_i\,.
\end{split}
\end{equation}
Here,  the sets $O$ and $E$, respectively contain the odd and even qubit labels. From the previous, we are interested in finding all equivariant standard unitary channels where $R^{\text{in}}=R^{\text{out}}$. Since the group is small, we can readily employ the twirling method to determine the set of all equivariant generators. The set of all such generators forms a subalgebra of the unitary algebra $\mf{u}(2^n)$, that we denote $\mf{u}^{\mbb{Z}_2\times\mbb{Z}_2}(2^n)$, and which is given by
\small
\begin{equation}
\begin{split}
    \mf{u}^{\mbb{Z}_2\times\mbb{Z}_2}(2^n) =& \spn_{\mbb{R}}i\Big\{\text{Pauli strings with even $\#$ of Y's and Z's} \\ &\text{on both $E$ and $O$} \Big\}.
\end{split}\nonumber
\end{equation}
\normalsize
Next, let us find pooling channels that are equivariant with respect to $\mbb{Z}_2\times \mbb{Z}_2$.  Similarly to the $\mbb{SU}(2)$ case previously considered, the representations of the group act locally, meaning that we can again consider local 2-to-1 equivariant pooling maps and later combine them into a global equivariant channel (see Lemma~\ref{lemma:local}). On any pair of neighbouring qubits, the input representations are $\{\id\otimes\id, \id \otimes X, X\otimes \id, X\otimes X\}$ and we set the output representation acting on a single qubit to be $\{\id, \id, X, X\}$. Note that the output representation is not faithful. As $\lvert \mbb{Z}_2 \times \mbb{Z}_2 \rvert =4$ we can again readily apply the twirling procedure. To do so, we begin with a basis for trace-preserving maps. That is,  matrices in the Pauli-string basis of the form $ 2\ket{\id}\rangle\langle\bra{\id\otimes\id}+\ket{P_\text{out}}\rangle\langle\bra{P_\text{in}}$ where $P_\text{out} \neq \id$. A simple counting argument reveals that there are 48 such matrices. By twirling all 48 maps, and extracting a linearly independent set, we find the  following 13-element basis for 2-to-1 equivariant projective poolings:
\begin{align}
    \varphi_1(\rho) = \Tr[\rho]\frac{\id}{2},\ & \varphi_2(\rho) = \Tr[\rho]X\\
    \varphi_3(\rho) = \Tr[(\id\otimes X)\rho]X,\ &
    \varphi_4(\rho) = \Tr[(X\otimes \id)\rho]X\\
    \varphi_5(\rho)=\Tr[(X\otimes X)\rho]X,\ & \varphi_6(\rho) = \Tr[(Y\otimes \id)\rho]Y\\ \varphi_7(\rho)=\Tr[(Y\otimes X)\rho]Y,\ &
    \varphi_8(\rho) = \Tr[(Z\otimes \id)\rho]Y\\ \varphi_9(\rho)=\Tr[(Z\otimes X)\rho]Y,\ & \varphi_{10}(\rho) = \Tr[(Z\otimes \id)\rho]Z\\ \varphi_{11}(\rho)=\Tr[(Z\otimes X)\rho]Z,\ & \varphi_{12}(\rho) = \Tr[(Y\otimes \id)\rho]Z\\ \varphi_{13}(\rho)=\Tr[(Y\otimes X)\rho]Z\,.
\end{align}

\subsubsection{$\mbb{Z}_n$-equivariant layers}

We proceed now to analyze a problem with $\mathbb{Z}_n$-symmetry, whose representation cyclically shifts qubits as $R(g^t) \bigotimes_{j=1}^n \ket{\psi_j}=\bigotimes_{j=1}^n \ket{\psi_{j+t \mod n}}$. Such symmetry arises naturally in condensed matter problems with periodic boundary conditions~\cite{sachdev2007quantum,cerezo2017factorization,caro2021generalization}.

Let us start by finding all equivariant standard unitary maps where $R^{\text{in}}=R^{\text{out}}$. Since the group is small, we opt for the twirling approach. For the sake of implementability, we will seek channels composed of one- and two-qubit gates. We can readily see that the twirl of a single qubit generator such as $X_1$, leads to the sum of single-qubit operators $\TC[X_1]=\frac{1}{n}\sum_{j=1}^n X_j$. Similarly, twirling a two-body generator  will lead to a sum of two-body generators (e.g., $Z_1Z_2\xrightarrow[]{\TC}\sum_{j}Z_iZ_{j+1}$). Notably, these generators lead to equivariant unitaries of the form $U_l=e^{-i \theta_l \frac{1}{n}\sum_{j=1}^n X_j}=\prod_{j=1}^n e^{-i \theta_lX_j/n}$. The previous shows a crucial implication of outer symmetries (such as $Z_n$): in many cases, equivariance in unitary layers can be achieved by \textit{correlating} parameters of local gates within a layer~\cite{volkoff2021large, sauvage2022building}. 

Note that a necessary condition for the previous to hold is that all the terms in the twirled operator must be mutually commuting. One can see however, that the twirl of  $Z_1Y_2$ leads to $\sum_{j}Z_jY_{j+1}$, which is a global operator whose terms are non-commuting and which can be challenging to implement on near-term devices. An alternative here is to employ a randomized method. First, we construct unitaries $U_{l,O} = \prod_{j\in O} e^{-i \theta_l} Z_jY_{j+1}$  and  $U_{l,E} = \prod_{j\in E} e^{-i \theta_l} Z_jY_{j+1}$ (where we recall that $O$ and $E$ respectively contain the odd and even qubit labels), which are  $\mathbb{Z}_{n/2}$-equivariant. Then, to achieve $\mathbb{Z}_n$-equivariance, we can apply either $U_{l,O}$ or $U_{l,E}$ at random and with equal probability, effectively performing the quantum channel $\rho \rightarrow (U_{l,O}\rho U_{l,O}\ad + U_{l,E} \rho U_{l,E}\ad)/2$. This channel can be readily shown to be $\mathbb{Z}_n$-equivariant.

This trick of randomly applying local channels can also be applied to equivariant layers that have different numbers of qubits in the input and output. For example, a $\mathbb{Z}_n$-equivariant projection layer  that reduces the number of qubits from $n$ to $n/2$ is $\Phi_{\mathbb{Z}_n}: \rho \rightarrow (\operatorname{Tr}_{\text{odd}} \rho + \operatorname{Tr}_{\text{even}} \rho)/2$. Observe that the output representation has a non-trivial kernel isomorphic to $\mathbb{Z}_2$, as the number of qubits is halved. One can readily extend these ideas to projecting pooling layer over other outer symmetry groups $G$, such as $S_n$-equivariance: $\Phi_{S_n}: \rho \rightarrow \left(\sum_{S} \operatorname{Tr}_{S} \rho\right)/{n \choose n/2}$, where $S$ are uniformly random subsets of $n/2$ qubits. The Hoeffding's bound implies that only $O(\log |G|)$ samples are needed for this method to converge to within a specified error bound. In a sense, this is similar to the dropout regularization technique in neural networks~\cite{srivastava2014dropout}.

\section{Numerical experiments}\label{sec:numerics}
In this section, we numerically compare the performance of $\mbb{SU}(2)$-equivariant QCNN constructed in Sec.~\ref{sec:su2_qcnn} against a problem-agnostic QCNN in a quantum phase classification task.
Similar to the classical problem of assigning the correct labels to images, the task
of classifying quantum phases of matter can be carried
out in a supervised setting and provides a natural
playground to study the efficiency of equivariant quantum
learning model.

\subsection{Bond-Alternating Heisenberg Model}
The 1-D XXX Heisenberg model describes the behavior of a one-dimensional chain of spin-1/2 particles coupled through the standard Heisenberg interaction Hamiltonian between nearest neighbors:
\begin{equation}\label{eq:heisenberg_term}
    H = \sum_{i} \sum_{k=x,y,z} J_k S_i^k S_{i+1}^k \,,
\end{equation}
where $S_i$ is the spin operators for the $i$-th spin, with $S = (S^x,S^y,S^z) = \frac{1}{2}(X,Y,Z)$ and $J_x = J_y = J_z = J$. The bond-alternating XXX Heisenberg model is a generalization of the regular XXX Heisenberg model, in which the exchange coupling constant alternates between two different values $J_1$ and $J_2$:

\begin{equation}\label{eq:alternating_xxx}
    H = J_1 \sum_{i \, \text{even}} S_{i} \cdot S_{i+1} + J_2 \sum_{i \, \text{odd}} S_i \cdot S_{i+1} \,.
\end{equation}

\begin{figure}[ht!]
    \centering
    \includegraphics[scale=0.6]{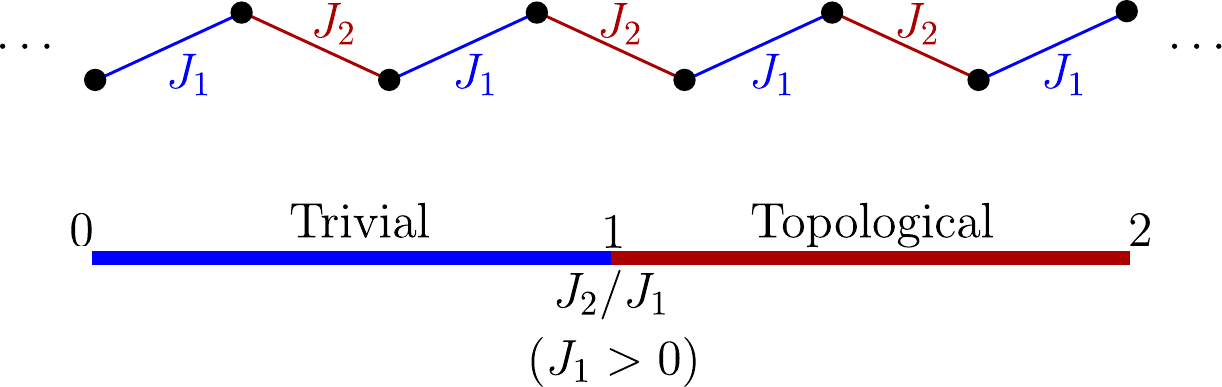}
    \caption{\textbf{1-D Bond-Alternating XXX Heisenberg Model} (Eq.~\eqref{eq:alternating_xxx}) and its phase diagram in terms of exchange coupling constants $J_{1,2}>0$.}
    \label{fig:alt_bond_phasediag}
\end{figure}

We consider the model described by Eq.~\eqref{eq:alternating_xxx} with open boundary conditions and with both the couplings in the ferromagnetic regime, {\ie} $J_{1,2}>0$.
In this case, a \textit{trivial} phase and a \textit{topologically protected} phase are defined by the expectation value of the partial reflection many-body topological invariant~\cite{pollmann2012detection,elben2020many,huang2021provably}. The quantum phase transition between these two phases occurs at a critical value of the bond alternation parameter $\alpha = J_2/J_1$. When $\alpha < 1$, the system is in the trivial phase, otherwise, the system is in the topologically protected phase.

The Hamiltonian can be readily seen to possess an $\mbb{SU}(2)$ symmetry.
The symmetry extends to the whole system through the tensor product representation $R^{\text{in}}(g) = R_{\text{tens}}(g)= g^{\otimes n}$. 
This is also a symmetry of the phase labels, as quantum phases are global properties of the ground states of the model, and symmetries of the Hamiltonian are also symmetries of the ground space. Indeed if $[H, g^{\otimes n}] = 0$ and $\ket{\psi}$ is a ground state, then $g^{\otimes n}\ket{\psi}$ is also a ground state.
Thus, a quantum phase classifier can be endowed with this inductive bias using the $\mbb{SU}(2)$-equivariant quantum maps from Sec.~\ref{sec:su2_qcnn}. Another inductive bias that we can utilize is translation symmetry: shifting the qubits by two sites leaves the model unchanged. One way to exploit this is parameter sharing within each layer of the $\mbb{SU}(2)$-equivariant QCNN as discussed in Sec.~\ref{sec:scattered}.

\subsection{Defining the learning models}


We now describe the learning model in more details. We use the $\mbb{SU}(2)$-EQCNN architecture in Fig.~\ref{fig:qcnn-main}, where each standard (convolution) layer consists of two brickwork (sub)layers of two-qubit $\mbb{SU}(2)$-equivariant gates of the form $U(\theta)=e^{-i\theta {\rm SWAP}}$. Due to translation symmetry, we further enable parameter sharing of two-qubit gates within each of such sublayer.
This leads to having two parameters for each standard convolution layer. If needed, we can repeat the standard layers more than once before applying the pooling layer.

Finally, we choose an equivariant observable operator $\hat{O}$ to measure at the end of the EQCNN. As discussed in the previous sections, the final equivariant measurements belong to the commutant of the output representation $R^{\text{out}}$ of the last layer of the EQCNN.  Now, since we need two outputs to label the two phases $y_{\text{trivial}}=1$ and $y_{\text{topological}}=0$, it is  tempting to end the EQCNN with a binary measurement on $m=1$ qubit. However, from the discussion in Sec.~\ref{sec:su2_qcnn} we know that the commutant of the defining representation of $\mbb{SU}(2)$ over a single qubit is the trivial set $\mf{comm}(R_\text{natural})=\{\id\}$.
Thus we choose a EQCNN that ends with $m=2$ such that the commutant of $g^{\otimes 2}$ contains the nontrivial element SWAP. Conveniently, SWAP has two eigenvalues $\pm 1$ so that we can bind, say, the +1 outcome to $y_{\text{trivial}}$ and the -1 one to $y_\text{topological}$.
We adopt this strategy, with a little modification to have the output of the EQCNN, that we will indicate as $f_{\thv} (\rho)$, take values in $[0,1]$. Namely, we define
\begin{equation}\label{eq:su2_classifier}
    f_{\thv} (\rho) =\frac{\Tr [\phi_{\thv}(\rho) SWAP] + 1}{2} \,.
\end{equation}
Where $\phi_{\thv}$ denotes the EQCNN that outputs two qubits in the last layer.
Then, we assign the predicted phase label to any input state $\rho$ as
\begin{equation}\label{eq:su2_predictions}
    y_{\thv}(\rho) = \left\{ 
    \begin{aligned}
        \text{trivial} \quad &\text{if} \quad f_{\thv}(\rho) > \tau \\
        \text{topological} \quad &\text{if} \quad f_{\thv}(\rho) < \tau \\
    \end{aligned}
    \right. \,,
\end{equation}
for some trainable threshold value $\tau$ that is initialized to be $\tau=0.5$.

We test this EQCNN architecture against a QCNN with no inductive biases. In particular, the QCNN whose standard layers are parametrized circuits inspired by the standard hardware efficient ansatz (HEA)  \cite{kandala2017hardware}, whereas the pooling layers consist of simple alternate partial traces, \ie, at each pooling operation we discard half of the qubits. The classification will then proceed as for the $\mbb{SU}(2)$-EQCNN, with a SWAP measurement and phase assignment described in Eq.~\eqref{eq:su2_predictions}.
We dub the non-equivariant QCNN as HEA-QCNN.

\subsection{Training procedure}\label{sec:su2_training_loop}

We use the standard ML pipeline of supervised learning.
\begin{enumerate}
    \item We collect a training dataset $\mathcal{D}_\text{train}^{N_T}$, where $N_T$ is the size of the dataset, by choosing some representative values of the parameter $J_2$ while always keeping $J_1=1$ and then analytically computing the ground states $\ket{\psi}_g^{J_2/J_1}$ of the Hamiltonian in Eq,~\eqref{eq:alternating_xxx}. Knowing the phase diagram of the alternating model, which is shown in Fig.~\ref{fig:alt_bond_phasediag}, especially that the critical value at which the transition happens $\alpha=J_2/J_1 = 1$, we can then associate these states with their true labels $y\in \{0, 1\}$. Particularly, we try training dataset made of $N_T = (2,4,6,8,10,12)$ ground states, always distributed homogeneously in the range $J_2/J_1 \in [0, 1]$. For example for $N_T = 2$ we use $\mathcal{D}_\text{train}^{2} = \{(\ket{\psi}_g^{0.25}, 1), (\ket{\psi}_g^{0.75}, 0)\}$.
    
    \item We initialize the learning model at hand, equivariant or not, with random parameters ${\thv}$.
    
    \item We select an optimizer for the learning model. In our case we always use ADAM \cite{kingma2014adam}, the golden standard of gradient-based optimization in ML.
    
    \item For a number of \textit{epochs} $E$, we divide the training dataset $\mathcal{D}_\text{train}^{N_T}$ in batches of size $n_{\text{batch}}=2$. For each batch, the training states $\ket{\psi_i}$ are processed by the model to output the predicted label $y_{\thv}(\ket{\psi})$ and the mean squared error loss function is computed by comparing the predictions to the real labels $y_i$
    \begin{equation}
        L_{\thv} = \frac{1}{n_{\text{batch}}}\sum_{i=1}^{n_{\text{batch}}} (y_{\thv}(\ket{\psi_i})- y_i)^2 \,.
    \end{equation}
    We then compute the gradient of $L_{\thv}$ and use the optimizer to update the model's parameters. The goal is to minimize $L_{\thv}$.
    
    \item The QCNN outputs for the training states are used to update the threshold $\tau$. Particularly, only the two training points that are closer to the critical value $\alpha = 1$ are considered, and the threshold value is set to the average of the corresponding outputs.
    
    \item As an additional figure of merit for the training we keep track of the prediction accuracy of the model.
    
    
    \item At the end of the last epoch, we let the model predict the labels of the whole test dataset, and we compute its final accuracy as a measure of the goodness of the training. Then, we also plot the predicted phase diagram to get a visual proof of the performance of the model.
    
\end{enumerate}

\subsection{Training results}
We are now ready to illustrate the results of our numerics. First thing first, we must state that we have not been able to train the EQCNN when using the general 2-to-1 pooling layers described in Sec.~\ref{sec:su2_qcnn}, as the projection step (Eq.~\eqref{eq:projection}) onto the feasible CPTP region seems to cause instability in the optimizing procedure. We leave a full numerical analysis of equivariant quantum learning models to a future upcoming work, and here focus on the more simple tracing pooling operations. That is, the EQCNN architecture is still the one depicted in Fig.~\ref{fig:qcnn-main}, but the pooling operations are just 2-to-1 partial trace channels, corresponding to the solution $\phi_3$ in Eq.~\eqref{eq:su2-equiv-maps}.
In other words, the $\mbb{SU}(2)$-EQCNN and HEA-QCNN use the same pooling layers (but still different convolution layers -- equivariant versus HEA). The training results are illustrated in Figs.~\ref{fig:predicted_phasediags},~\ref{fig:eqcnn_vs_heaqcnn}.\\

\begin{figure*}[ht!]
    \centering
    \includegraphics[scale=0.7]{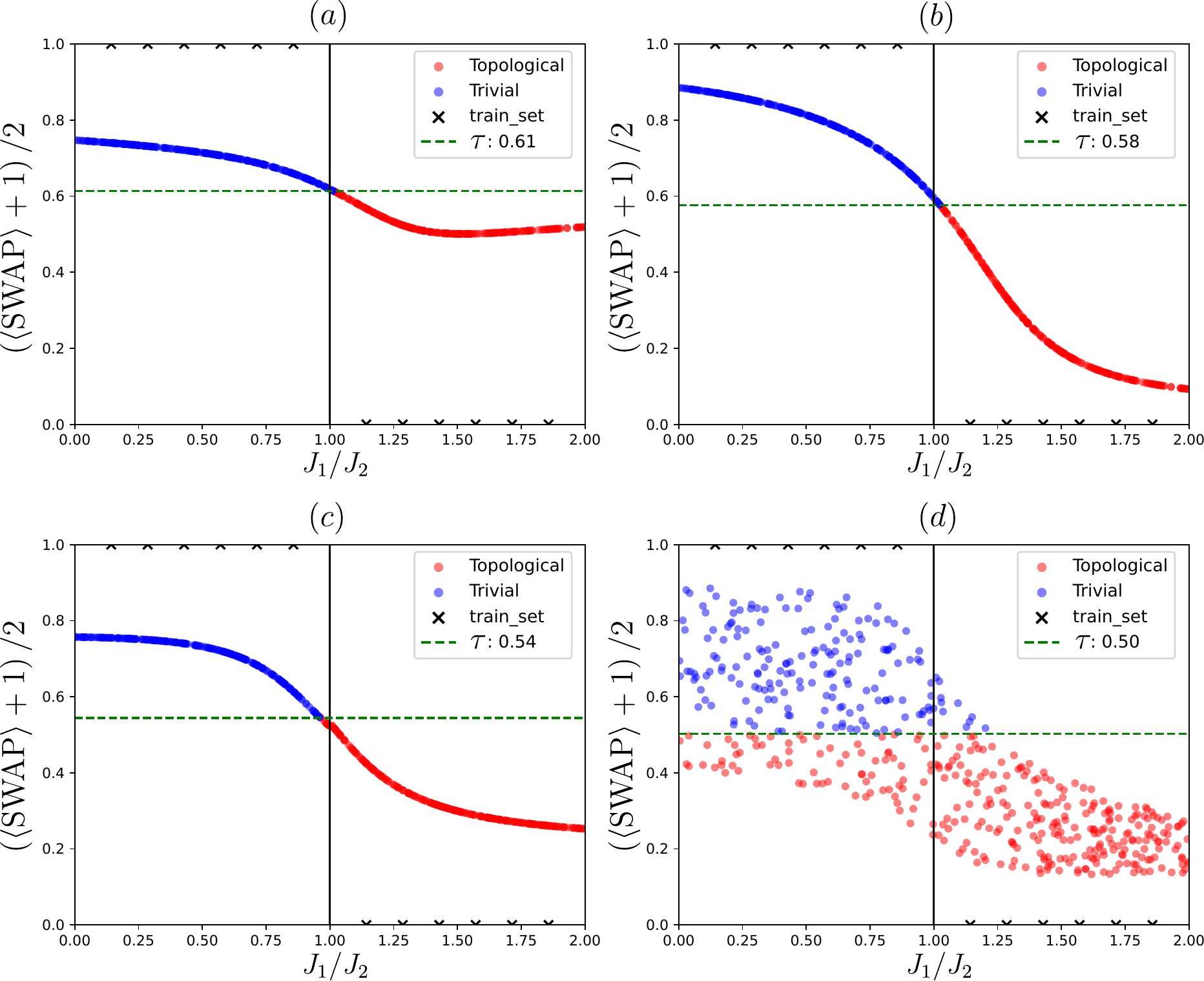}
    \caption{\textbf{Predicted Phase Diagrams.} The four panels show the phase diagram of the 1D Bond Alternating XXX Heisenberg Model for system sizes of $N=12$ and $N=13$ qubits as reconstructed by a trained $\mbb{SU}(2)$-EQCNN or HEA-QCNN. Particularly, each panel shows the QCNN output when it is tested against a dataset of 500 homogeneously distributed ground states. States whose output is above (below) the optimal threshold $\tau$ (green dashed line) are colored in blue (red) and classified as belonging to the trivial (topological) phase. The training points are shown as black crosses. The vertical solid black line is the theoretical critical value $J_2/J_1=1$. The configurations leading to the panels are the following. $(a)$: $\mbb{SU}(2)$-EQCNN, $N=12$, 60 trainable parameters, 12 training points; $(b)$: HEA-QCNN, $N=12$, 63 trainable parameters, 12 training points; $(c)$: $\mbb{SU}(2)$-EQCNN, $N=13$, 66 trainable parameters, 12 training points; $(d)$: HEA-QCNN, $N=13$, 66 trainable parameters, 12 training points. Details about the training procedure are given in the main text.}\label{fig:predicted_phasediags}
\end{figure*}

We considered system sizes ranging from $N=6$ to $N=13$, and trained both the EQCNN and the HEA-QCNN according to the training loop described in Section~\ref{sec:su2_training_loop} for a fixed number of training epochs $E=750$. Since the two architectures are very different, and the EQCNN, as opposed to the HEA-QCNN, uses parameter sharing, in order to have a fair comparison we decided to stack multiple standard layers before each pooling one in the EQCNN, in such a way as to have a similar amount of training parameters for both the learning models. In Fig.~\ref{fig:predicted_phasediags} we show the predicted phase diagrams for $N=12$ and $N=13$. The thing that immediately stands out is the $(d)$ panel of that figure. While the other three plots basically showcase similar behavior, with the QCNN at hand being able to efficiently separate the two phases of the alternating model, panel (d) shows a cloudy behavior of the HEA-QCNN predictions, as it assigned different phases even for states with similar parameters $\alpha$. This is in sharp contrast with the trained EQCNN (panel $(c)$) that successfully learned to classify the two phases with excellent accuracy, demonstrating the advantage of equivariant models.

Interestingly, for the $N=12$ qubit case (panels $(a)$ and $(b)$ in Fig.~\ref{fig:predicted_phasediags}), the EQCNN does not significantly outperform the HEA-QCNN. This is due to fact that there is actually no need for equivariance in that case!
Indeed, equivariance is meant to enhance the performance of learning models that deal with labels invariant under some symmetry group, but this invariance should \textit{not} come from the invariance of the input states themselves. Think yet again of the classical problem of classifying images of cats and dogs, the labels, {\ie} the semantic meaning of the images, are invariant if we translate the images, but the images themselves are not translation-invariant. On the other hand, if we translate images that are full of either black or white pixels, instead of showing cats and dogs, the labels (the colors) of the images are translation-invariant simply because the images do not change. 
This is what is happening in our case. We have already discussed that if a Hamiltonian $H$ is symmetric under a group $G$, any unitary representation of it $U_G$ leaves the ground space unchanged. For non-degenerate ground space, this means the unique ground state is invariant under the group action, in analogy to the above black/white image example, and thus equivariant learing models are not needed.
For degenerate ground states, {\ie} when the Hamiltonian symmetry is broken, the symmetry group action does change the ground states nontrivially and rotate them within the ground space. Degenerate ground states are akin to images of cats and dogs, and as such equivariance can finally shine.
Indeed, the alternating model is degenerate for odd system sizes, while for even system sizes the ground state is unique.
This explains the different behaviors shown in Fig.~\ref{fig:predicted_phasediags} between $N=12$ and $N=13$.\\

\begin{figure*}
    \centering
    \includegraphics[scale=0.7]{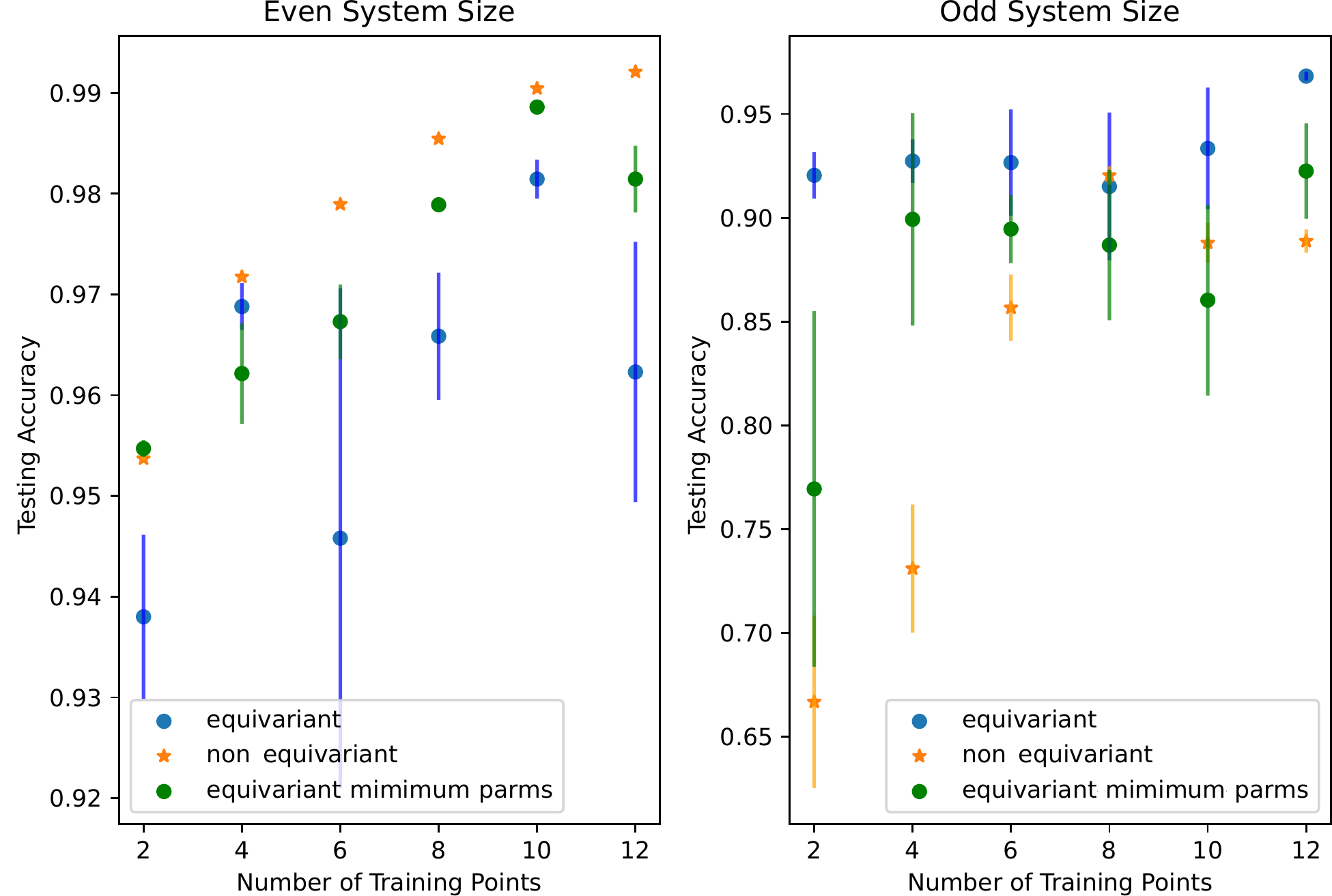}
    \caption{ \textbf{The actual power of equivariance.} The two panels show the mean and variance of the testing accuracy reached by trained QCNNs on the bond-alternating Heisenberg model of even (left) and odd (right) sizes. The average is conducted on both the chosen sizes, $(6,8,10,12)$ for the even case and $(7,9,11,13)$ for the odd one, and on 10 different, randomly initialized training runs for each problem size. The results are plotted against the number of training datapoints $N_T$. The blue circles refer to the EQCNN with the same number of parameters as the HEA-QCNN (orange stars). The green circles describe the EQCNN with the minimum number of parameters possible, \ie, the architecture, as it is, in Fig.~\ref{fig:qcnn-main} with only two standard layers before each pooling layer. These plots demonstrate that equivariance provides significant improvements when, and only when, there is degeneracy in the ground space. }
    \label{fig:eqcnn_vs_heaqcnn}
\end{figure*}

The previous discussion also motivates the analysis shown in Fig.~\ref{fig:eqcnn_vs_heaqcnn}. There, we show a statistical study of the performances of EQCNN and HEA-QCNN when tackling even and odd system sizes. As is evident from the left panel, enforcing equivariance when it is not needed can be more detrimental than beneficial. Indeed, the reduced expressibility of the learning model is not compensated by any benefit and training instabilities emerge, as evidenced by the large error bars in the left panel. However, when equivariance has a reason to be used, as it is for the odd size states studied in the right panel, the advantage of using the EQCNN against a non-informed one is clear. Already with only two training points the equivariant QCNN performs greatly, while the HEA-QCNN needs more training data to generalize well. Interestingly, even with a minimum number of trainable parameters, that for the system sizes studied ranges from 4 to 6, the EQCNN seems to perform better than the non-equivariant one.\\

As stated in the beginning, this is only a preliminary analysis on a simple learning task, and as such we postpone any general conclusion until further studies on the performance of equivariant quantum learning models on more complex systems, against different non-equivariant architectures, and for different symmetry groups. Nonetheless, we think that the preliminary numerical results shown in this section hint at confirming that injecting inductive biases into quantum neural networks boosts their performance, paving the way to the design of new, more efficiently implementable and trainable variational quantum machine learning models.

\section{Discussions and outlook}

Geometric quantum machine learning is a new and exciting field which seeks to produce helpful inductive biases for quantum machine learning models based on the symmetries of the problem at hand.  While there already exist several proposals in the literature within the field of GQML~\cite{larocca2022group,meyer2022exploiting, mernyei2021equivariant, skolik2022equivariant,zheng2021speeding,verdon2019quantumgraph}, these mainly deal with unitary models which maintain the same group representation throughout the computation. In this work, we generalize previous results and we present a theoretical framework to understand, design, and optimize over general equivariant channels, which we refer to as EQNNs. While presented in the setting of supervised learning, our work is readily applicable to other contexts such as unsupervised learning~\cite{otterbach2017unsupervised,kerenidis2019q}, generative modeling~\cite{dallaire2018quantum,benedetti2019generative,kieferova2021quantum,romero2021variational} or reinforcement learning~\cite{saggio2021experimental,skolik2021quantum}.

Our first main contribution is a characterization of the action of equivariant layers as generalized Fourier actions. We argue that the isotypic decomposition of the symmetry's representation determines a generalized Fourier space over which the EQNNs act. This realization not only allows us to characterize the number of free parameters in an EQNN,  but it also unravels the crucial importance that the choice of representation has. That is,  different representations have different block-diagonal structures, and hence, can act on different generalized Fourier spaces and see different parts of the information encoded in the quantum states.  Then, we provide a general classification of EQNN layers, introducing the so-called standard, embedding, pooling, projection, and lifting layers, and we note that non-linearities can be introduced via multiple copies of the data. As a by-product, we highlight the exciting possibility of accessing higher-dimensional irreps of the group symmetry via these non-linearities, which can be a venue to access information that would otherwise be unavailable. 

Our next main contribution is the description of three methods to construct EQNN layers. In the first, which we call the nullspace method, we map the equivariance constraints to a linear system of matrix equations and then solve for their joint nullspace. The second method leverages the technique of twirling over a group, whereby a channel is projected onto the space of equivariant maps. In our third method, we use the Choi operator of the map to  create equivariant layers with specific irrep actions.  Our methods can find unitary or non-unitary equivariant layers efficiently even when the symmetry group is exponentially large, rendering applications for groups that are inaccessible using existing methods in prior literature.  We then compare the strengths and shortcomings of each method, presenting scenarios where one should be favored over the other. Our final key contribution is showing how to  parametrize and optimize EQNNs. In particular, since our work seeks to find equivariant channels, we discuss how one can guarantee that the ensuing maps are physical and potentially easy to implement. To finish, we exemplify our methods by generalizing standard QCNNs to group-equivariant QCNNs, and show how to create, from the ground-up, an $\mbb{SU}(2)$-equivariant QCNN. Finally, we apply this model to a quantum phase classification task on the 1D bond-alternating Heisenberg model and numerically demonstrate its superior performance over a symmetry-agnostic QCNN.

\subsection{Equivariance versus barren plateaus, local minima and data requirements}

Here we argue why EQNNs can alleviate some of the crucial issues in QML, such as barren plateaus, excessive local minima, and poor data requirements.

First, we recall that the barren plateau phenomenon refers to the exponential concentration of gradients exhibited by certain variational quantum models that result in an exponential flattening of the training landscape, and concomitantly in an exponential demand of measurement shots to accurately resolve a parameter update~\cite{mcclean2018barren,cerezo2020cost,sharma2020trainability,holmes2020barren,holmes2021connecting,cerezo2020impact,arrasmith2020effect,arrasmith2021equivalence,marrero2020entanglement,patti2020entanglement,uvarov2020barren,thanasilp2021subtleties,larocca2021diagnosing,wang2020noise,wiersema2020exploring}. The presence or absence of barren plateaus has been directly linked to the expressibility of the model~\cite{holmes2021connecting,larocca2021diagnosing,marrero2020entanglement,patti2020entanglement}, such that highly expressible architectures  exhibit smaller gradients. In our context, the imposition of symmetry constraints to the quantum neural network is expected to shrink -- in a problem oriented way -- its  expressibility, alleviating such gradient vanishing issues. 

Another challenge in training QML models is spurious local minima in the loss landscape. It is known that agnostic models exhibit landscapes that are plagued by local minima~\cite{bittel2021training,anschuetz2022beyond,fontana2022nontrivial,anschuetz2021critical}. However, it is also known that there exists a critical number of trainable parameters above which the model can become overparametrized, meaning that all spurious local traps disappear~\cite{larocca2021theory}. While reaching the overparametrization regime requires exponentially deep circuits for agnostic ansatzes, it has been proven that certain architectures (with reduced  expressibilities) can be efficiently overparametrized with polynomial depth. Thus, the hope is that by restricting the expressibility of the model via geometric priors, one can reduce and realistically reach the overparametrization threshold, thus getting rid of fake local minima. 

Finally, we discuss sample complexity. The ultimate goal of supervised machine learning is to make predictions on unseen data. This is often characterized by generalization bounds, which measure the difference between the performance of the learned model on training and testing data. Recent work has studied the training sample complexity needed for QML models to generalize \cite{huang2021power, caro2021generalization, banchi2021generalization, bu2021effects}.  In particular, Ref. \cite{caro2021generalization}   showed that the training sample complexity typically scales polynomially with the number of trainable parameters. Given a trainable QNN, we have seen that imposing equivariance can drastically reduce the numbers of free parameters, thus implying that incorporating equivariance allows for stronger (more optimistic) bounds on generalization performance. Further,  the bounds of~\cite{caro2021generalization} are  statistical and worst-case (over all possible learning tasks), meaning that EQNNs when applied to the corresponding symmetric learning tasks could potentially achieve better generalizations than indicated by these bounds.

While the previous arguments merely indicate \textit{why} equivariance can improve the performance of a model (in terms of trainability and generalization), these do not constitute a proof that equivariance can indeed fulfill these promises. However, we refer the reader to the recent work of~\cite{schatzki2022theoretical} which studies $S_n$-equivariant models, and rigorously proves that  the equivariance constraints lead to an architecture that avoid barren plateaus, can be efficiently overparametrized, and generalizes well with only polynomially many training points. Thus, the results   in~\cite{schatzki2022theoretical}  showcase the extreme power of EQNNs and GQML.

\subsection{Implications of our work and future directions}
Many concepts and results in our work can be thought of as quantum analogues of existing classical techniques that have enjoyed tremendous successes~\cite{cohen2021equivariant}. We envision that GQML will soon be a  thriving field as it provides blueprints  to create arbitrary architectures and inductive biases suitable for a given problem. As such, the first direct application of our work is  building appropriate schemes to embed classical data into quantum states. Currently, most proposals dealing with classical data use problem-agnostic embedding architectures which completely obviate  and destroy the symmetries in the input data~\cite{havlivcek2019supervised,thanasilp2021subtleties,thanasilp2022exponential}. As such, it is crucial to create embedding schemes that will preserve said symmetries and promote them from the classical to the quantum realm.


As near-term quantum hardware's main challenge is noise, a most important future research direction is to study the interaction of noise and equivariance. Here, there are two possible paths. On one side, one can accept that noise will break equivariance and study the effects of such approximate equivariance. Interestingly, it has been observed in the machine learning literature that mildly breaking equivariance can improve the performance over strictly equivariant models in certain tasks~\cite{wang2022approximately}. On the other side, one can attempt to equivariantize the noise. Being framed in the general superoperator formalism, the present work contains all the necessary tools to study and develop symmetrization strategies to project noise into the symmetric subspace of a given group representation. 
Finally, near-term computation will also be limited in resources (e.g., circuit depth, hardware connectivity, etc) which could prohibit exact equivariance enforcement. Can we derive ``cheaper'' EQNNs at the cost of only approximate equivariance enforcing? How well do EQNNs perform as a function of the symmetry breaking?

\textit{Note added.} After the publication of our work as an arXiv preprint, there have been a number of follow-up works exploring the methods and open questions discussed in our work to develop EQNNs and GQML applications in various contexts. For examples, \cite{le2023symmetry, west2023reflection, Chang2023Approximately, west2023provably} use the twirling method to construct QNN equivariant to finite groups, with applications ranging from calculating molecular force fields, quantum phase detection, image classification. These works all provided numerical results demonstrating improved performance using EQNN. Ref.~\cite{east2023all} propose a method based on spin networks that is shown to be equivalent to the Choi method and apply it to several  lattice Hamiltonian models. Ref.~\cite{das2023role} study the role of choosing representations in EQNN performance. Ref.~\cite{tuysuz2024symmetry} studies the beheaviour of EQNNs in the prescence of noise and derives strategies to protect equivariance.

\section*{Acknowledgements}
This work was partly supported by the U.S. Department of Energy (DOE) through a quantum computing program sponsored by the Los Alamos National Laboratory (LANL) Information Science \& Technology Institute. Q.T.N. acknowledges support from the Harvard Quantum Initiative.
L.S. was partially supported by the NSF Quantum Leap Challenge Institute for Hybrid Quantum Architectures and Networks (NSF Award 2016136).
M.R. was partially supported  by the National Science Foundation through DMS-1813149 and DMS-2108390.
P.J.C. and M.L. were initially supported by the U.S. DOE, Office of Science, Office of Advanced Scientific Computing Research, under the Accelerated Research in Quantum Computing (ARQC) program. P.J.C. and L.S. were also supported by the LANL ASC Beyond Moore's Law project. 
F.S. was supported by the Laboratory Directed Research and Development (LDRD) program of LANL under project number 20220745ER. 
M.L. was also supported by the Center for Nonlinear Studies at LANL. 
M.C. acknowledges support by the LDRD program of LANL under project numbers 20210116DR and 20230049DR.

\bibliography{quantum,numerics_bib}

\clearpage
\newpage

\appendix
\onecolumngrid

\section*{Appendices for ``Theory for Equivariant Quantum Neural Networks''}


\section{Equivariance in existing quantum algorithms}\label{sec:perspective}
In this section we discuss some notable (non-variational) quantum algorithms from the perspective of group equivariance. While not explicitly mentioned in the original works, these algorithms rely on equivariant operations. This shows the significance of equivariance in designing quantum algorithms beyond variational QML.

\subsection{Quantum state purity}
Computing the purity $\operatorname{Tr}[\rho^2]$ is a canonical task in quantum information theory. Since the purity is unitary-invariant, it should be expected that algorithms aiming at measuring it must use unitary-equivariant transformations. We consider two algorithms in \cite{cincio2018learning, van2012measuring} and show that, indeed, these use equivariant operations as defined in this work.

The Bell-basis algorithm described in Figure 6 of \cite{cincio2018learning} starts from two copies of $\rho$ and performs a change of basis to the Bell basis before measuring the observable $CZ^{\otimes n}$ (controlled-phase). This is equivalent to simply measuring the equivariant observable ${\rm SWAP}^{\otimes n}$, which belongs to the commutant of the representation $R(U)=U^{\otimes 2}$ of $\mathbb{U}(2^n)$, such that $\operatorname{Tr}[\rho^{\otimes 2} {\rm SWAP}^{\otimes n}]=\operatorname{Tr}[U^{\otimes 2}\rho^{\otimes 2} U^{\dagger \otimes 2}{\rm SWAP}^{\otimes n}]$ for any $U \in \mathbb{U}(2^n)$.

In contrast, the algorithm in \cite{van2012measuring} only uses one copy of $\rho$ and is based off tools from random matrix theory. It starts by appending $m-n$  (where $2^n \ll 2^m$) zero-initialized qubits to the $n$-qubit state $\rho$. Then, a random $m$-qubit unitary is applied on the composite system. Finally, a $2^n$-dimensional projective measurement is performed, where one estimates the probability of obtaining some (fixed) outcome $\ket{k}$, $\text{Pr}(k)$. This procedure is repeated many times to estimate $\langle \text{Pr}(k)^2 \rangle$ (with the average taken over the random distribution of $m$-qubit unitaries), from which one infers the purity as $\operatorname{Tr}[\rho^2] = 2^m \langle \text{Pr}(k)^2 \rangle -1$. Here we show that this algorithm is effectively composed of equivariant transformations. Observe that

\begin{align}
    \langle \text{Pr}(k)^2 \rangle &= \int_{U \in \mathbb{U}(2^m)} d\mu(U) \operatorname{Tr}[(\rho \otimes \ket{0}\bra{0}_{m-n}) U^{\dagger} \ket{k}\bra{k} U ]^2  \nonumber\\
     &= \int_{U \in \mathbb{U}(2^m)} d\mu(U) \operatorname{Tr}[(\rho \otimes \ket{0}\bra{0}_{m-n})^{\otimes 2} U^{\dagger \otimes 2} \ket{kk}\bra{kk} U^{\otimes 2} ] \nonumber\\
     &=  \operatorname{Tr}\left[(\rho \otimes \ket{0}\bra{0}_{m-n})^{\otimes 2} \int_{U \in \mathbb{U}(2^m)} d\mu(U) U^{\dagger \otimes 2} \ket{kk}\bra{kk} U^{\otimes 2} \right]\,. \label{eq:vanenk}
\end{align}

We can interpret Eq. \eqref{eq:vanenk} as follows. The first step is $\rho \rightarrow (\rho \otimes \ket{0}\bra{0})^{\otimes 2}$, which is a $(\mathbb{U}(2^n),R^{\text{in}},R^{\text{out}})$-equivariant non-linear embedding (Definition~\ref{def:pooling-embedding}), where we define $R^{\text{in}}(U)=U$ and $R^{\text{out}}(U)=(U\otimes \id_{m-n})^{\otimes 2}$ for $U \in \mathbb{U}(2^n)$. The second step is measuring the observable $\widetilde{O} = \int_{U \in \mathbb{U}(2^m)} U^{\dagger \otimes 2} \ket{kk}\bra{kk} U^{\otimes 2}$, which commutes with $R^{\text{out}}$ due to the invariance of the Haar measure of $\mathbb{U}(2^m)$.

\subsection{Quantum convolutional neural networks (QCNN)}
Quantum convolutional neural networks (QCNN) were proposed in Ref.~\cite{cong2019quantum}.  The architecture presented takes inspiration from classical CNNs~\cite{lecun1998gradient} and relies on local gates with (potentially) shared parameters in between the gates belonging to the same layer.
As an example of application, a QCNN is used to classify phases of the ground states of a Haldane chain~\cite{haldane1983nonlinear}. 
Of particular relevance, the  Hamiltonian $H_{\rm{hal}}$ of the Haldane chain can be verified to commute with the group of symmetry $\mathbb{Z}_2\times \mathbb{Z}_2 \equiv \{1, r, s, rs\}$, with unitary representation of the generators $r$ and $s$ given by $R(r)= X_{\text{even}} \equiv \prod_{i \text{ even}} X_i$ and $R(s) = X_{\text{odd}} \equiv \prod_{i \text{ odd}} X_i$. 
Since the nature of the ground state of $H_{\rm{hal}}$ does not change under the action of these unitaries, we identify $G=\mathbb{Z}_2\times \mathbb{Z}_2$ as the symmetry group of the task.

For this ground-state classification problem, two QCNNs were studied in Ref.~\cite{cong2019quantum}: one trainable (see Supplementary Figure 2 in Ref.~\cite{cong2019quantum}), and one ``exact'' (Figure 2b of Ref.~\cite{cong2019quantum} and reproduced in Fig.~\ref{fig:qcnn}) that is obtained based on the MERA representation of the ground states of $H_{\rm{hal}}$.
Remarkably, close inspection of the exact QCNN reveals that it is composed of equivariant layers and measurement. That is, the exact model for this task follows the framework of EQNNs laid down in this work.
However, one can see that the choice of trainable model adopted in Ref.~\cite{cong2019quantum} does not comply with equivariance requirements, such that we expect that it could be further improved by imposing equivariance.
In the following, we briefly discuss how one can identify equivariance of the different layers of the exact QCNN.

First we consider the action of a convolution layer (denoted C in Fig.~\ref{fig:qcnn} and excluding the SWAP operations) onto $R(r)$. Using the following identity
\begin{equation*}
\begin{quantikz}[row sep={1.0cm,between origins}]
\qw & \gate{X} & \ctrl{1} & \qw  \\
\qw& \qw & \control{} & \qw  \\
[-0.5cm]
\end{quantikz}
=
\begin{quantikz}[row sep={1.0cm,between origins}]
\\[-0.2cm]
\qw  & \ctrl{1} & \gate{X} & \qw \\
\qw & \control{} & \gate{Z} & \qw  \\
\end{quantikz}
\end{equation*}
and noticing that, due to the connectivity of the controlled-Z gates the resulting $Z$s unitaries acting on any of the qubits can only be created by pairs (yielding an identity), one can verify that $R(r)$ commutes with the action of all the controlled-Z gates. Additionally, given that both the controls and target of the Toffoli gates act on the eigenbasis of X operators, one can see that $R(r)$ commutes with the action of the whole convolution layer. Similar reasoning shows that such commutativity properties also holds true for $R(s)$, and thus for $R(rs)$. Overall, we find the convolution layers to be $(G,R,R)$-equivariant.

Second, we consider the action of the pooling layer (denoted P in Fig.~\ref{fig:qcnn} and including the SWAP operations). 
 One can verify that $\phi_{\rm{pool}} \circ \Adj_{R(g)} = \Adj_{R'}(g) \circ \phi_{\rm{pool}}$, where we have denoted as $\phi_{\rm{pool}}$ the map realized by the pooling from $n$ to $n/3$ qubits, and by $R'(g)$ the representation of $g\in G$ on this reduced space. In particular, $R'$ is defined in equivalence to $R$ but over a reduced number $n/3$ of qubits. Overall, we find the pooling layers to be $(G,R,R')$-equivariant.

\begin{figure}[t]
    \centering
    \includegraphics[width=0.4\textwidth]{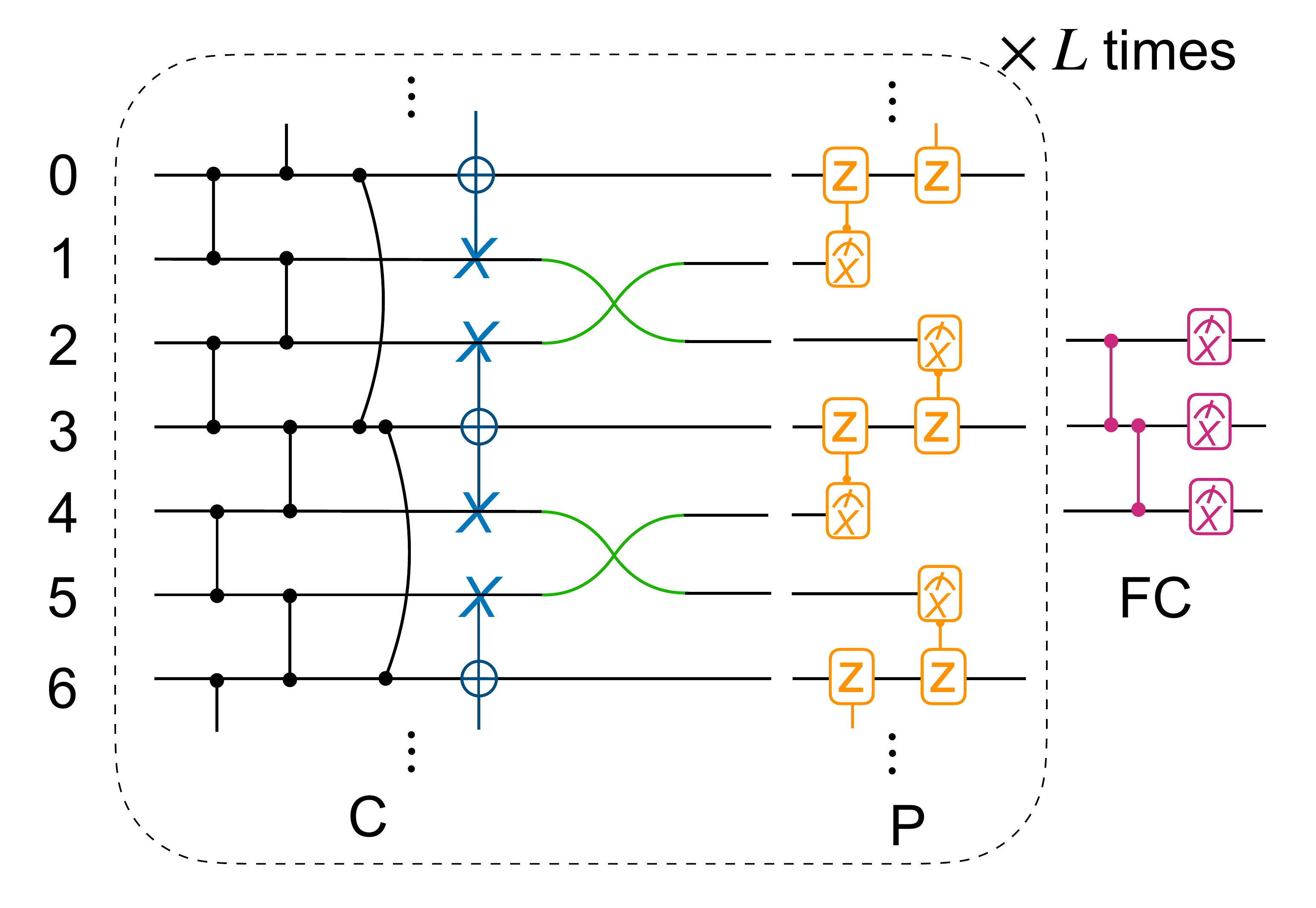}
    \caption{Reproduction of the "exact" QCNN architectures used in \cite{cong2019quantum} for the classification of the ground states of the Haldane model. C stands for convolution, P for pooling, and FC for fully connected. The blue three-qubit gates are Toffoli gates with control qubits taken in the X basis. 
    The orange two-qubit gates apply a pauli Z onto the target qubit when an X measurement of the control qubit results in a $-1$ outcome, but leaves the target qubit unchanged otherwise.
    The black two-qubit gates are controlled-Z gates.
    The green interleaved lines corresponds to a SWAP of $2$ qubits.
    The C-P structure is repeated $L$ times until the system is left with only three qubits.
    }
    \label{fig:qcnn}
\end{figure}

Finally, note that the measurement realized by the fully connected layer (FC in Fig.~\ref{fig:qcnn}) corresponds to a measurement of the Pauli observable $O=ZXZ$ on the $3$ remaining qubits. Notably $O \in \mf{comm}(R^{\rm out})$, where $R^{\rm out}$ is the representation of $G$ on the remaining qubits and is defined as $R^{\rm out}(r)=XIX$ and $R^{\rm out}(s)=IXI$. That is, we find the measurement to be $(G, R^{\text{out}})$-equivariant, such that the overall exact QCNN follows the requirements of Proposition.~\ref{prop:invar} ensuring invariance of the model.
$R_1$-invariant ($\mathbb{Z}_2 \times \mathbb{Z}_2$-invariant).

\subsection{Quantum-enhanced experiments}
The quantum-enhanced experiment in \cite{huang2021quantum} leverages coherent access to multiple copies of a quantum state obtained from physical experiments to learn its properties. In their ``predicting observables'' task (Theorem 1 in \cite{huang2021quantum}), the goal is to predict the \emph{absolute value} of an $n$-qubit Pauli observables $O$ on states of the form $\rho = (\id + 0.9sP)/2^n$, where $P$ is also an $n$-qubit Pauli string and $s\in \{0,\pm 1\}$. Notice that this task is Pauli-invariant since for any Pauli string $\sigma$ we have that $|\operatorname{Tr}(\rho O)|=|\operatorname{Tr}(\sigma \rho \sigma O)|$. The first step in their algorithm is adding a copy: $\rho \rightarrow \rho^{\otimes 2}$, which is a non-linear equivariant layer as defined in Definition~\ref{def:nonlinear}. Then, a Bell measurement is performed followed by classical post-processing. Let $G=\{\sigma: n\text{-qubit Pauli strings}\}$ and consider the following representations of  $G$: $R_{\text{def}}(\sigma)=\sigma$ and  $R_{2}(\sigma) = \sigma^{\otimes 2}$. They derived the following equation (Section D.2 of \cite{huang2021quantum}'s Supplementary Materials)
\begin{equation}
    |\operatorname{Tr}[\rho O]|  = \mathbb{E} \left[\operatorname{Tr}\left[O \otimes O  \bigotimes_{k=1}^{n} S_k \right] \right],
    \label{eq:q-enhanced}
\end{equation}
where $S_k$ is the Bell projector corresponding to the Bell measurement result on qubit pair $k$. That is $S_k=\ket{\Psi_k}\bra{\Psi_k}$ for a Bell state $\ket{\Psi_k}$. Expanding Eq.~\eqref{eq:q-enhanced} we have
\begin{align}
|\operatorname{Tr}[\rho O]|  &= \sum_{\ket{\Psi_1\hdots \Psi_n}} \operatorname{Tr}\left[O^{\otimes 2}  \ket{\Psi_1\hdots \Psi_n} \bra{\Psi_1\hdots \Psi_n} \right] \bra{\Psi_1\hdots \Psi_n} \rho^{\otimes 2} \ket{\Psi_1\hdots \Psi_n} \nonumber\\
& = \operatorname{Tr} \left[\rho^{\otimes 2} \sum_{\ket{\Psi_1\hdots \Psi_n}} \ket{\Psi_1\hdots \Psi_n}  \bra{\Psi_1\hdots \Psi_n}  O^{\otimes 2}  \ket{\Psi_1\hdots \Psi_n} \bra{\Psi_1\hdots \Psi_n} \right]  \nonumber\\
& = \operatorname{Tr}[\rho^{\otimes 2} \widetilde{O}].
\end{align}
We can see that after adding another copy of $\rho$, the algorithm effectively performs a measurement of the observable $\widetilde{O}$. It is readily verified that $\widetilde{O}$ commutes with $R_2$ using the fact that, for any single-qubit Pauli $\sigma$, the operator $\sigma^{\otimes 2}$ admits the Bell states $\ket{\Psi_k}$ as eigenvectors with eigenvalues $\pm 1$. Hence, the entire algorithm is $R_{\text{def}}$-invariant.

\subsection{Classical shadows}
Classical shadows is an efficient protocol for predicting observables on quantum states using randomized measurements \cite{huang2020predicting}. In this protocol, one applies a random unitary $U$ drawn from a unitary ensemble $\mathcal{E}$ on the state $\rho$, then performs a computational basis measurement to obtain a bit string $\vec{z}$. E.g., Ref. \cite{huang2020predicting} considered Pauli and Clifford ensembles. Repeating this process many times allows one to predict properties of $\rho$. We first consider $\mathcal{E}$ to be the $n$-qubit Clifford group. The expected classical shadow (see Section 5.B of Ref. \cite{huang2020predicting}'s Supplement Information) can be rewritten as
\begin{align}
    \mathbb{E}_{\vec{z},U} [(2^{n}+1)U^{\dagger} \ket{\vec{z}}\bra{\vec{z}}U - \id]&= (2^{n}+1) \sum_{\vec{z}} \frac{1}{|\mathcal{E}|} \sum_{U \in \mathcal{E}} U^{\dagger} \ket{\vec{z}}\bra{\vec{z}}U \operatorname{Tr}[\ket{\vec{z}}\bra{\vec{z}}U \rho U^{\dagger}] - \id \nonumber \\
    &= (2^{n}+1) \operatorname{Tr}_B \left[ \underbrace{ \left(\frac{1}{|\mathcal{E}|} \sum_{U \in \mathcal{E}} U^{\dagger \otimes 2} \left(\sum_{\vec{z}} \ket{\vec{z}\vec{z}}\bra{\vec{z}\vec{z}} \right)U^{\otimes 2} \right)}_{\widetilde{O}} (\id \otimes \rho) \right] - \id , \label{eq:shadow}
\end{align}
where $\operatorname{Tr}_B$ denotes the partial trace over the second subsystem.

Notice that the sum over $\mathcal{E}$ is equal to the Haar integral over $\mathbb{U}(2^n)$ as the Clifford group forms a 3-design \cite{webb2016clifford}, i.e., $\widetilde{O}= \int_{U \in \mathbb{U}(2^n)} d\mu (U) U^{\dagger \otimes 2} O U^{\otimes 2} $, where $O:= \sum_{\vec{z}} \ket{\vec{z}\vec{z}}\bra{\vec{z}\vec{z}}$. Thus, $\widetilde{O}$ commutes with the representation $R_2(U):=U^{\otimes 2}$ of $\mathbb{U}(2^n)$. We can therefore interpret Eq.~\eqref{eq:shadow} as a composition of three equivariant layers in Definition~\ref{def:pooling-embedding}:

\begin{equation}
    \rho \qquad \overset{Embedding}{\longrightarrow} \qquad \id \otimes \rho \overset{Standard}{\longrightarrow} \qquad \widetilde{O} \rho  \overset{Pooling}{\longrightarrow} \qquad \operatorname{Tr}_B [\widetilde{O}(\id \otimes \rho)].
    \label{eq:convolution-shadow}
\end{equation}
The corresponding representations transform as
\begin{equation}
    R_{\text{def}} \qquad \overset{(G,R_{\text{def}},R_2)\text{-equivariant}}{\longrightarrow} \qquad R_2 \overset{(G,R_2,R_2)\text{-equivariant}}{\longrightarrow} \qquad  R_2  \overset{(G,R_2,R_{\text{def}})\text{-equivariant}}{\longrightarrow} \qquad R_{\text{def}}.
\end{equation}

Next, we consider $\mathcal{E}$ to be the Pauli ensemble. The expected classical shadow is
\begin{align}
    \mathbb{E}_{\vec{z},U} \left[\bigotimes_{j=1}^{n} (3 U_j^{\dagger} \ket{\vec{z}_j}\bra{\vec{z}_j}U_j - \id )\right]&= \sum_{\vec{z}} \frac{1}{|\mathcal{E}|} \sum_{U \in \mathcal{E}} \bigotimes_{j=1}^{n} (3 U_j^{\dagger} \ket{\vec{z}_j}\bra{\vec{z}_j}U_j -\id) \operatorname{Tr}[\ket{\vec{z}}\bra{\vec{z}}U \rho U^{\dagger}]  \nonumber\\
    &= \operatorname{Tr}_B \left[  \left(\frac{1}{|\mathcal{E}|} \sum_{U \in \mathcal{E}} \bigotimes_{j=1}^{n} U_j^{\dagger \otimes 2} \underbrace{ \left(\sum_{\vec{z}_j} 3 \ket{\vec{z}_j\vec{z}_j}\bra{\vec{z}_j\vec{z}_j}-\id \otimes \ket{\vec{z}_j}\bra{\vec{z}_j} \right)}_{M_j} U_j^{\otimes 2} \right) (\id \otimes \rho) \right] \nonumber\\
    &= \operatorname{Tr}_B  \left[ \left(  \bigotimes_{j=1}^{n} \underbrace{\int_{U_j \in \mathbb{U}(2)} U_j^{\dagger \otimes 2} M_j U_j^{\otimes 2}}_{\widetilde{M}_j} \right) (\id \otimes \rho) \right] \nonumber\\
    & = \operatorname{Tr}_B \left[ \widetilde{M} (1\otimes \rho)\right],
\end{align}
where in the last equality we used the fact that the 1-qubit Pauli group forms a 3-design of $\mathbb{U}(2)$. By explicitly evaluating the integral, we find that $\widetilde{M}_j= SWAP$, thus $\widetilde{M}=SWAP^{\otimes n}$, which commutes with the representation $R_2$. A similar composition of equivariant layers to Eq.~\eqref{eq:convolution-shadow} thus follows.

\section{Equivariant maps as generalized group convolutions in Fourier space}\label{app:structure}



In this section,  we provide further details on the interpretation of classical equivariant maps as generalized group convolutions, given in Sec.~\ref{sec:interpret}. More specifically, by looking at the Fourier space, we find that group convolution is a special case of equivariant maps when the representation is the regular representation. We first consider a finite group $G$ and later the Lie group.

Recall that the group convolution of two vectors $a, b \in \mathbb{C}^{|G|}$ (a.k.a. functions mapping group elements to scalars) is defined as
\begin{equation}
     (a \circledast b) (u) = \sum_{v \in G} a(uv^{-1}) b(v).
\end{equation}
Note that the above expression can be rewritten as a matrix vector multiplication as follows:
\begin{equation}
a \circledast b = A \Vec{b}, \qquad \text{where } A = \sum_{u \in G} a(u) R_{\text{left}}(u).
\end{equation}
Here, $R_{\text{left}}$ denotes the \emph{left regular representation} of the group $G$, which maps each group element $u$ to a permutation matrix $R_{\text{left}}(u)$ that performs $R_{\text{left}}(u)\ket{v}= \ket{uv}$. For example, in CNNs \cite{lecun1998gradient} the convolution matrix $A$ is a circulant matrix since the group $\mathbb{Z}_n\times \mathbb{Z}_n$ is abelian.

The key property here is that the group Fourier transform, $F_G := \sum_{v \in G} \sum_{\rho \in \hat{G}} \sqrt{\frac{d_{\rho}}{|G|}} \sum_{i,j=1}^{d_{\rho}} \rho(v)_{j,k} \ket{\rho,j,k} \bra{v}$, block-diagonalizes the left regular representations into irreps as~\cite{childs2010quantum}
\begin{equation}
    R_{\text{left}}(u) = F_G^{\dagger} \left( \bigoplus_{\xi \in \hat{G}} \xi (u) \otimes \id_{d_{\xi}} \right) F_G,
\end{equation}
where $\hat{G}$ denotes the set of inequivalent irreps of $G$ and $d_\xi$ is the dimension of the irrep $\xi$.

Thus, under $F_G$, the convolution matrix $A$ is block-diagonalized as
\begin{equation}
    A \cong \bigoplus_{\xi \in \hat{G}} \hat{a} (\xi) \otimes \id_{d_{\xi}},
\label{eq:block-diag}
\end{equation}
where $\hat{a}(\xi) := \sum_{u \in G} a(u) \xi(u)$ is the \emph{Fourier transform} of the kernel $c$ and $\hat{G}$ is the set of inequivalent irreps.

Comparing Eqs.~\eqref{eq:comm-structure} and \eqref{eq:choi-decompose} to Eq.~\eqref{eq:block-diag} we see that equivariant channels generalize group convolution. Group convolution are then equivariant maps to the regular representation, where the multiplicities $m_\xi$ are equal to the irrep dimensions $d_\xi$, and the basis is the regular Fourier basis. This is why most of classical homogeneous ENN architectures \cite{kondor2018generalization, cohen2019general} (ENNs on regular representations) implement equivariant layers via group convolutions.

We complete this section with a similar analysis on compact Lie groups. For these groups, an irrep decomposition of the regular representation similar to Eq.~\eqref{eq:block-diag} also exists due to the Peter-Weyl theorem \cite{folland2016course}. Let $U(L^2(G))$ denote the group of unitary operators on the Hilbert space of $L^2$-integrable functions on the group $G$, then the left regular representation, $R_{\text{left}}: G\to U(L^2(G))$, is defined as $ R_{\text{left}}(u) f(x) = f(u^{-1} x)$, for $u\in G$ and $f\in L^2(G).$

\begin{theorem}[Peter-Weyl theorem]
	Let $G$ be a compact group.
	Then the regular representation $R_{\text{left}}: G\to U(L^2(G))$ is isomorphic to a direct sum of the irreducible representations of $G$:
\begin{equation}
    R_{\text{left}} \cong \bigoplus_{\xi \in \hat{G}} \xi \otimes \id_{d_{\xi}}.
    \label{eq:peter-weyl}
\end{equation}
\end{theorem}

Above, the isomorphism is the group Fourier transform, which maps functions $f \in L^2(G)$ to operator-valued functions, $\hat{f}(\xi) = \int_Gd\mu(u) f(u)\xi(u)$. The Lie group convolution is similarly defined as $(f \circledast k)(u) := \int_G  d\mu(v) f(uv^{-1})k(v)$ for $f,k \in L^2(G).$ One can verify that it is possible to derive  a block-diagonalization of $f$ in the Fourier basis similar to that in Eq.~\eqref{eq:block-diag}. Note that the sum over inequivalent irreps in Eq.~\eqref{eq:peter-weyl} is infinite as Lie groups has infinitely many inequivalent irreps. In contrast, for Lie group representations on finite qubit systems (as we consider in this work), the EQNN equivariant layers only process irreps up to some truncated irreps.

\section{Deferred proofs}\label{app:proofs}

In this section we present proof for some of the results in the main text. For convenience of the reader, we recall the statement of the theorems and propositions prior to their proofs.

\subsection{Deferred proofs from Section \ref{sec:connecting}}

\setcounter{theorem}{2}

\begin{theorem}[Free parameters in equivariant channels] Let the irrep decomposition of $R:=R^{\text{in}*} \otimes R^{\text{out}}$ be $R(g) \cong \bigoplus_{q} R_q(g) \otimes  \id_{m_q}$. Then any $(G,R^{\text{in}},R^{\text{out}})$-equivariant CPTP channels can be fully parametrized via $\sum_{q} m_q^2 - C(R^{\text{in}},R^{\text{out}})$ real scalars, where $C(R^{\text{in}},R^{\text{out}})$ is a positive constant that depends on the considered representations.
\label{thm:paramcountchannel-sm}
\end{theorem}

\begin{proof}[Proof of Theorem~\ref{thm:paramcountchannel}] Let $\phi$ be a $(G,R^{\text{in}}, R^{\text{out}})$-equivariant channel. By Theorem \ref{thm:comm} and Lemma \ref{lem:choi-equivariance}, the Choi operator $J^\phi$ is decomposed as $J^\phi \cong \bigoplus_{q=1}^{Q} \id_{d_q} \otimes J_q^\phi$, where each $J_q^\phi$ is an operator in an $m_q$-dimensional subspace corresponding to the irrep decomposition $R:=R^{\text{in}*}\otimes R^{\text{out}}$. For convenience of notation, we will denote as $\HC_B\otimes\HC_A$  the Hilbert space over which $J^\phi$ acts. Imposing $J^\phi \geq 0$ (CP) is equivalent to imposing $J_q^\phi \geq 0$ for each irrep $q$. An $m_q$-dimensional complex positive semidefinite operator is parametrized by $m_q^2$ real scalars, for a total of $\sum_q m_q^2$ parameters. Next, we impose TP via $\operatorname{Tr}_B[J^\phi]=\id_A$, where $\id_A$ denotes the identity over $\HC_A$.

 Let the change of basis in the irrep decomposition be $W$, i.e., $J^\phi=W\ad(\bigoplus_{q=1}^{Q} \id_{d_q} \otimes J^\phi_q) W$, where we dropped the superscript $\phi$ for brevity. The TP condition reads
\begin{align}
    \operatorname{Tr}_B [J^\phi]= \id_A &= \sum_j ( \bra{j}_B \otimes \id_A ) W\ad \left( \bigoplus_{q=1}^{Q} \id_{d_q} \otimes J_q^\phi \right) W (\ket{j}_B \otimes \id_A)= \sum_j T_j\ad \left( \bigoplus_{q=1}^{Q} \id_{d_q} \otimes J_q^\phi \right) T_j,
\end{align}
where $T_j = W ( \ket{j}_B \otimes \id_A)$. 
Vectorizing the above equation,  we can use the property $\text{vec}(M_1M_2M_3)=(M_3^\top\otimes M_1)\text{vec}(M_2)$ to obtain
\begin{equation}
    D \cdot \text{vec} \left( \bigoplus_{q=1}^{Q} \id_{d_q} \otimes J_q^\phi \right) = \text{vec} (\id_A),\qquad \text{where } D := \sum_{j \in \text{dim}(\HC_B)} T_j^\top \otimes T_j\ad\in \mathbb{C}^{\text{dim}(\HC_A)^2 \times (\text{dim}(\HC_A) \text{dim}(\HC_B))^2}.
\end{equation}
Let $\widetilde{D}$ be the $\text{dim}(\HC_A)^2 \times \sum_q m_q^2$ matrix whose columns correspond to the nonzero entries in $\text{vec} \left( \bigoplus_{q=1}^{Q} \id_{d_q} \otimes J_q^\phi \right)$. Then $\text{rank}(\widetilde{D})= C(R^{\text{in}},R^{\text{out}})$.

It is readily verified that in the non-equivariant case, i.e., $W=\id$ and $J^\phi$ is fully parameterized, the matrix $\widetilde{D}=D$ is full row-rank, in which case imposing TP reduces $\text{dim}(\HC_A)^2$ free parameters as expected.
\end{proof}

\setcounter{proposition}{1}

\begin{proposition}[Insensitivity to  equivalent representations] Consider an EQNN as defined in Definition \ref{def:eqnn}. Then changing
an intermediate representation, $R^l$, to another representation equivalent to it, $VR^l V^{\dagger}$, where $V$ is a unitary, does not change the expressibility of the EQNN.
\label{prop:equivalentreps-sm}
\end{proposition}

\begin{proof}[Proof of Proposition~\ref{prop:equivalentreps}]
Here we show that changing a representation to another equivalent representation does not change the expressibility of the EQNN. In particular, consider two EQNNs that undergo the same representations except at one place
\begin{align}
    \mathcal{N}: \qquad R^{\text{in}} \longrightarrow \hdots \longrightarrow R^1 \overset{\NC^1}{\longrightarrow} &R \overset{\NC^2}{\longrightarrow} R^2 \longrightarrow \hdots \longrightarrow R^{\text{out}}.\\
    \mathcal{N}': \qquad R^{\text{in}} \longrightarrow \hdots \longrightarrow R^1 \overset{{\NC^1}'}{\longrightarrow} &R' \overset{{\NC^2}'}{\longrightarrow} R^2 \longrightarrow \hdots \longrightarrow R^{\text{out}},
\end{align}
where $R'=V  R V^{\dagger}$ for some unitary $V$.

Observe that the set of $(G,R_1,R)$-equivariant channels is in one-to-one correspondence to the set of $(G,R_1,R')$-equivariant channels. Indeed, for any $(G,R_1,R)$-equivariant $\NC_1$, the channel $\NC_1'= \Adj_V \circ \NC_1$ is $(G,R_1,R')$-equivariant as we have that
\begin{align}
    \NC^1 =  \Adj_R \circ \NC^1 \circ \Adj_{R_1^{\dagger}} = \Adj_R \circ (\Adj_{V^{\dagger}} \circ ({\NC^1}') \circ \Adj_{R_1^{\dagger}} \Leftrightarrow  {\NC^1}' = \Adj_{V R V^{\dagger}} \circ {\NC^1}' \circ \Adj_{R_1^{\dagger}}.
\end{align}
Similarly, we obtain ${\NC^2}'=\NC^2 \circ \Adj_{V^{\dagger}}$. Thus ${\NC^2}'\circ {\NC^1}'=\NC^2 \circ \NC^1$. A similar argument can be made for changing between equivalent output representations, in which case there is a bijection between the observables that commute with $R^{\text{out}}$ and those that commute with ${R^{\text{out}}}'$.
\end{proof}

\subsection{Deferred proofs from
Section~\ref{sec:framework}}

\begin{theorem}[Finite group equivariance]
    Given a finite group $G$ with generating set $S$, a linear map $\phi$ is $(G, R^{\text{in}}, R^{\text{out}})$-equivariant if and only if
    \begin{align}\label{eq_finite_grp_eq-sm}
       \phi \circ \Adj_{R^{\text{in}}(g)} - \Adj_{R^{\text{out}}(g)} \circ \phi=0, \quad \forall g \in S.
    \end{align}
    \label{thm:gen-set-sm}
\end{theorem}

\begin{proof}[Proof of Theorem~\ref{thm:gen-set}]
    Given a finite group $G$ and a generating set $S=\{h_1,\cdots,h_{|S|}\}\subset G$, we can identify any group element $g$ with a sequence $\kappa = (\kappa_1,\kappa_2,\cdots, \kappa_N)$ where $\kappa_i \in \{1,\cdots,|S|\}$, such that $g = \prod_{i=1}^N h_{\kappa_i}$~\cite{finzi2021practical}. Assuming the equivariance condition is satisfied for the generating set, i.e., for any $h\in S$
    \begin{equation}\label{eq_proof_th_gen_set}
        \Adj_{R^{\text{out}}(h)}  \circ \phi \circ \Adj_{R^{\text{in}}(h)}\ad = \phi\,.
    \end{equation}
Then, we can readily show its also satisfied for any $g \in G$
    \begin{align}
    \left(\Adj_{R^{\text{out}}(h_{\kappa_1})}\circ \cdots \circ \Adj_{R^{\text{out}}(h_{\kappa_1})}\right)  \circ \phi \circ \left(\Adj_{R^{\text{in}}(h_{\kappa_1})}\circ \cdots \circ \Adj_{R^{\text{in}}(h_{\kappa_1})}\right)\ad = \phi\,,
    \end{align}
    where we have applied $N$ times Eq.~\eqref{eq_proof_th_gen_set}.
\end{proof}

\begin{theorem}[Lie group equivariance]\label{thm:equiv_algebra-sm}
    Given a compact Lie group $G$ with a Lie algebra $\mathfrak{g}$ generated by $s$ such that exponentiation is surjective, a linear map $\phi$ is $(G, R^{\text{in}}, R^{\text{out}})$-equivariant if and only if 
    \begin{align}
        \adj_{r^{\text{out}}(a)} \circ \phi - \phi \circ \adj_{r^{\text{in}}(a)}=0, \quad \forall a \in s,
    \end{align}
    where $r^{\text{in}},r^{\text{out}}$ are the representations of $G$ induced by $R^{\text{in}}, R^{\text{out}}$.
    \label{thm:lie-alg-sm}
\end{theorem}

\begin{proof}[Proof of Theorem~\ref{thm:lie-alg}]
    Let us start by recalling that for any element $a \in \mathfrak{g}$ there is a corresponding $e^a \in G$. If $\phi$ is $(G, R^{\text{in}}, R^{\text{out}})$-equivariant, then $\Adj_{R^{\text{out}}(g)} \circ \phi = \phi \circ \Adj_{R^{\text{in}}(g)}, \ \forall g \in G$. Differentiating this expression yields that $\adj_{r^{\text{out}}(a)} \circ \phi = \phi \circ \adj_{r^\text{in}(a)}$, where $g = e^a$. Since this holds for any element $a$ in $\mathfrak{g}$, it must hold for the generating set of $\mathfrak{g}$.
    
    We now prove the other direction by first showing that $\adj_{r^{\text{out}}(a)} \circ \phi = \phi \circ \adj_{r^{\text{in}}(a)}, \ \forall a \in \mathfrak{g}$ assuming this relation holds for the generating set $s=\{a_1,\hdots,a_{|s|}\}$ of $\mathfrak{g}$. First we consider the element $[a_i, a_j]$ in $\mathfrak{g}$ as follows
    \begin{align*}
        \adj_{r^{\text{out}}([a_i,a_j])} \circ \phi
        &= \adj_{r^{\text{out}}(a_i)}\circ \adj_{r^{\text{out}}(a_j)} \circ \phi - \adj_{r^{\text{out}}(a_j)}\circ \adj_{r^{\text{out}}(a_i)} \circ \phi\\
        &= \phi \circ \adj_{r^{\text{in}}(a_i)} \circ \adj_{r^{\text{in}}(a_j)} - \phi \circ \adj_{r^{\text{in}}(a_j)} \circ \adj_{r^{\text{in}}(a_i)}\\
        &= \phi \circ \adj_{r^{\text{in}}([a_i,a_j])}.
    \end{align*}

    Recursively applying the above calculation we find that $\adj_{r^{\text{out}}(a')}\circ \phi = \phi \circ \adj_{r^{\text{in}}(a')} $ for any nested commutator $a'=[a_{i_1},[a_{i_2}, \ldots]]$. Hence, $\adj_{r^{\text{out}}(a)} \circ \phi = \phi \circ \adj_{r^{\text{in}}(a)}, \ \forall a \in \mathfrak{g}$ since $s$ generates $\mathfrak{g}$. It follows that $e^{\adj_{r^{\text{out}}(a)}}\circ \phi  = \phi \circ e^{\adj_{r^{\text{in}}(a)}}$. Since exponentiation is surjective, for all $g \in G$ there is a corresponding $a \in \mathfrak{g}$ such that $g = e^a$ and accordingly $e^{\adj_{r(a)}}=\Adj_{R(g)}$. Therefore, $\Adj_{R^{\text{out}}(g)} \circ \phi = \phi \circ  \Adj_{R^{\text{in}}(g)}, \ \forall g \in G$.
\end{proof}

\subsection{Relaxing the assumption of surjectivity}
In Theorem~\ref{thm:equiv_algebra} we assumed that the exponentiation map $e: \mathfrak{g} \rightarrow G$ between the Lie algebra and group was surjective. This let us work interchangeably at both algebra or group level. However, even when the surjectivity assumption is relaxed (i.e., when the exponential of the algebra generates only some subset of the Lie group) there may still be a finite set of generators such that equivariance with respect to this set implies equivariance with the entire group. This happens since the exponential map takes a Lie algebra to a connected component of the Lie group (that containing the identity), but Lie groups can have multiple connected components. In cases where there are finitely many connected components and the quotient of the Lie group with a finite subgroup $H$ yields a connected Lie group $G/H$, the theory does not require much adjustment. When this holds, one can write the group as
\begin{align}
    G = e^\mathfrak{g}H.
\end{align}
That is, for all $g\in G$, there exists $a \in \mathfrak{g}$ and $h \in H$ such that  $g = e^a h$. Ref.~\cite{finzi2021practical} achieves Lie group equivariance by imposing the constraint over a basis of $\mf{g}$ and a generating set for $H$. We here leverage Theorem~\ref{thm:equiv_algebra} to reduce the number of constraints by imposing them over the connected component generated by $\mf{g}$ to a generating set of it.
\begin{theorem}
    Consider a Lie group G with a Lie algebra $\mathfrak{g}$ generated by $s = \{a_1, \ldots, a_{|s|}\}$ such that $G = e^\mathfrak{g}H$ where $H$ is a subgroup generated by $S = \{ h_1,\ldots,h_{|S|}\}$. Then a linear map $\phi$ is $(G,R^\text{in},R^\text{out})-$equivariant if and only if
    \begin{align}
    \adj_{r^\text{out}(a_i)} \circ \phi = \phi \circ \adj_{r^\text{in}(a_i)}\, \ \forall a_i \in s,\\
        \Adj_{R^\text{out}(h_i)} \circ \phi = \phi \circ \Adj_{R^\text{in}(h_i)},\ \forall h_i \in S.
    \end{align}
\end{theorem}
\begin{proof}
 As shown in the proof of Theorem~\ref{thm:equiv_algebra}, if $\phi$ is equivariant then it commutes with the representation on the algebra level. Further, it also clearly commutes with the elements of $H$.
 The other direction also follows similarly from the proof of Theorem~\ref{thm:equiv_algebra}. From there we have that $\Adj_{R^{\text{out}}(g)} \circ \phi = \phi \circ  \Adj_{R^{\text{in}}(g)}, \ \forall g \in e^\mathfrak{g}$. Then by the assuming that $\phi$ commutes on $H$ and $G = e^\mathfrak{g} H$, we have that, for any $g=e^a h$,
 \begin{align}
     \Adj_{R^\text{out}(g)} \circ \phi & = \Adj_{R^\text{out}(e^a h)} \circ \phi\nonumber\\
     & = \Adj_{R^\text{out}(e^a)} \circ \Adj_{R^\text{out}(h)} \circ \phi\nonumber\\
     & = \phi \circ \Adj_{R^\text{in}(e^a)} \circ \Adj_{R^\text{in}(h)}\nonumber\\
     & = \phi \circ \Adj_{R^\text{in}(g)}. \qedhere
 \end{align}
\end{proof}

As an example, consider $\mbb{SO}(3)$ and $\mbb{O}(3)$. Exponentiation of $\mf{so}(3)$ is surjective on $\mbb{SO}(3)$ and thus we can take $H=\{\id_3\}$ (the trivial group). However, $\mbb{O}(3)$ consists of two connected components corresponding to $\pm 1$ determinant. Due to determinant being multilinear, one can map between these components with the subgroup $H = \{\id_3, (-1)\oplus \id_2 \}$.

\section{Advanced twirling methods}
\label{app:advanced_twirling}

In the main text we have discussed how to obtain equivariant maps via twirling. Here we discuss several advanced techniques to implement the twirling operator. First, we will show how the twirl of an operator or a map can be obtained via the Weingarten calculus. Next, we will showcase two methods for in-circuit twirling, and one for approximate twirling.

\subsection{Weingarten calculus}

The Weingarten calculus~\cite{collins2006integration, puchala2017symbolic,collins2022weingarten} is an extremely powerful tool that can be used to find the exact expression of the twirl of an operator over a group. At its core, the Weingarten calculus leverages the key property that twirling is equivalent to projecting into the commutant. Hence, if the commutant of the representation is well known (for instance via the Schur-Weyl duality~ \cite{goodman2009symmetry,koike1989decomposition, grinko2022linear}) then one can analytically find an expression for the twirled operator in terms of its components over a basis of the commutant. In what follows we exemplify the Weingarten calculus for twirling an operator when $R^{\text{in}}=R^{\text{out}}=R$. For  a channel, one can use the techniques here presented by finding the components of the Choi operator in the commutant of the   ${R^{\text{in}}}^* \otimes R^{\text{out}}$, which can be identified using the \emph{mixed} Schur-Weyl duality~\cite{koike1989decomposition, grinko2022linear}.

Given an operator $X$, a group $G$ and a representation $R$, we already know that the twirl is a projection over the commutant. That is, 
\begin{equation}\label{eq:twirled_X_comm}
    \TC_{G}[X]=\int_{G}d\mu(g) R(g) X R(g)\ad=\sum_{i=1}^{\dim(\mf{comm}(R))} c_\mu(X) P_i\,, \quad \text{with} \quad P_i \in \text{basis}(\mf{comm}(R))\,,
\end{equation}
Hence, in order to solve Eq.~\eqref{eq:twirled_X_comm} one needs to determine the $\dim(\mf{comm}(R))$ unknown coefficients $\{c_i\}_{i=1}^{\dim(\mf{comm}(R))}$. The previous can be achieved by finding $\dim(\mf{comm}(R))$ such equations and, and solving a linear system problem. 
In particular, note that changing $X\rightarrow X P_j$ for some $P_j\in \text{basis}(\mf{comm}(R))$ leads to 
\begin{align}\label{eq:twirled_X_comm_2}
    \TC_{G}[X P_j ]&=\int_{G}d\mu(g) R(g) X P_j R(g)\ad=\int_{G}d\mu(g) R(g) X  R(g)\ad P_j=\sum_{i=1}^{\dim(\mf{comm}(R))} c_i(X) P_i P_j\,,
\end{align}
where we have used the fact that $P_j$ commutes with all $R(g)$. Then, taking the trace on both sides leads to 
\begin{align}\label{eq:LSP}
\Tr[XP_i]= \Tr[\mathcal{T}_G[XP_i]] = \sum_{i=1}^{\dim(\mf{comm}(R))}c_i(X)\Tr[P_iP_j] \,.
\end{align}
Repeating Eq.~\eqref{eq:LSP} for all $P_j$'s in $ \text{basis}(\mf{comm}(R))$ leads to $\dim(\mf{comm}(R))$ equations. Thus, one can find the vector of unknown coefficients $\vec{c}(X)=(c_1(X),\ldots,c_{\dim(\mf{comm}(R))}(X))$ by solving $A\cdot \vec{c}(X)=\vec{b}(X)$, where $\vec{b}(X)=(\Tr[X P_1],\ldots, \Tr[X P_{\dim(\mf{comm}(R))}])$. Here,  $A$ is a the so-called Gram matrix, a $\dim(\mf{comm}(R))\times \dim(\mf{comm}(R))$ symmetric matrix with entries $[A]_{ij}=\Tr[P_i P_j]$. One can then solve the linear system problem by inverting the  Gram matrix, as $\vec{c}(X)=A^{-1}\cdot\vec{b}(X)$. The matrix $A^{-1}$ is known as the Weingarten matrix.

\subsection{In-circuit twirling with ancillas or classical randomness}
While we usually find the set of equivariant channels analytically, in the cases of small finite groups, we note one could perform the twirling directly on the quantum circuit using the following unitaries
\begin{equation}
\begin{split}
    &\mathbf{U}^{\text{in}} = \sum_{g \in G} \ket{g}\bra{g} \otimes R^{\text{in}}(g)^{\dagger}\,\\
    &\mathbf{U}^{\text{out}} = \sum_{g \in G} \ket{g}\bra{g} \otimes R^{\text{out}}(g)\,,
\end{split}
\end{equation}
along with $a=\log_2 |G|$ ancilla qubits initialized to the uniform superposition state $H^{\otimes a}\ket{0}$. That is, the in-circuit twirling of a $n$-to-$m$-qubit channel $\phi$ can be realized via the following circuit.
\begin{equation*}
\begin{quantikz}
\lstick[wires=1]{$\ket{0}$} & \gate{H} \qwbundle{a} & \gate[2,disable auto height]{\mathbf{U}^{\text{in}}} & \qw &   \qwbundle{a} & \qw  & \gate[2,disable auto height]{\mathbf{U}^{\text{out}}} & \trash{\text{discard}} \\ &
\lstick[wires=1]{$\rho$} & \qwbundle{n} & \qwbundle{n} & \gate{\phi} & \qwbundle{m} & \qw & \qwbundle{m} \\
\end{quantikz}
\end{equation*}
It can be readily verified that this circuit performs the twirling formula in Eq.~\eqref{eq:twirl-finite}. With this, $\phi$ can be any parametrized channel native to the circuit platform.

Alternatively, the ancilla qubits can be replaced by classical randomness. That is, we classically sample a group element $g$ then apply ${R^{\text{in}}}^{\dagger}(g)$ and $R^{\text{out}}(g)$ as follows.
\begin{equation*}
\begin{quantikz}
    & \gate[style=circle,nwires={1}]{g}\vcw{1} & & &  &\gate[style=circle,nwires={1}]{g}\vcw{1} & \\
    \lstick[wires=1]{$\rho$} &  \gate{{R^{\text{in}}}^{\dagger}(g)} & \qwbundle{n} & \gate{\phi} & \qwbundle{m} & \gate{R^{\text{out}}(g)}&\qw 
\end{quantikz}
\end{equation*}
The latter method can be favorable on near-term devices. Furthermore, the Hoeffding's bound implies that only $O(\log|G|)$ classical samples are needed to achieve a good approximate of the twirled channel.

One disadvantage of in-circuit twirling, however, is that albeit ensuring equivariance we lose the parameter count reduction as in the case when we first compute equivariant channels analytically before parametrization.

\subsection{Recursive approximate twirling}
For Lie groups with more intricate representation theory, computing the twirling formula can quickly become complex and difficult. Instead, \cite{toth2007efficient} provided an algorithm for approximating twirling operators that converges \emph{exponentially} fast in the number of Haar-random samples of group elements. While the authors did not mention this, their proofs do not rest upon any assumptions beyond that the representations they consider are unitary representations of compact Lie groups and the input of the twirling formula is self-adjoint. Their algorithm can be applied to our case when one can efficiently sample from the Haar measure and is summarized in Algorithm~\ref{alg:twirl}. This approximate twirling algorithm can also be implemented in-circuit using classical randomness, similarly to what we saw earlier in the case of finite groups.

\begin{algorithm}
\caption{Recursive approximate twirling \cite{toth2007efficient}}\label{alg:twirl}
\begin{algorithmic}
\Require Group $G$, unitary representations $R^{\text{in}}$, $R^{\text{out}}$, map $\phi$, tolerance $\varepsilon >0$, $\delta >0$.
\Ensure Output map $\widetilde{\phi}$ such that $\|\widetilde{\phi} - \mathcal{T}_G[\phi]\|_{HS} < \varepsilon$ with probability at least $1 -\delta$. \Comment{HS: Hilbert-Schmidt norm}
\State $\widetilde{\phi} \gets \phi$
\For{$k \gets 1$ to $K = O( \log \frac{1}{\delta\varepsilon})$}
    \State Draw random $g \in G$
    \State $\widetilde{\phi} \gets \frac{1}{2} \left( \widetilde{\phi} + R^{\text{out}}(g) \circ \widetilde{\phi} \circ {R^{\text{in}}}^{\dagger}(g) \right)$
\EndFor
\end{algorithmic}
\end{algorithm}

\section{Implementing and optimizing equivariant channels}\label{app:compile}
\subsection{Channel compiling}
In the process of creating equivariant QNNs we consider not just equivariant unitaries, but also more general quantum channels. Via the Stinespring dilation theorem any channel $\phi: \mathcal{B}(\mathcal{H}_A) \rightarrow \mathcal{B}(\mathcal{H}_B)$ can be represented as a unitary operation on a larger space:
\begin{align}
    \phi(\rho) = \Tr_E[U (\rho \otimes \ket{e}\bra{e}) U^\dagger],
\end{align}
where $\ket{e}$ is a fixed reference state on the environment $E$. The size of this environment system is directly related to the Kraus rank of $\phi$. Recall that a quantum channel can be written as 
\begin{align}
    \phi(\rho) = \sum_i K_i \rho K_i^\dagger,
\end{align}
where we say $\{K_i\}$ are the channel's Kraus operators. Note that the spectral decomposition of the Choi operator yields one possible set of Kraus operators. At most $\dim(\mathcal{H}_A)\dim(\mathcal{H}_B)$ Kraus operators are necessary to represent any channel \cite{wilde2013quantum}. One can then define a unitary with action
\begin{align}
    U^\phi = \sum_i K_i \otimes \ket{i}\bra{e} + U',
\end{align}
where $U'$ is some arbitrary operator that completes $U$ to be unitary. Thus, to represent a channel with $k$ Kraus operators, an environment of dimension at least $k$ suffices. Working with qubit systems, if the input space is $n$ qubits and the output $m$, then the maximum Kraus rank is $2^{m+n}$, which may require an environment of up to $m+n$ qubits.

We will not go into great detail on how to perform this circuit compilation, but rather refer the interested reader to the considerable body of work on compilation. For example, a software package for this decomposition with nearly optimal CNOT-count can be found in \cite{iten2019Introduction} with corresponding theory in \cite{iten2016Quantum, iten2017Quantum}. For more general works on circuit compiling we direct the reader to \cite{khatri2019quantum, sharma2019noise, cincio2021machine, cincio2018learning, caro2021generalization, chong2017programming, haner2018software}.

\subsection{Converting vectorized channels to Choi operators}\label{app:transfer_to_choi}
In solving the nullspace for equivariant maps, we work with vectorized channels. That is $\phi \mapsto \overline{\phi} = \sum_{i,j}\phi_{i,j}\ket{P_i}\rangle\langle\bra{P_j},$ where $P_j$ and $P_i$ are Pauli strings. In some references this is refered to as a transfer matrix. The canonical construction of transfer matrices is as follows. Consider the Kraus operator form of a channel, $\phi(\rho) = \sum_i K_i \rho K_i^\dagger$. The vectorization map $\text{vec}(X)$ takes a matrix $X$ to a vector through the mapping
\begin{align}
    \text{vec}(\ket{i}\bra{j}) = \ket{i}\otimes\ket{j}\,,
\end{align}
with linear extension. Then, using the identity $\text{vec}(AXB)=(B^T\otimes A)\text{vec}(X)$, one can write a matrix representation of $\phi$ as 
\begin{align}
    \overline{\phi} = \sum_i K_i^* \otimes K_i\,.
\end{align}
This is the transfer matrix. Further, the transfer matrix can be directly mapped to the Choi operator via
\begin{align}
    \tau^\mathcal{N} = N^\Gamma\,,
\end{align}
where $\Gamma$ is an involution map such that
\begin{align}
    \bra{i,j} \overline{\phi} \ket{k,l} = \bra{l,j}\overline{\phi}^\Gamma\ket{k,i}\,.
\end{align}

For us to apply this identity, we need convert from the transfer matrix in terms of Pauli strings to that of the computational basis. As any $\ket{i}\bra{j}$ can be written as a sum over Pauli strings, there is some change of basis $U$ from Pauli strings to the computational operator basis. Explicitly, the columns of $U$ will be vectors corresponding to the expansion of Pauli strings in the computational basis. If we obtain a matrix in the Pauli basis, $X$, we can then write it in the computational basis via $UAV^{-1}$, where $U$ is the change of basis on the output space and $V$ is that on the input space. Further, if we take Pauli strings to be normalized such that $\Tr[P_iP_j]=\delta_{i,j}$, then $U$ and $V$ can be assumed to be unitary matrices. In this case,

\begin{align}
    J^\phi = (U \overline{\phi} V^\dagger)^\Gamma\,.
\end{align}

\subsection{Optimizing equivariant channels}\label{sec:trainchannels}
We now provide a strategy for optimizing $n$-to-$m$ qubit equivariant channels. Assume we have found a basis $\{\phi_i\}$ for such channels via methods outlined in this work, equivariant channels in the span of this basis can then be written as
\begin{equation}
    \phi_{\Vec{x}}[\cdot]= \sum_{i=1} x_i \phi_i[\cdot] + \frac{\operatorname{Tr}[\cdot]}{2^m}\id.
\end{equation}
The coefficients must be constrained so that $ \phi_{\Vec{x}}$ is CPTP. For convenience we have fixed $\phi(\rho) = \frac{\Tr[\rho]}{2^m}\id$ as one of the basis elements. Without loss of generality we can then take all $\phi_i$ to be trace annihilating, i.e., $\Tr[\phi_i(\rho)]=0$. Note that depending on methods used to find these maps, we can bake in CP or TP or even both.
Given a valid set of parameters for this pooling layer, we can \emph{classically} solve the following circuit compiling problem as described before. 
\begin{equation}
\begin{quantikz}
\lstick[wires=1]{$\ket{0}$} & \qwbundle{m+n} & \gate[3,disable auto height]{U(\theta)} &  \\
\lstick[wires=2]{$\rho$} &\qwbundle{n-m} &\qw & &  \\
&\qwbundle{m} & \qw  & \qw & \qwbundle{m}
\end{quantikz}=
\begin{quantikz}
\lstick[wires=2]{$\rho$} \qw &\gate[2,disable auto height]{\phi_{\Vec{x}}} &  \\
\qw & & \qw & \qw
\end{quantikz}
\end{equation}
The top $2n$ qubits are discarded. $U(\theta)$ is a general $(2n+m)$-qubit unitary. 
Note that $m+n$ is an upperbound on the number of ancilla qubits needed, but depending on the ranks of the basis channels we could potentially need fewer qubits. After solving this classical circuit compilation problem, one can then implement it on a real quantum circuit.

Now that we have a way to implement and parametrize equivariant channels, we can train them via two approaches:
\begin{enumerate}[label=(\roman*)]
    \item Projected gradient descent (GD) in the circuit parameters $\thv$ space (parameter-shift rule works),
    \item Projected GD in the classical variables $\Vec{x}$ space.
\end{enumerate}
For more details on projected GD, we refer the reader to \cite{trefethen1997numerical}. Other constrained optimizers can be used as well.

In the first approach, we treat $x_i$ as functions of $\thv$ and train $\thv$ via the projected GD algorithm, i.e., perform the following in each training step:
\begin{equation}
\begin{aligned}
    \theta^{k+\frac{1}{2}} &= \theta_k - \eta \nabla_{\thv} \mathcal{L}  \qquad \text{(regular GD)}\\
    \theta^{k+1} &= \min_{\thv'} \|\theta' - \theta^{k+\frac{1}{2}} \| \qquad \text{(projection)} \\
    & \text{subject to: } \phi_{\Vec{x}(\thv')} \text{ is CPTP}.
\end{aligned}
\end{equation}
In the second approach, we treat $\thv$ as functions of $x_i$ and train $x_i$. In this case, however, we cannot use finite different methods to compute the derivatives with respect to these parameters.
\begin{equation}
\begin{aligned}
    x^{k+\frac{1}{2}} &= x^k - \eta \nabla_{x} \mathcal{L}  \qquad \text{(regular GD)}\\
    x^{k+1} &= \min_{x'} \|x' - x^{k+\frac{1}{2}} \| \qquad \text{(projection)} \\
    & \text{subject to: } \phi_{\Vec{x}} \text{ is CPTP}.
\end{aligned}
\end{equation}

One might also want to recycle the ancilla qubits. This is possible if we can replace the partial trace operation by a measurement on the ancillas followed by a controlled unitary (on the possible outcomes). This requires the channel to be unital, which is the case if the output representation is irreducible.

\begin{lemma} Equivariant channels whose output representation is irreducible are unital.
\end{lemma}
\begin{proof} We require that $\varphi(R^{\text{in}}(g) \rho R^{\text{in}}(g)^{\dagger}) = R^{\text{out}}(g) \varphi(\rho) R^{\text{out}}(g)^{\dagger}$ for any $g \in G$. Substituting $\rho=\id^{\otimes n}$ and applying the Schur's lemma we find $\varphi(\id^{\otimes n})=\id^{\otimes m}$.
\end{proof}
\section{$\mbb{SU}(2)$-equivariant 2-to-1-qubit and 1-to-2-qubit channels}\label{app:su2}
\subsection{From equivariant maps to channels}
Let us first define the five $(\mbb{SU}(2), U^{\otimes 2}, U)-$equivariant linear maps, which we denote as 
\begin{align}
    \phi_1(\rho)&=\Tr[\rho]\frac{\id}{2},\quad
    \phi_2(\rho)&=\Tr[\rho\SWAP]\frac{\id}{2},\quad
    \phi_3(\rho)&=\Tr_A[\rho],\quad
    \phi_4(\rho)&=\Tr_B[\rho],\quad
    \phi_5(\rho)&=\sum_{ijk=1}^3\Tr[\rho \sigma_i\sigma_j]\epsilon_{ijk}\sigma_k\,.
\end{align}
We will refer to $\phi_5$ as the \textit{cross-product} channel. Here we will give slightly more detail on how we find the feasible region for channels. To see that $\phi_2$ may act non-trivially on the trace of an input, we consider its vectorization in the Pauli basis. The row corresponding to $\ket{\id_\text{out}}\rangle\langle\bra{P}$ contains all information about $\Tr[\phi_2(\rho)]$ as $\id$ is the only Pauli string of non-zero trace. One can show that $\Tr[\SWAP (\sigma_i\otimes \sigma_j)] = 2\delta_{ij}$ and thus $\phi_2(\sigma_i\otimes \sigma_i) = \frac{\id}{2}$. But this is problematic as $\Tr[\phi_j(\sigma_i\otimes\sigma_i)] = 0$ for $j\in\{1,3,4,5\}$. Thus, $\phi_2$ increases trace of a state such as $\frac{1}{2}(I+XX)$ but none of the other maps could cancel out this increase. Thus, we drop $\phi_2$ from our set of maps.

To finish finding the feasible region, we will make some modifications to our basis elements. That is, we want to modify the basis set such that all elements except for $\phi_1$ are trace annihilating. Then we can, without loss of generality, fix the coefficient for the trace preserving channel to be $1$. As the cross-product channel is traceless, we need only modify the partial trace channels. This is easy in the Pauli string basis, simply remove the upper left hand corner entry corresponding to $\id\otimes\id \mapsto \id$. For example, the $\Tr_A$ channel becomes
\begin{align}
\begin{pmatrix}
    0 & 0 & 0 & 0 & 0 &\cdots\\
    0 & 2 & 0 & 0 & 0 & \cdots\\
    0 & 0 & 2 & 0 & 0 & \cdots \\
    0 & 0 & 0 & 2 & 0 & \cdots
\end{pmatrix}
\end{align}
With these modified channels, we know that the set of equivariant channel can be characterized as
\begin{align}
    \{x,y,z \in \mathbb{R}^3: J^{\phi_1}+xJ^{\phi_5}+yJ^{\phi_3'}+zJ^{\phi_4'} \geq 0\}.
\end{align}
Requiring the eigenvalues of this linear combination to be non-negative yields the feasible region
\begin{align}
    \{ x,y,z: y + z \leq 1 \text{ and } y + z \geq \sqrt{3x^2 + 4(y^2 - yz+z^2)}-1\}.
\end{align}

\subsection{Action of $\mbb{SU}(2)$-equivariant maps}

Here we further analyze the action of the 2-to-1 qubit maps. This analysis will show that different channels can ``see'' different parts of the input state, and hence that they are complimentary. First, define the Bell basis states as 
\begin{align}
\begin{split}
    \ket{\beta_{00}}&=\frac{1}{\sqrt{2}}(\ket{00}+\ket{11})\,, \quad 
    \ket{\beta_{01}}=\frac{1}{\sqrt{2}}(\ket{00}-\ket{11})\\
    \ket{\beta_{10}}&=\frac{1}{\sqrt{2}}(\ket{01}+\ket{10})\, \quad
    \ket{\beta_{11}}=\frac{1}{\sqrt{2}}(\ket{01}-\ket{10})\,.
    \end{split}
\end{align}

Then, consider a two qubit quantum state in the Bell basis
\begin{equation}
    \rho=\begin{pmatrix}
    a_{11} & a_{12} & a_{13} & a_{14}\\
    a_{12}^* & a_{22} & a_{23} & a_{24}\\
    a_{13}^* & a_{23}^* & a_{33} & a_{34}\\
    a_{14}^* & a_{24}^* & a_{34}^* & a_{44}
    \end{pmatrix}\,,
\end{equation}
where $a_{11} + a_{22} + a_{33} + a_{44}=1$.

One can readily find that
\begin{align}
    \phi_1(\rho)&=\frac{1}{2}\begin{pmatrix}
    a_{11} + a_{22} + a_{33} + a_{44} & 0\\
    0 & a_{11} + a_{22} + a_{33} + a_{44}
    \end{pmatrix}\,,\\
    \phi_2(\rho)&=\frac{1}{2}\begin{pmatrix}
    a_{11} + a_{22} + a_{33} - a_{44} & 0\\
    0 & a_{11} + a_{22} + a_{33} - a_{44}
    \end{pmatrix}\,,\\
    \phi_3(\rho)&=\begin{pmatrix}
   \frac{1}{2} +\Re[a_{12}-a_{34}] & \Re[a_{13}+a_{24}]+i\Im[a_{23}+a_{14}]\\
    \Re[a_{13}+a_{24}]-i\Im[a_{23}+a_{14}] & \frac{1}{2} - \Re[a_{12}-a_{34}]
    \end{pmatrix}\,,\\
   \phi_4(\rho)&= \begin{pmatrix}
   \frac{1}{2} +\Re[a_{12}+a_{34}] & \Re[a_{13}-a_{24}]+i\Im[a_{23}-a_{14}]\\
    \Re[a_{13}-a_{24}]-i\Im[a_{23}-a_{14}] & \frac{1}{2} - \Re[a_{12}+a_{34}]
    \end{pmatrix}\,,\\
    \phi_5(\rho)&=4\begin{pmatrix}
   -\Im[a_{34}] & \Im[a_{24}]-i\Re[a_{14}]\\
    \Im[a_{24}]+i\Re[a_{14}] & \Im[a_{34}]
    \end{pmatrix}\,.
\end{align}
These equations show how different channels combine different pieces of the information of $\rho$. 

\subsection{Cross-product channel}\label{app:cp}
We now take a closer look at $\phi_5$, which we call the \textit{cross-product} channel $\CP:\HC^{\otimes 2}\rightarrow \HC$, where $\HC$ is the single-qubit Hilbert space. The action of the $\CP$ (not to be confused with complete positivity) channel is as follows
\begin{align}
    \CP(\rho)=&\sum_{ijk=1}^3\Tr[\rho \sigma_i\sigma_j]\epsilon_{ijk}\sigma_k
    =\Tr[\rho (YZ-ZY)]X+
    \Tr[\rho (ZX-XZ)]Y+\Tr[\rho (XY-YX)]Z\,,
\end{align}
where $\sigma_\mu\in\{X,Y,Z\}$ for $\mu=i,j,k$.

First, note that in the form above the $\CP$ channel is not truly a channel as it is not trace preserving nor completely positive. Rather, $\Tr[\mathcal{CP}(\rho)] = 0$. To be a channel, we must consider superoperators of the form ${\phi} + \alpha \mathcal{CP}$, where ${\phi}$ is some trace preserving map. For simplicity we consider ${\phi}(\rho) = \frac{\Tr[\rho]}{2}\id$. By solving for the eigenvalues of the Choi operator of ${\phi} + \alpha \mathcal{CP}$, one can show that this is a channel for $\alpha \in [-\frac{1}{\sqrt{3}}, \frac{1}{\sqrt{3}}]$.

Now that we know when this can actually be physical, we'd like to better understand the action of the $\CP$ channel. As before, we recall the four Bell basis states 
\begin{align}
\begin{split}
    \ket{\beta_{00}}&=\frac{1}{\sqrt{2}}(\ket{00}+\ket{11})\,, \quad 
    \ket{\beta_{01}}=\frac{1}{\sqrt{2}}(\ket{00}-\ket{11})\\
    \ket{\beta_{10}}&=\frac{1}{\sqrt{2}}(\ket{01}+\ket{10})\, \quad
    \ket{\beta_{11}}=\frac{1}{\sqrt{2}}(\ket{01}-\ket{10})\,,
    \end{split}
\end{align}
where $\ket{\beta_{00}}$, $\ket{\beta_{01}}$ and $\ket{\beta_{10}}$ are eigenstates of the ${\rm SWAP}$ operator with eigenvalue $+1$, while $\ket{\beta_{11}}$ is an eigenvalue of the SWAP operator with eigenvalue $-1$. 
Then, we note the following operator expansion in the Bell basis
\begin{align}
    YZ-ZY=2i(\ketbra{\beta_{01}}{\beta_{11}}-\ketbra{\beta_{11}}{\beta_{01}})\\
    ZX-XZ=2(\ketbra{\beta_{00}}{\beta_{11}}+\ketbra{\beta_{11}}{\beta_{00}})\\
    XY-YX=-2i(\ketbra{\beta_{10}}{\beta_{11}}-\ketbra{\beta_{11}}{\beta_{10}})\,.
\end{align}

We can then express the density matrix in the Bell basis (ordered as $\{\ket{\beta_{00}},\ket{\beta_{01}},\ket{\beta_{10}},\ket{\beta_{11}}\}$),
\begin{equation}
    \rho=\begin{pmatrix}
    \fbox{ $\begin{array}{c c c}
    & \,\,\,\,\,\,\,\,\,\,\,\,\,\,\,& \\
    & & \\
    & & \\
    \end{array}$} \!\!\!\!\!\!&  \begin{array}{c}
     a_{14} \\
    a_{24} \\
    a_{34}
    \end{array}    \vspace{.4mm}\\
    \begin{array}{c c c}
    a_{14}^* & a_{24}^* &a_{34}^* 
    \end{array} &\fbox{ $\begin{array}{c}
    \end{array}$}
    \end{pmatrix}\,,
    \vspace{.5mm}
\end{equation}
where the $3\times 3$ and $1\times 1$ diagonal blocks correspond to the symmetric and anti-symmetric subspaces, respectively. Then, the matrix elements $a_{14}$, $a_{24}$ and $a_{34}$ determine if the state is in a superposition of symmetric and anti-symmetric Bell basis states.

With the previous, one can verify that 
\begin{align}
    \Tr[\rho(YZ-ZY)]=4\Im[a_{24}]\,,\quad
    \Tr[\rho(ZX-XZ)]=4\Re[a_{14}]\,,\quad
    \Tr[\rho(XY-YX)]=-4\Im[a_{34}]\nonumber\,.
\end{align}

Note that the action of the $\CP$ channel is to check if $\rho$ is in a superposition of states with different symmetries. As such, at the output of the map the  coefficients  associated with the different Pauli operators correspond to the entries of the matrix entries that account superposition between the anti-symmetric state and the three different symmetric states.  Hence, the $\CP$ channel outputs the zero matrix for any state that is block diagonal in the symmetric and anti-symmetric subspaces (such as $\rho=\sigma^{\otimes 2}$ for any single-qubit state $\sigma$).

Let us note here an interesting fact. Namely, that the $\CP$ channel, in its vanilla version, has an asymmetry embedded into it. Namely, it only accounts for either the real or imaginary part of the matrix of $\rho$. This can be solved by defining the following alternative version of the $\CP$ channel 
\begin{equation}
     \CP'(\rho)=i\sum_{ijk=1}^3\Tr[\rho {\rm SWAP} \sigma_i\sigma_j]\epsilon_{ijk}\sigma_k\,,
\end{equation}
where we have multiplied by ``$i$'' to have the output matrix be Hermitian. 
One can readily check that this new CP channel is also equivariant. That is, 
\begin{equation}
    \CP'(U^{\otimes 2}\rho (U\ad)^{\otimes 2})= U \CP'(\rho) U\ad\,.
\end{equation}

Now, let us note that
\begin{align}
    {\rm {\rm SWAP}}(YZ-ZY)=2i(\ketbra{\beta_{01}}{\beta_{11}}+\ketbra{\beta_{11}}{\beta_{01}})\\
    {\rm {\rm SWAP}}(ZX-XZ)=2(\ketbra{\beta_{00}}{\beta_{11}}-\ketbra{\beta_{11}}{\beta_{00}})\\
    {\rm {\rm SWAP}}(XY-YX)=-2i(\ketbra{\beta_{10}}{\beta_{11}}+\ketbra{\beta_{11}}{\beta_{10}})\,,
\end{align}
which means
\begin{align}
    \Tr[\rho {\rm {\rm SWAP}}(YZ-ZY)]&= 4i  \Re[\rho_{24}]\nonumber\\
    \Tr[\rho {\rm {\rm SWAP}}(ZX-XZ)]&= -4i  \Im[\rho_{14}]\\
    \Tr[\rho {\rm {\rm SWAP}}(XY-YX)]&= -4i  \Re[\rho_{34}]\nonumber\,.
\end{align}

Thus, we can now create combinations of the two versions of the $\CP$ channels. For instance
the channel
\begin{align}
    \frac{1}{4}(\CP(\rho)+\CP'(\rho)) = & -(\Re[a_{24}]- \Im[a_{24}]) X\nonumber\\ &+ (\Re[\rho_{14}]+\Im[a_{14}]) Y\\
    &+ (\Re[ \rho_{34}]-\Im[a_{34}]) Z\,.\nonumber
\end{align}
and 
\begin{align}
    \frac{1}{4}(\CP(\rho)-\CP'(\rho)) = & (\Re[a_{24}]+ \Im[a_{24}]) X\nonumber\\ &+ (\Re[a_{14}]-\Im[a_{14}]) Y\\
    &- (\Re[ a_{34}]+\Im[a_{34}]) Z\,.\nonumber
\end{align}
By applying these channels (with the appropriate amount of completely depolarizing channel added such that the operation is physical), one may then recover the off-diagonal terms between the symmetric and anti-symmetric subspaces.

Just as a curiosity, we note that the following combination (although not physical) is still interesting
\begin{equation}
    \frac{1}{4}(\CP(\rho)-i\CP'(\rho))=i \rho_{24}^* X + \rho_{14}^* Y  - i \rho_{34}^* Z\,.
\end{equation}

In Fig.~\ref{fig:CP_circuit} we give a circuit for achieving the $\CP$ channel in expectation. The intuition behind this circuit is that, without the Pauli rotations, this protocol changes the state of the ancilla such that $\Tr[Z\rho_\text{out}] = \Tr[ZZ\rho_\text{in}]$, where $\rho_\text{out}$ is the state on the ancilla. Thus, by adding in Pauli rotations, we permute the Pauli group to embed $\Tr[\sigma_i\sigma_j\rho]$ as $\Tr[\sigma_k \rho_\text{out}]$ on the ancilla. Note that a similar circuit does not exist for $\CP'$ as $\SWAP\otimes \sigma_i\sigma_j$ is not Hermitian. One could instead accomplish $\CP'$ in expectation through the Hadamard test.
\begin{figure}[t]
    \centering
    \begin{quantikz}
        & \gate[style=circle,nwires={1}]{i}\vcw{1} & \gate[style=circle,nwires={1}]{j}\vcw{2}\\
        \lstick[wires=2]{$\rho$} &  \gate{R_i(\frac{\pi}{2})} & \qw & \ctrl{2} & \qw & \trash{\text{trash}}\\
        & \qw & \gate{R_j(\frac{\pi}{2})} & \qw & \ctrl{1} & \qw & \trash{\text{trash}}\\
        \lstick{$\ket{0}$} & \gate{X^{\frac{1-\epsilon_{ijk}}{2}}} & \qw & \targ{} & \targ{} & \gate{R_k^\dagger(\frac{\pi}{2})} & \qw & \rstick{$\CP(\rho)$}\qw\\
        & \gate[style=circle,nwires={1}]{i,j,k}\vcw{-1} &  &  &  & \gate[style=circle,nwires={1}]{k}\vcw{-1}
    \end{quantikz}
    \caption{\textbf{Circuit for the cross-product channel.} Here $i,j,k\in\{X,Y,Z\}$. These classical random variables are drawn from $\{X,Y,Z\}$ uniformly and without replacement. That is, $\{i,j,k\} = \{X,Y,Z\}$ always but with the individual rotations randomly chosen. After performing this protocol, the reduced state on the ancilla will be equal, in expectation, to $\CP(\rho)$. Note that $\epsilon_{ijk}$ is the Levi-Civita symbol.}
    \label{fig:CP_circuit}
\end{figure}
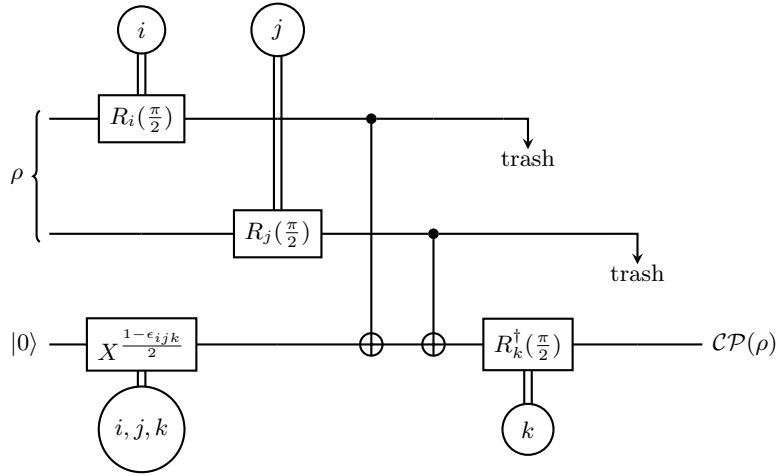

\subsection{1-to-2 qubit $\mbb{SU}(2)$ equivariant maps}\label{app:1to2}
If we solve for the nullspace of the set of equivariant 1-to-2 maps we see that all trace preserving maps must have the following action. Say our input is $\rho = \frac{1}{2}(\id+\vec{\sigma}\cdot \vec{r})$, then the output state of the map will be:
\begin{align}
    \frac{1}{4}\id\otimes\id + \frac{a}{2}(XX+YY+ZZ)+\frac{b}{2}\id\otimes(\vec{\sigma}\cdot \vec{r}) + \frac{c}{2} (\vec{\sigma}\cdot \vec{r})\otimes\id + \frac{d}{2}\vec{r}\cdot (YZ-ZY,ZX-XZ,XY-YX),
\end{align}
where a, b, c, and d are arbitrary real numbers (to be CP there will be additional constraints). Notice that this is effectively a linear combination of adjoints of the 2$\rightarrow$1 maps. The first being Tr, the second SWAP, the third and fourth partial traces, and the final one being the CP channel. This follows logically from representations of $\mbb{SU}(2)$ being self-dual. Recall that a linear map is equivariant if and only if its Choi operator satisfies $[(R^\text{in}(g))^*\otimes R^\text{out}(g), J^\phi]=0$. That is, $J^\phi$ lies in the commutant of the tensor product of the (dual) input and output representations. An equivariant $1\mapsto 2$ map then can be associated to the commutant of 
\begin{align}
   g^*\otimes g^{\otimes 2} \cong g^{\otimes 3},
\end{align}
which is exactly the same as the representation that commutes with the Choi operators of $2\mapsto 1$ $\mbb{SU}(2)$ equivariant channels.

\end{document}